\documentclass[letterpaper,onecolumn,superscriptaddress,unpublished,11pt]{quantumarticle}
\pdfoutput=1
\usepackage[utf8]{inputenc}
\usepackage[english]{babel}
\usepackage[T1]{fontenc}
\usepackage{amsmath}
\usepackage{amsthm}
\usepackage{amssymb}
\usepackage[numbers,sort&compress]{natbib}
\usepackage{nicefrac}
\usepackage{subcaption}
\usepackage{graphicx}
\graphicspath{{./circuits/}{.}{./app-circuits/}}
\usepackage{multirow}

\usepackage{letltxmacro}
\LetLtxMacro\amsproof\proof
\LetLtxMacro\amsendproof\endproof
\usepackage{thmtools}
\AtBeginDocument{%
  \LetLtxMacro\proof\amsproof
  \LetLtxMacro\endproof\amsendproof
}

\usepackage[usenames,dvipsnames]{xcolor}
\usepackage[
 unicode=true,bookmarks=true,breaklinks=false,pdfborder={0 0 1},backref=false,colorlinks=false,bookmarksdepth=3]
 {hyperref}
\hypersetup{colorlinks=true, linkcolor=Maroon, citecolor=OliveGreen, filecolor=magenta, urlcolor=Blue}
\usepackage{tikz}

\newcommand{\abs}[1]{\ensuremath{\left\vert#1\right\vert}}
\newcommand{\ket}[1]{\ensuremath{\vert#1\rangle}}
\newcommand{\bra}[1]{\ensuremath{\langle #1\vert}}

\newcommand{\kb}[2]{\ensuremath{\vert #1 \rangle \langle #2 \vert}}

\def\id{\mbox{\small 1} \!\! \mbox{1}}

\newtheorem{claim}{Claim}

\def\z{\mathbb{Z}}
\def\ve{\varepsilon}

\def\q{\mathbb{Q}}

\def\G{\mathcal{G}}

\def\PG#1{\mathcal{P}_{#1}} % Pauli group on {#1}-qubits
\DeclareMathOperator{\Stab}{Stab}
\DeclareMathOperator{\Spec}{Spec}

\def\ip#1{\mathopen{}\left\langle #1\right\rangle\mathclose{}}
\def\nrm#1{\mathopen{}\left\Vert #1\right\Vert\mathclose{}}

\def\ceil#1{\mathopen{}\left\lceil #1\right\rceil\mathclose{}}
\def\at#1{\mathopen{}\left(#1\right)\mathclose{}}
\def\atsm#1{\mathopen{}(#1)\mathclose{}}
\def\of#1{\mathopen{}\left[#1\right]\mathclose{}}
\def\set#1{\mathopen{}\left\{ #1 \right\}\mathclose{}}
\def\ket#1{\mathopen{}\left|#1\right\rangle\mathclose{}}
\def\ketsm#1{|#1\rangle}
\def\bra#1{\mathopen{}\left\langle #1\right|\mathclose{}}
\def\brasm#1{\mathopen{}\langle #1|\mathclose{}}
\def\abs#1{\mathopen{}\left|#1\right|\mathclose{}}

\newtheorem*{lem*}{Lemma}
\newtheorem*{thm*}{Theorem}
\newtheorem{thm}{Theorem}[section]
\newtheorem{lem}[thm]{Lemma}
\newtheorem{prop}[thm]{Proposition}
\newtheorem{alg}[thm]{Algorithm}

\newtheorem{dfn}[thm]{Definition}
\newtheorem{rem}[thm]{Remark}

\newcommand{\eq}[1]{\hyperref[eq:#1]{(\ref*{eq:#1})}}
\renewcommand{\sec}[1]{\hyperref[sec:#1]{Section~\ref*{sec:#1}}}
\newcommand{\secsm}[1]{\hyperref[sec:#1]{Sec.~\ref*{sec:#1}}}
\newcommand{\app}[1]{\hyperref[app:#1]{Appendix~\ref*{app:#1}}}
\newcommand{\theo}[1]{\hyperref[thm:#1]{Theorem~\ref*{thm:#1}}}
\newcommand{\algo}[1]{\hyperref[alg:#1]{Algorithm~\ref*{alg:#1}}}
\newcommand{\lemm}[1]{\hyperref[lem:#1]{Lemma~\ref*{lem:#1}}}
\newcommand{\defn}[1]{\hyperref[defn:#1]{Definition~\ref*{defn:#1}}}
\newcommand{\corr}[1]{\hyperref[cor:#1]{Corollary~\ref*{cor:#1}}}
\newcommand{\fig}[1]{\hyperref[fig:#1]{Figure~\ref*{fig:#1}}}
\newcommand{\figsm}[1]{\hyperref[fig:#1]{Fig.~\ref*{fig:#1}}}
\newcommand{\tab}[1]{\hyperref[tab:#1]{Table~\ref*{tab:#1}}}
\newcommand{\tabsm}[1]{\hyperref[tab:#1]{Tab.~\ref*{tab:#1}}}
\newcommand{\propos}[1]{\hyperref[prop:#1]{Proposition~\ref*{prop:#1}}}
\newcommand{\propsm}[1]{\hyperref[prop:#1]{Prop.~\ref*{prop:#1}}}
\newcommand{\rema}[1]{\hyperref[rem:#1]{Remark~\ref*{rem:#1}}}

\makeatletter
\newcommand{\xRightarrow}[2][]{\ext@arrow 0359\Rightarrowfill@{#1}{#2}}
\makeatother
\def\catalyse#1{\xRightarrow{#1}}

\begin{document}

\title{Lower bounds on the non-Clifford resources for quantum computations}
\date{\today}
% use alphabetical order for now; we will talk about it later
\author{Michael Beverland}
\affiliation{QuArc, Microsoft Quantum, Redmond, Washington, US}
%\homepage{http://todo.org}
%\orcid{0000-0003-0290-4698}
%\email{mibelerl@microsoft.com}
\author{Earl Campbell}
\affiliation{Department of Physics and Astronomy, University of Sheffield, Sheffield, UK}
%\email{todo@todo.org}
\author{Mark Howard}
\affiliation{School of Mathematics, Statistics \& Applied Mathematics, NUI Galway, Ireland}
%\email{mark.howard@nuigalway.ie}
\orcid{0000-0002-6910-185X}
\author{Vadym Kliuchnikov}
\affiliation{QuArc, Microsoft Quantum, Redmond, Washington, US}
%\email{vadym@microsoft.com}
\orcid{0000-0002-7076-5864}

\maketitle

\begin{abstract}
Treating stabilizer operations as free, we establish lower bounds on the number of resource states, also known as magic states, needed to perform various quantum computing tasks.
Our bounds apply to adaptive computations using measurements with an arbitrary number of stabilizer ancillas.
We consider (1) resource state conversion, (2) single-qubit unitary synthesis, and (3) computational subroutines including the quantum adder and the multiply-controlled $Z$ gate. 

To prove our resource conversion bounds we introduce two new monotones, the stabilizer nullity and the dyadic monotone, and make use of the already-known stabilizer extent.
We consider conversions that borrow resource states, known as catalyst states, and return them at the end of the algorithm. 
We show that catalysis is necessary for many conversions and introduce new catalytic conversions, some of which are optimal.

By finding a canonical form for post-selected stabilizer computations, we show that approximating a single-qubit unitary to within diamond-norm precision $\varepsilon$ requires at least $1/7\cdot\log_2(1/\varepsilon) - 4/3$ $T$-states on average.
This is the first lower bound that applies to synthesis protocols using fall-back, mixing techniques, and where the number of ancillas used can depend on $\varepsilon$. 

Up to multiplicative factors, we optimally lower bound the number of $T$ or $CCZ$ states needed to implement the ubiquitous modular adder and multiply-controlled-$Z$ operations. 
When the probability of Pauli measurement outcomes is 1/2, some of our bounds become tight to within a small additive constant.
\end{abstract}

\setcounter{tocdepth}{2} % do not include subsubsections in toc
\newpage
\renewcommand{\baselinestretch}{0.95}\normalsize
\tableofcontents
\renewcommand{\baselinestretch}{1.0}\normalsize
\newpage
\section{Introduction and high-level overview}

Many promising architectures for universal fault-tolerant quantum computing~\cite{karzig17,fowler2012} perform computation by applying \textit{stabilizer operations} to carefully prepared \textit{resource states} known as magic states \cite{knill2005quantum,bravyi2005universal,campbell2017roads}.
The stabilizer operations, which consist of Clifford gates, preparation of stabilizer ancilla states, and measurements in the Pauli basis, tend to be relatively easy to implement in these architectures.
On the other hand, due to restrictions imposed by error correction \cite{Eastin2009,Bravyi2013,Pastawski2015,Beverland2016}, the resource states tend to be produced by hefty distillation protocols \cite{bravyi2012,jones2013multilevel,Haah2017magicstate,Campbell2018magicstateparity}, that dominate the space-time overhead of the overall computation.
It is therefore very natural and practically motivated to ask: 

\textit{What is the minimum number of copies of a particular resource state that must be consumed to perform a given computational task using an arbitrary number of stabilizer operations?}

We address this question by providing lower bounds for a number of computational tasks.

In the early days of quantum computing research, upper bounds for the resources required to implement compelling algorithms were crucial to motivate the development of scalable quantum computing hardware.
Today, lower bounds are arguably more important since they can identify opportunities for further optimization.
%While upper bounds are critical when ascertaining whether some action (whether it is running a particular algorithm or building a quantum computer in the first place) is worth implementing, once a decision has been made to carry out the action, lower bounds are key.
%It is lower bounds that then tell us how much the action can potentially be improved or whether it is already optimal. 
Unfortunately, lower bounds are infamously elusive: where an upper bound of resource requirements can be obtained by identifying an explicit algorithm, proving that no algorithm exists with certain properties can be very difficult.
% In previously studied (but arguably less practically motivated) models of quantum computation which took single qubit gates as free and two qubit gates as expensive, tight (up to a multiplicative factor) lower bounds to implement an arbitrary unitary could be obtained \cite{shende2004smaller,shende2004minimal}. \magnote{modify the last sentence in this paragraph.}

Good lower bounds have however been forthcoming for some models of quantum computation, such as the two-qubit gate cost in the absence of measurement considering single-qubit gates as free \cite{shende2004smaller,shende2004minimal,Shende2006,Iten2016,Knill1995}. 
For example, the multiply controlled phase operation needs at least a linear number of two-qubit gates \cite{Barenco1995}.
% % Lemma 7.7 on https://arxiv.org/pdf/quant-ph/9503016.pdf#page=22
% There are also well-developed approaches that 
% match the worst-case lower-bound
% up-to a multiplicative constant factor for
% implementing arbitrary unitaries and isometries
% \cite{shende2004smaller,shende2004minimal,Shende2006,Iten2016,Knill1995}.
These results are not tailored to the fault-tolerant setting where two-qubit CNOT gates and measurements are much cheaper than single-qubit non-Clifford gates, leading us to seek new theoretical tools.

We provide lower bounds for the production of particular target states, the implementation of important subroutines such as the adder and the multiply controlled phase operation, as well as approximating an arbitrary unitary to a desired precision.
Our bounds significantly strengthen the best that were previously known and in some cases are the first non-trivial bounds that apply.
We give separate bounds in terms of a variety of the most common basic resource states, and map out how these basic states themselves can be converted into one another.

\subsection{Monotones under stabilizer operations}

In \sec{basic-techniques}, we introduce the \textit{stabilizer nullity}, a function $\nu(\ket{\psi})$ of any pure state $\ket{\psi}$ that is non-increasing under stabilizer operations. 
The stabilizer nullity is surprisingly powerful given its simplicity: it is the number of qubits that $\ket{\psi}$ is hosted in, minus the number of independent Pauli operators that stabilize $\ket{\psi}$.
It is easy to see that $\nu = 0$ for any stabilizer state.
We also leverage a previously known monotone called the \textit{stabilizer extent} \cite{regula2017convex,bravyi2018simulation} (see \sec{stabiziler-extent}) which is also non-increasing under stabilizer operations.
These monotones allow us to bound some state preparation tasks, for example $n$ copies of the state $\ket{\psi}$ cannot be sufficient to produce a target state $\ket{\text{tar.}}$ if $\nu(\ket{\psi}^{\otimes n})< \nu(\ket{\text{tar.}})$.
We write this as 
$$ \nu(\ket{\psi}^{\otimes n})< \nu(\ket{\text{tar.}}) ~~~ \text{implies} ~~~ \ket{\psi}^{\otimes n} \not\longrightarrow \ket{\text{tar.}}. $$
An important factor in understanding the limitations of stabilizer operations is that their power can be increased not only by consuming resource states, but also by borrowing other resource states, known as \textit{catalyst states} and returning them unchanged at the end of the algorithm, that is
\begin{align*}
    \ket{A} \not\longrightarrow \ket{B}\text{ but }   \ket{A}\ket{\text{cat.}} \longrightarrow \ket{B}\ket{\text{cat.}}, \text{ written as } \ket{A}\catalyse{\ket{\text{cat.}}} \ket{B}.
\end{align*}

In \sec{catalysis}, after establishing a general no-go theorem, we give several examples of resource state conversion where catalysis is necessary and sufficient. Similar results have been proven in the past for a more restricted set of scenarios e.g.~\cite{Campbell2010,Selinger2012,aqft2018}.

\subsection{Resource state conversion}

Standard choices of which basic resource states are consumed by an algorithm could include the $T$-state $\ket{T}:=T \ket{+}$, the $CCZ$-state $\ket{CCZ}:= CCZ\ket{+}^{\otimes 3}$, the $\sqrt{T}$-state $\ketsm{\sqrt{T}}:= \sqrt{T}\ket{+}$, along with many others.
Before we turn to resource lower bounds for computational tasks such as implementing arbitrary unitaries, we should first consider how various basic resource states relate to one another, which is the focus of \sec{Conversions}. 
Thankfully, the stabilizer nullity is \textit{additive}, such that $\nu(\ket{\psi} \ket{\phi}) = \nu(\ket{\psi}) + \nu(\ket{\phi})$ for all $\ket{\psi}$ and $\ket{\phi}$, which allows us to say even more. 
For example, we can rule out catalyzed conversions since $\nu(\ket{A})< \nu(\ket{B})$ implies that $\nu(\ket{A}  \ket{\text{cat.}})< \nu(\ket{B}  \ket{\text{cat.}})$ for any catalyzing state $\ket{\text{cat.}}$.
Moreover, tensor powers of states simplify, allowing us to make asymptotic implications, i.e.,
$$ \nu(\ket{A})< r \cdot \nu(\ket{B}) ~~~ \text{implies} ~~~ \ket{A}^{\otimes n} \not\Longrightarrow \ket{B}^{\otimes \ceil{r n}} ~~~\forall n.$$
Here, the double arrow indicates that \textit{even with an arbitrary catalyst state} $\ceil{r n}$ copies of $\ket{B}$ cannot be produced from $n$ copies of $\ket{A}$ using stabilizer operations.

These state conversion bounds and algorithms put our computational task lower bounds on more solid footing by allowing us to analyze the cost in terms of different input resource states.
We also foresee our conversion results being useful in a much broader context, such as allowing a meaningful comparison of protocols that distill $T$-states with protocols that distill $CCZ$-states.
For example, two $T$-states can be produced from a single $CCZ$-state, and therefore a distillation protocol $A$ for $\ket{CCZ}$ outperforms a distillation protocol $B$ for $\ket{T}$ if the protocol formed from converting the output of $A$ into $T$-states outperforms $B$.

The bounds we obtain show there is a conversion gap: starting with $n$ copies of $\ket{T}$, and applying the best possible $\ket{T}$ to $\ket{CCZ}$ conversion followed by the best possible $\ket{CCZ}$ to $\ket{T}$ conversion will yield fewer than $n$ copies of $\ket{T}$.
This gap survives in the asymptotic limit.

In \tab{T_to_Targ} and \tab{CCZ_to_Targ} at the end of \sec{Conversions} we summarize many of our conversion bounds along with the most efficient known conversion algorithms (some of which were previously known, some of which we introduce in \sec{polynomialcatalysis} and \sec{adder-conversion-protocols}). 
Many of the conversion algorithms do not match the bounds suggesting that more efficient algorithms remain to be found.

\subsection{Computational tasks}
\label{sec:intro-computational-tasks}

We have outlined how monotones can help bound the resources required to produce a particular target state.
In \sec{computational-tasks} we bootstrap these techniques to lower bound the resources required to perform certain computational tasks.

To exactly implement a unitary $U$,
note that a lower bound of the number of copies of  $\ket{\psi}$ needed to produce a state $U\ket{S}$, where $\ket{S}$ is a stabilizer state, also serves as a lower bound for applying $U$.
We use this strategy in \sec{CnZ-bound} to study the multiply controlled $Z$ gate $C^n Z$, which is a key component of many important algorithms, including part of the reflection step in Grover's search \cite{Grover1996}.
By calculating the stabilizer nullity of the state $\ket{C^{n-1} Z} = C^{n-1} Z \ket{+}^{\otimes n}$ we straightforwardly show that for $n\geq 3$ it is not possible to apply the multiply controlled $Z$ gate $C^{n-1} Z$ with stabilizer operations consuming fewer than $n$ copies of $\ket T$, or $n/2$ copies of $\ket{CS}$, or $n/3$ copies of $\ket{CCZ}$.
For comparison, the most efficient known algorithm \cite{Jones2012} produces $C^{n-1} Z$ with $n-2$ copies of $\ket{CCZ}$.

It is also useful to consider catalysis as a proof technique when establishing lower bounds for computational tasks.
For example, suppose $U$ maps a state $\ket{S}\ket{\Psi}$ to a state 
$\ket{\Phi}\ket{\Psi}$ for stabilizer state $\ket{S}$ and non-stabilizer states $\ket{\Phi}$ and $\ket{\Psi}$.
Then, a resource lower bound for catalytically producing the state $\ket{\Phi}$ must
also serve as a lower bound for implementing $U$.
In \sec{adder-bound}, we use this strategy to lower bound one of the most fundamental quantum arithmetic operations: the adder circuit, which acts on $n$-qubit basis states as $A(\ket{i}\ket{j}) = \ket{i}\ket{i + j}$ with $i+j$ evaluated modulo $2^n$. 
The key is that the modular adder circuit acts on the input $\ket{+}^{\otimes n}\ket{QFT_n}$ to produce $\ket{QFT^*_n}\ket{QFT_n}$, where $\ket{QFT_n}$ is sometimes known as the quantum Fourier state, which becomes $\ket{QFT^*_n}$ under complex conjugation of coefficients in the computational basis.
Crucially, we find that $\nu(\ket{QFT_n}) = \nu(\ket{QFT^*_n})= n-2$ implying the adder circuit cannot be implemented with fewer than $n-2$ copies of $\ket{T}$, or $(n-2)/2$ copies of $\ket{CS}$, or $(n-2)/3$ copies of $\ket{CCZ}$.
The most efficient known implementation of a modular adder uses $n-1$ copies of $\ket{CCZ}$ state \cite{Gidney2018}.\footnote{
After the first posting of this paper, Craig Gidney \cite{GidneyBlog2019} showed that the state $\ket{C^{n} Z}$ can be produced using the $n$-qubit modular adder. 
We reproduce his argument in \app{Gidney-bound-reduction} for completeness. 
Using our (slightly stronger) bounds for $\ket{C^{n-1} Z}$ the adder circuit cannot be implemented with fewer than $n+1$ copies of $\ket{T}$, or $(n+1)/2$ copies of $\ket{CS}$, or $(n+1)/3$ copies of $\ket{CCZ}$.}

\begin{figure}
    \centering
    \includegraphics{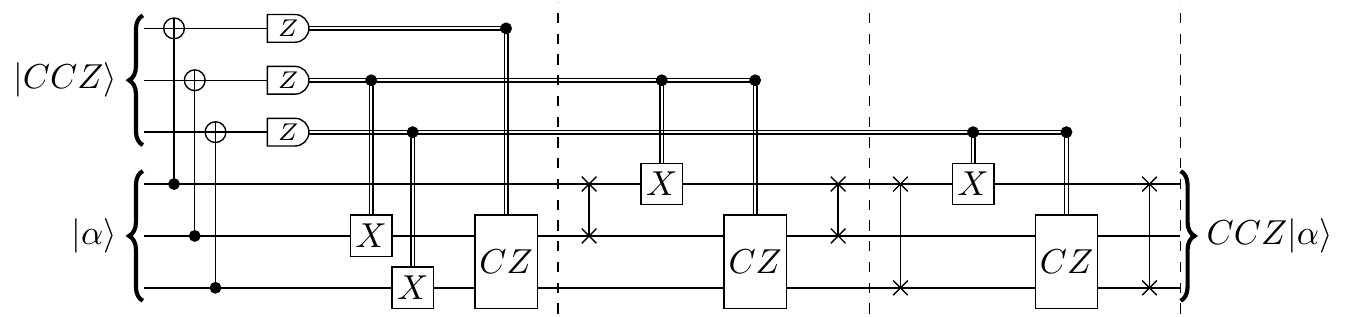}
    \caption{Implementation of the $CCZ$ gate via injection of the $\ket{CCZ}$ state.}
    \label{fig:ccz-injection}
\end{figure}

There are diagonal unitaries for which the cost of the resource state $U\ket{+}^{\otimes n}$ is the same as the cost of implementing the unitary itself.
In particular, this is the case for all diagonal unitaries from the third level of the Clifford hierarchy, as can be seen from the state injection protocol described in \app{diagonal-gate-injection}.

\subsection{Unitary synthesis}

In \sec{unitary-synthesis} we consider the number of resource states required to approximate an arbitrary single-qubit unitary to within diamond-norm precision $\varepsilon$ using stabilizer operations.
Although the importance of unitary synthesis has been long recognized, less is known regarding synthesis strategies exploiting measurements, classical feed-forward and ancilla qubits. Crucially our lower bounds apply to synthesis algorithms in this general setting. 

Loosely, our proof strategy is to select a target unitary $U$ that can just be resolved from the identity given the required precision $\varepsilon$, and lower bound the resources required to produce the state $U\ket{0}$.
As $\varepsilon$ becomes small, $U\ket{0}$ approaches (but never quite reaches) the basis state $\ket{0}$.
Unfortunately, the associated resource requirement divergence is not captured by either of the monotones we have discussed as they do not diverge for states approaching $\ket{0}$.  

To achieve the required lower bound, we find a canonical form for stabilizer circuits applied to resource states (which may be of independent interest  - see \theo{canonical-form}), and turn to a number theoretic approach. 
First, note that if the state $U\ket{0}$ is measured, if $\varepsilon$ is small then the probability $|\langle 1 |U\ket{0}|^2$ must be finite but very close to zero.
Second, we use the canonical form to show that any single-qubit state produced using stabilizer circuits on a fixed number of resource states can only have a discrete set of allowed measurement probabilities, irrespective of the length of the stabilizer circuit and the number of stabilizer ancillas.
This establishes a bound since producing a state $U\ket{0}$ with sufficiently small $|\langle 1 |U\ket{0}|^2$ requires a sufficiently large number of resource states.

Our unitary synthesis results do not hold when a catalyst state is allowed, in contrast to those bounds proven with the stabilizer nullity due to its additive property.
Another complication is that the number of resource states consumed by a protocol is actually a random variable, which can depend on the sequence of measurement outcomes obtained during the protocol.
Our previous bounds held for every possible sequence of measurement outcomes, but for unitary synthesis we have to account for this subtlety.
Let $\mathcal{N}_{\ket{\Psi}}(U,\varepsilon)$ be the number of copies of $\ket{T}$ consumed by a stabilizer circuit that approximates a unitary $U$ to within diamond-norm precision $\varepsilon$.
We show that there exist (diagonal) target unitaries $U$ such that the expectation of $\mathcal{N}_{\ket{\Psi}}(U,\varepsilon)$ satisfies
\[
\mathbf{E} \mathcal{N}_{\ket{T}}(U,\varepsilon) 
\ge
\frac{1}{7}\log_2\at{1/\varepsilon} - \frac{4}{3} .
\] 
The best existing algorithm requires at most $2\log_2\at{1/\varepsilon} + O(\log(\log 1/\varepsilon))$ $T$-gates to implement any diagonal unitary, which comes from a combination of Refs~\cite{CampbellRandom17, hastings2016turning} with \cite{RossSelinger}.
Our bound above is inferred from \theo{approximation-lower-bound}, which is a stronger but more nuanced result.
We also have results that apply to stabilizer circuits with post selection in \sec{unitary-and-state-approx} and \sec{approx-with-CS}.

\subsection{Measurement with probability 1/2}

We can further tighten some lower bounds in a common restricted setting \cite{Gottesman1999, Jones2012, Gidney2018, gidney2018efficient} where arbitrary single-qubit Pauli measurements are not permitted, but only those measurements with outcomes which occur with probability one half.
In \sec{prob-half-bounds}, we introduce a quantity similar to the stabilizer nullity, which we call the \textit{dyadic monotone} $\mu_2(\ket{\psi})$ to prove a number of known subroutines exhibit nearly optimal resource consumption in this setting, narrowing the search for future algorithm improvements. 
The dyadic monotone requires that we restrict to states (including catalyst states) which can be written in the computational basis with coefficients that are integer combinations of $\exp(i\pi j/2^d) / 2^k$ for integers $ d, j, k$. This includes stabilizer states as well as resource states associated with all higher levels of the Clifford hierarchy~\cite{Gottesman1999}.
Among other results we show that in this setting, the well-known circuit \cite{Jones2012} which implements the multiply-controlled-Z gate $C^{n-1} Z$ using $n-2$ copies of $\ket{CCZ}$ is optimal. 
In addition, we show that $n-2$ copies of $\ket{CCZ}$ are required to implement modular adder.
The best known modular adder circuit~\cite{Gidney2018} uses just one more $\ket{CCZ}$ gate.
\footnote{
This bound becomes tight in light of Craig Gidney's blog post \cite{GidneyBlog2019} which showed that the $n$-qubit modular adder can be used to produce the state $\ket{C^{n} Z}$, which requires at least $n-1$ copies of $\ket{CCZ}$ in this setting.}

\section{Some basic techniques}
\label{sec:basic-techniques}

In this section we present some general techniques that are used throughout the paper, and defer our more specialized techniques to later sections and appendices. 
In \sec{state-properties} we introduce a number of properties of quantum states, including a simple but surprisingly powerful monotone under stabilizer operations which we call the stabilizer nullity.
In \sec{stabiziler-extent} we review another monotone under stabilizer operations known as the stabilizer extent which was recently introduced \cite{regula2017convex,bravyi2018simulation}.
In \sec{catalysis} we show that the number of resources required to accomplish a computational goal can depend upon whether or not an additional catalyzing resource state is allowed which is returned unchanged at the end of an algorithm.

\subsection{Stabilizer nullity}
\label{sec:state-properties}

Let us first recall the definition of a stabilizer state and introduce a slight generalization of it.
\begin{dfn} \label{defn:state-stabilizer}
Let $\ket\psi$ be a non-zero $n$-qubit state. 
\emph{The stabilizer} of $\ket\psi$, denoted $\Stab\ket\psi$,
is the sub-group of the Pauli group $\PG{n}$ on $n$ qubits for which $\ket \psi$ is a $+1$ eigenstate,
that is $\Stab{\ket\psi} = \set{P\in \PG{n}:P\ket{\psi} = \ket{\psi}}$.
The states for which the size of the stabilizer is $2^n$ are called \emph{stabilizer states}.
States for which the stabilizer contains only the identity matrix are said to have a \emph{trivial stabilizer}. If Pauli $P$ is in $\Stab{\ket\psi}$, we say that $P$ stabilizes $\ket\psi$.
\end{dfn}
% Remark: Above definition is common in mathematical literature for a general group acting on some objects. 
% Here we have a specific matrix group acting on vectors
% See also: https://proofwiki.org/wiki/Definition:Stabilizer and https://groupprops.subwiki.org/wiki/Group_action#Stabilizer.
Note that $\Stab\at{\ket\psi}$ can not contain $-I$.
In addition, note that all Pauli group elements contained in $\Stab\at{\ket\psi}$ commute with each other and are Hermitian matrices.
The size of the stabilizer of any state is equal to some power of two.
For any Clifford unitary $C$, the size of $\Stab\ket\psi$ is always equal to the size of $\Stab( C\ket\psi )$.
The size of the stabilizer is also multiplicative for the tensor product of states,
that is $\abs{\Stab\at{\ket\psi\ket\phi}} = \abs{\Stab{\ket\psi}}\cdot\abs{\Stab{\ket\phi}}$. A key quantity that we use throughout the paper is simply related to  $\Stab\at{\ket\psi}$.
% Should above facts be left as an exercise for the reader, or should we provide proof of them in Appendix ?
\begin{dfn}[Stabilizer nullity] \label{defn:stabilizer-monotone}
	Let $\ket\psi$ be a non-zero $n$-qubit state. 
	\emph{\bf The stabilizer nullity} of $\ket\psi$
	is $\nu(\ket\psi) = n - \log_2|\Stab\ket\psi|$.
\end{dfn}
Let us next see that the stabilizer nullity is non-increasing when multiple-qubit Pauli measurements are applied.
\begin{prop} \label{prp:measurement-and-stabilizer}
Let $\ket\psi$ be a non-zero $n$-qubit state and let $P$ be an $n$-qubit Pauli matrix
and suppose that the probability of a $+1$ outcome when measuring $P$ on $\ket\psi$ is non-zero.
Then there are two alternatives for the state $\ket\phi$ after the measurement: 
either $\abs{\Stab\ket\phi} = \abs{\Stab\ket\psi}$,
or $\abs{\Stab\ket\phi} \ge 2\abs{\Stab\ket\psi}$, both of which satisfy $\nu(\ket\phi) \leq \nu(\ket\psi)$.
\end{prop}
% Remark: probability zero is equivalent to $-P$ being in $\Stab\ket\psi$
\begin{proof}
First consider the simple case when $P$ is in $\Stab\ket\psi$.
In this case, the ``$+1$'' measurement outcome occurs with probability $1$ and $\ket\psi$ is unchanged.
When $P$ is not in $\Stab\ket\psi$ we consider two alternatives.
The first alternative is that $P$ commutes with all elements of $\Stab\ket\psi$,
then $\Stab\ket\phi$ contains $\Stab\ket\psi \cup P\Stab\ket\psi$ and its size is at least double
that of $\Stab\ket\psi$.
The second alternative is that $P$ anti-commutes with some element $Q$ from $\Stab\ket\psi$.
In this case, we will see that the size of the stabilizer does not change as a result of the measurement. 
First note that in this case the probability of the +1 measurement outcome is $1/2$, because the probability of the $+1$ outcome is $\bra\psi(I+P)\ket\psi/2$ and equal to 
$\bra\psi Q (I+P) Q\ket\psi/2 = \bra\psi (I-P)\ket\psi/2$ which is the probability of the $-1$ outcome, where we have used $Q\ket \psi = \ket\psi$ and $QPQ=-P$.
Therefore $\ket\phi = (I + P)/\sqrt{2} \ket\psi$, where we have fixed the normalization such that $\bra \phi \phi \rangle = \bra \psi \psi \rangle$. 
Using that $Q$ stabilizes $\ket\psi$ we also see that, $\ket\phi = (I + PQ)/\sqrt{2} \ket\psi$.
Finally we observe that $(I + PQ)/\sqrt{2}$ is a Clifford unitary equal to $\exp\at{i\pi P'/4}$ for
Hermitian Pauli matrix $P' = -i PQ$. 
As $\ket \phi$ and $\ket \psi$ differ by a Clifford, we conclude that $\Stab \ket\psi$ and $\Stab \ket\phi$ are the same size.
\end{proof}
%Remark 1: Above we relied on a fact: $\exp( i\phi P) = I \cos(\phi) + i \sin(\phi) P$
%Remark 2: To see that $\exp\at{i\pi P/4}$ is Clifford note the $exp(i \pi/4 Z )$ is $S$ gate and
%you can turn exp(i \pi/4 Z ) into \exp\at{i\pi P/4} by conjugating it with CNOTs and single qubit Pauli matrices.
%Remark 3: Applying \exp\at{i\pi P/4} to a state either leaves states' stabilizer unchanged or changes half of the stabilizers
%Remark 4: Above observations lead to an improved Stabilizer state simulation by simplifying the measurement procedure and cutting the required memory in half.
One might wonder, if the second alternative in the proposition statement above should be 
$\abs{\Stab\ket\phi} = 2\abs{\Stab\ket\psi}$ instead of 
$\abs{\Stab\ket\phi} \ge 2\abs{\Stab\ket\psi}$.
Here is an example that shows that the size of the stabilizer can more than double after one measurement.
Consider an initial state $\ket\psi= \ket{T}\ket{T}$ states and measure the $+1$ outcome of $-Z\otimes Z$.
The resource state $\ket{T}$ has a trivial stabilizer and so do its tensor powers, such that $|\Stab \ket\psi|=1$.
The result of the measurement is $\ket\phi =\at{\ketsm{01}+\ketsm{10}}/\sqrt{2}$ which is a two-qubit stabilizer state which therefore has $|\Stab \ket\phi|=4$.

By Prop.~\ref{prp:measurement-and-stabilizer} and the other aforementioned properties of $\Stab\ket\psi$, we see that the stabilizer nullity $\nu$ is invariant under Clifford unitaries, is non-increasing under Pauli measurements, and is additive under the tensor product.
Moreover, as $\nu(\ket\psi) = 0$ when $\ket{\psi}$ is a stabilizer state, the stabilizer nullity is invariant under the inclusion or removal of stabilizer states.
Another useful definition is the Pauli spectrum.

\begin{dfn}[Pauli spectrum] \label{defn:Pauli-spectrum}
	Let $\ket\psi$ be a non-zero $n$-qubit state. 
	\emph{The Pauli spectrum} $\Spec\ket\psi$ of $\ket\psi$ is:
	\begin{eqnarray}
	\Spec{\ket\psi} = \left\{ \frac{|\bra{\psi} P \ket{\psi}|}{\bra{\psi}\psi \rangle},~~~ \forall ~P \in  \{ I,X,Y,Z \}^{\otimes n} \right\}.
	\end{eqnarray}
\end{dfn}

The Pauli spectrum is a list of $4^n$ real numbers each between $0$ and $1$ which is invariant under Clifford gates.
Consider the following example

\begin{prop} \label{prp:RotationStatePauliSpectrum}
	The Pauli spectrum of the state $\ket{\theta}=(\ket{0}+e^{i\theta}\ket{1})/\sqrt{2}$ is $\{\cos{\theta},\sin{\theta},0 \}$.
	The state $\ket{\theta}$ is therefore a stabilizer state only for $\theta = m ~\pi/2$ for some integer $m$.
\end{prop}

\begin{proof}
	This follows from direct verification of the Pauli spectrum of $(\ket{0}+e^{i\theta}\ket{1})/\sqrt{2}$.
\end{proof}

We will make further use of the Pauli spectrum later in the paper, but for now note that the number of 1s in the Pauli spectrum of $\ket\psi$ is $|\Stab\ket{\psi}|$.

\subsection{Stabilizer extent}
\label{sec:stabiziler-extent}

In \defn{stabilizer-monotone} we introduced the stabilizer nullity $\nu$ which is additive and monotonic under stabilizer operations.
Another monotone under stabilizer operations known as the \emph{stabilizer extent} $\xi$ was recently introduced \cite{regula2017convex,bravyi2018simulation}. 
The stabilizer extent has a number of other very desirable properties shown in \cite{regula2017convex}.

\begin{dfn}[stabilizer extent]
\label{defn:stabilizer-extent}
For an arbitrary pure state $\ket{\psi}$, 
\emph{\bf the stabilizer extent}, denoted $\xi(\ket{\psi})$, is
\begin{align}
\xi(\ket{\psi})=\min||(c_1,\ldots,c_k)||_1^2 \text{ s.t. } \ket{\psi}=\sum_{\alpha=1}^k c_\alpha \ket{\phi_\alpha}.
\end{align}
where the minimization is over all complex linear combinations of stabilizer states $\{\ket{\phi_\alpha}\}$.
\end{dfn}

It is clearly submultiplicative but for many interesting cases it has been proven to be strictly multiplicative. More precisely,
% assume that each  $\psi_j$ of $\{\psi_1,\psi_2,\ldots,\psi_\ell\}$ satisfies either: (i) $\psi_j$ is a 1,2 or 3 qubit state; or (ii) there exists a stabilizer state $\phi_{stab}$ such that $|\bra{\psi_j} \phi_{stab} \rangle|^2 \geq 1/4$.  Then it follows that
%\begin{align}
%	\xi(\vert \psi_1 \otimes \psi_2  \otimes \ldots  \otimes \psi_\ell \rangle)=\prod_{j=1}^\ell \xi(\ket{\psi_j}) \quad\text{if }\forall j\text{ either}\quad \begin{cases}
%	\psi_j \text{ is an }n\text{-qubit state } (1\leq n \leq 3)\\
%	 \text{or}\\
%	 \ket{\psi_j}\text{ is a Clifford magic state \textbf{and}}\\
%	\exists \ket{\phi_\alpha} \text{ s.t. } |\bra{\psi_j} \phi_\alpha \rangle|^2 \geq 1/4\\
%		 \text{or}\\
%	 \exists \ket{\omega_j} \text{ s.t. } \xi(\ket{\psi_j})=\frac{|\langle \psi_j \vert \omega_j \rangle|^2}{\max_{\phi_\alpha} |\bra{\omega_j} \phi_\alpha \rangle|^2  } \textbf{ and}\\
%	\exists \ket{\phi_\alpha} \text{ s.t. } |\bra{\omega_j} \phi_\alpha \rangle|^2 \geq 1/4
%	\end{cases} \label{eqn:MultCondition}
%\end{align}

\begin{lem}
\label{lem:stabilizer-extent-mult}
The stabilizer extent is multiplicative with respect to a given set of states $\{ \ket{\psi_1}, \ket{\psi_2},\dots \ket{\psi_\ell} \}$, such that, $\xi(\ket{\psi_1} \ket{\psi_2}  \ldots  \ket{\psi_\ell})=\prod_{j=1}^\ell \xi(\ket{\psi_j})$, if for each state $\ket{\psi_j}$ at least one of the following conditions is satisfied ($\ket{\phi_j}$ is always a stabilizer state):
\begin{enumerate}
    \item $\ket{\psi_j}$ is a state of at most three qubits,
%    \item $\ket{\psi_j}$ is a Clifford magic state \textbf{and} there exists a state $\ket{\phi_j}$ such that $| \bra{\psi_j} \phi_j \rangle|^2 \geq 1/4$,
	\item There exist states $\ket{\omega_j}$ and $\ket{\phi_j}$ such that $\xi(\ket{\psi_j})=\frac{|\langle \psi_j \vert \omega_j \rangle|^2}{\max_{\phi_j} |\bra{\omega_j} \phi_j \rangle|^2  }$  \textbf{and} $ |\bra{\omega_j} \phi_j \rangle|^2 \geq 1/4$.
\end{enumerate}
\end{lem}

This may actually hold more generally as we do not know of any counterexamples to multiplicativity of the stabilizer extent.

\subsection{Catalysis} 
\label{sec:catalysis} 

When considering the action of a sequence of stabilizer operations, it is important to consider scenarios in which another resource state is present which is returned at the end of the sequence.
As the additional resource state is not consumed, we refer to it as a \textit{catalyst}.
By considering restrictions of the entries of density matrices, we prove here that a broad class of resource state conversions are impossible without catalysis, but can be achieved with catalysis.
Similar results have been proven in the past for a more restricted set of scenarios \cite{Campbell2010,Selinger2012,aqft2018}.

It will be useful to recall some standard number fields: 
\begin{eqnarray*}
\q\at{i} &=& \set{ a_0 + i a_1 : a_0, a_1 \in \q },\\
\q\at{\zeta_8} &=& \set{a_0 + \zeta_8 a_1 + \ldots + \zeta^3_8 a_3 : a_k \in \q}, \zeta_8 = \exp\at{2\pi i / 8},\\
\q\at{\zeta_{16}} &=& \set{a_0 + \zeta_{16} a_1 + \ldots + \zeta^7_{16} a_7 : a_k \in \q}, \zeta_{16} = \exp\at{2\pi i / 16},
\end{eqnarray*}
where $\q$ is the set of rational numbers. 
As fields, each of these sets is closed under addition, multiplication, negation and taking the inverse. 
Note that $\q\at{i}$ is a subset of both $\q\at{\zeta_8}$ and $\q\at{\zeta_{16}}$.
It is straightforward to verify that in the computational basis $\ket{CS}$ and $\ket{CCZ}$ states have density matrices with all entries in $\q\at{i}$, but that the density matrix for the $\ket{T}$ state has some entries outside of $\q\at{i}$ -- instead all its entries are in $\q\at{\zeta_8}$.
More generally, the density matrix for 
$\ket{T}^{\otimes k}\otimes \ket{0}^{n-k}$ is given by: 
\[
\frac{1}{2^k}
\sum_{a,b \in \set{0,1}^k}
  \zeta_8^{\mathrm{weight}\at{a}-\mathrm{weight}\at{b}} \ket{a}\bra{b} \otimes \at{\ket{0}\bra{0}}^{\otimes n-k}.
\]
where $\mathrm{weight}\at{a}$ is the Hamming weight of the bit string $a$.
Similarly, we observe that the density matrix of $\ket{\sqrt{T}}^{\otimes k}\otimes \ket{0}^{n-k}$ is given by: 
\[
\frac{1}{2^k}
\sum_{a,b \in \set{0,1}^k}
  \zeta_{16}^{\mathrm{weight}\at{a}-\mathrm{weight}\at{b}} \ket{a}\bra{b} \otimes \at{\ket{0}\bra{0}}^{\otimes n-k}.
\]

Consider the following theorem, which rules out a number of uncatalysed resource state conversions.

\begin{thm}\label{thm:number-field-constraint} Let $F$ be a number field which contains $\q\at{i}$ and which is closed under complex conjugation. 
%For example, $F$ can be $\q\at{i}$, $\q\at{\zeta_8}$ or $\q\at{\zeta_{16}}$. 
Any stabilizer circuit applied to a density matrix with all entries in $F$ produces a density matrix with all entries in $F$, with both density matrices written in the computational basis. 
\end{thm}

For example, no stabilizer circuit on any number of $\ket{CS}$ or $\ket{CCZ}$ states (which have density matrices with all entries in $\q\at{i}$) can be used to produce a $\ket{T}$ state (which has a density matrix with all entries in $\q\at{\zeta_8}$). Similarly, no stabilizer circuit on any number of $\ket{T}$ states can be used to produce a $\ketsm{\sqrt{T}}$ state (with entries in $\q\at{\zeta_{16}}$).

\begin{proof} 
Suppose our stabilizer circuit acts upon $n$ qubits initially in the $\ket{0}$ state. 
Clearly the density matrix $\rho_{\text{initial}} = \at{\ket{0}\bra{0}}^{\otimes n}$ has entries over $\q$.
We point out that all Clifford unitaries can be written as matrices with entries over $\q\at{i}$, and therefore as matrices with entries over $F$. 
Explicitly, the Clifford group is generated by $H$, $CZ$ and $S$
\begin{eqnarray*}
H & = & \frac{1}{1+i}\begin{bmatrix}
    1       & 1 \\
    1       & -1
\end{bmatrix}, \\
S & : & \ket{0} \mapsto \ket{0},\, \ket{1} \mapsto i\ket{1}, \\
CZ & : & \ket{ab} \mapsto \at{-1}^{a \wedge b} \ket{ab}.
\end{eqnarray*}
Given that any gate $U$ in the circuit is a tensor product of a unitary with entries over $F$ and $I$ and $\rho$ has entries over $F$ the product $U \rho U ^\dagger $ is a density matrix with entries over $F$. 
Therefore applying the gates in the circuit preserves required property. 

Note that measurement with or without post-selection can be described as: 
\begin{eqnarray*}
\rho & \mapsto & \frac{P \rho P}{\mathrm{Tr}{\rho P}}, \\
\rho & \mapsto & \sum_{ P \in \mathcal{P} } P \rho P.
\end{eqnarray*}
The projectors $P$ above correspond to measurement in the computational basis and therefore can be written as matrices with entries over $\q\at{i}$ and therefore over $F$. 
The product of matrices over $F$ is a matrix over $F$. 
The trace of a matrix over $F$ is also in $F$ by the definition of a field. 
The quotient of a matrix over $F$ and an element of $F$ is again a matrix over $F$ because any field is closed under the division operation. 
This completes the proof. 
\end{proof}

Importantly, the no-go results of \theo{number-field-constraint} can be evaded by including a catalyst state. For example in \fig{CStoT} we show how a $\ket{CS}$ state can be used to produce a $\ket{T}$ state by using an additional $\ket{T}$ state which is not consumed.
Some examples of catalytic conversion have been noted before \cite{Campbell2010,Selinger2012}, and this particular example is Clifford equivalent to that in \cite{Campbell2010}. 
In \sec{Conversions} we introduce two new catalytic conversion families.

% \begin{figure}[ht!]
% 	\includegraphics{TviaControlledS}
% 	\caption{The rightmost circuit implements the conversion of a $\ket{CS}$ into a $\ket{T}$ state, with an additional catalytic $\ket{T}$ state. 
% 	From Theorem~\ref{lem:number-field-constraint}, this conversion would be impossible without catalysis.
% 	The left two circuits explain how the conversion works.}
% 	\label{fig:statedownconversionCSStateToTState}
% \end{figure}

\begin{figure}[ht!]
\center
	\includegraphics{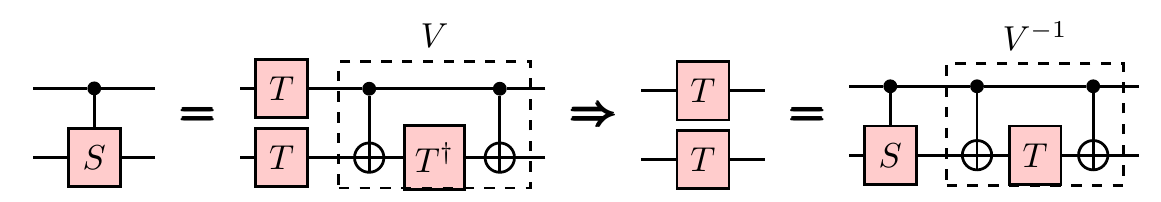}
	\caption{Upon applying the rightmost circuit to the $\ket{+}\ket{+}$ state, $\ket{CS}$ is converted into two copies of $\ket{T}$, using an additional $T$-gate. 
	In terms of resource states $\ket{CS}\ket{T} \longrightarrow \ket{T}\ket{T}$, or equivalently $\ket{CS}\implies \ket{T}$.
	From \theo{number-field-constraint}, the conversion of $\ket{CS}$ into $\ket{T}$ would be impossible without catalysis.
	The leftmost gate identity is the standard implementation of a $CS$-gate from $T$-gates \cite{Barenco1995}.}
	\label{fig:CStoT}
\end{figure}

Clearly we must distinguish scenarios in which catalysts are allowed from those in which they are not allowed. 
Throughout the remainder of the paper we use the following notation of single and double arrows.

\begin{dfn}[Conversion notation]
\label{defn:conversion-equations}
The equation $\ket{A} \rightarrow \ket{B}$ indicates that resource state $\ket{A}$ can be converted into resource state $\ket{B}$ with stabilizer operations in the absence of a catalyst. 
On the other hand, $\ket{A} \catalyse{\ketsm{C}} \ket{B}$, which is equivalent to $\ket{A}\ketsm{C} \rightarrow \ket{B}\ketsm{C}$, indicates the conversion can proceed with the use of a catalyst $\ketsm{C}$ (which we sometimes omit above the arrow).
When a process is impossible, we strike through the arrow, for example $\ket{A} \not\catalyse{} \ket{B}$ signifies that $\ket{A}$ cannot be converted to $\ket{B}$ by stabilizer operations even in the presence of an arbitrary catalyst.
In cases involving multiple copies of a given state such as $\ket{A}^{\otimes 2} \catalyse{\ketsm{C}} \ket{B}$, we sometimes write $2 \ket{A} \catalyse{\ketsm{C}} \ket{B}$ to avoid clutter.
\end{dfn}

\section{Conversion between resource states}
\label{sec:Conversions}

In this section we collect several results on the inter-conversion of resource states.
These state conversion bounds and algorithms put the computational task lower bounds in later sections on more solid footing by allowing us to analyze the cost in terms of different input resource states.
We also foresee our conversion results being useful in a much broader context, such as allowing a meaningful comparison of protocols that distill different types of resource states.

We saw in \sec{catalysis} that some resource conversions are impossible without access to a non-consumable resource often called a catalyst. 
Here we give two families of catalyzed conversion circuits, generalizing the previously known examples  \cite{Campbell2010,Selinger2012,Gidney2018}. 
First, in \sec{polynomialcatalysis} we introduce a general set of techniques for catalytic conversion of Clifford magic states.
Second, in \sec{adder-conversion-protocols} we specify the use of adder circuits to perform catalysis, building on ideas of Gidney~\cite{Gidney2018}. Finally, in \sec{montonebounds}, we utilize the monotones discussed in \sec{basic-techniques} to bound the optimal rates for conversion of resource states. 
One interesting observation is that although many pairs of resource states can be exactly converted into one another, it is impossible to do so without loss, even asymptotically. 
This complements the recent work~\cite{wang2018efficiently} which applies to odd-prime qudits but not qubits.

\subsection{Phase polynomial protocols}
\label{sec:polynomialcatalysis}

\begin{figure}[ht]
    \centering
    \includegraphics{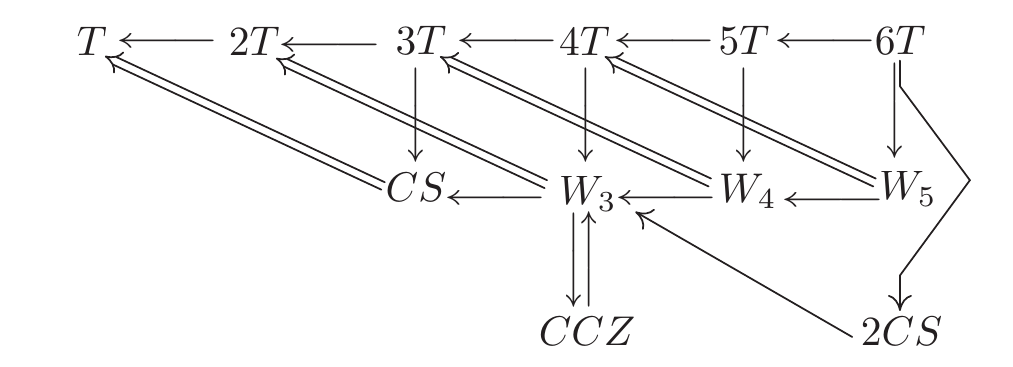}
    \caption{Phase polynomial conversion protocols for states of the form $\ket{U}:=U \ket{+}^{\otimes n}$ where $U$ is any diagonal unitary in the 3$^{\mathrm{rd}}$ level of the Clifford hierarchy.  
    The single arrow $\longrightarrow$ shows when a conversion can be directly realized, whereas a double arrow $\implies$ indicates catalysis is used (and necessary).
    A subclass of $n$-qubit Clifford magic states are denoted $\ket{W_n}$ and arise from the diagonal unitary $W_n = \sum_{x}  \exp(i \pi g(x) / 4   ) \kb{x}{x} ,\quad g(x) = ( \oplus_{i=1}^n x_i )   + \sum_{i=1}^n x_i$, where $\oplus$ denotes addition modulo 2. }
    \label{fig:ConversionGraphic}
\end{figure}    
    
Here we introduce a general set of techniques for catalytic conversion of Clifford magic states.
Our main results are summarized in \fig{ConversionGraphic}.
Recall that for any diagonal unitary $U$ in the 3$^{\mathrm{rd}}$ level of the Clifford hierarchy, the resource state $\ket{U}:=U \ket{+}^{\otimes n}$ can be used to deterministically apply $U$ and is known as a \textit{Clifford magic state}.  The unitary $U$ can always be implemented using CNOT, $S$ and $T$ gates~\cite{amy2016t,heyfron2018efficient}.  The Clifford hierarchy is nested, so that the Clifford group (the 2$^{\mathrm{nd}}$ level) is contained within the 3$^{\mathrm{rd}}$ level.  %Consequently, some results presented here will also apply to Clifford group unitaries (for which $\tau(U)=0$) but we are interested the nonClifford gates in the 3$^{\mathrm{rd}}$ level.
We have the following result
%\begin{thm}
\begin{restatable}{thm}{PhasePolyCat}
\label{thm:PhasePolyCat}
	Let $\ket{U}=U\ket{+}^{\otimes n}$ be an $n$-qubit magic state for a diagonal unitary $U$ from the 3$^{\mathrm{rd}}$ level of the Clifford hierarchy, and let $\tau(U)$ be the minimum number of $T$ gates needed to implement $U$ using the gate set $\{CNOT, S,T \}$.
    The following resource conversion is possible
	\begin{equation}
	\ket{U}   \catalyse{\ket{T}^{\otimes \tau(U)-\nu( \ket{U} )}} \ket{T}^{\otimes 2\nu( \ket{U} ) - \tau(U)}.
\end{equation}	
% \end{thm}
\end{restatable}

In the theorem, we follow the conversion notation of \defn{conversion-equations} and use $\nu$ that was defined earlier as the stabilizer nullity (recall \defn{stabilizer-monotone}). 
The proof of this theorem can be found in \app{phase-polynomial}.

An interesting family of  ($n>1$-qubit) unitaries that we call $W_n$ are 
% defined via the phase polynomial matrix 
%  \begin{equation}
%      P_n =( \id_n, 1) = \left(  \begin{array}{ccccc}
%      1 & 0 & & 0 & 1\\
%      0 & 1 & & 0 & 1\\
%      &   & \ddots & 0 & 1 \\
%      0 & 0 & 0 & 1  & 1
%      \end{array} 
%      \right) \label{eq:P_for_W},
% \end{equation}
% which is the identity matrix padded with an all-one column. More explicitly, we have $W_n$ 
\begin{equation}
     W_n = \sum_{x}  \exp(i \pi g(x) / 4   ) \kb{x}{x} , \text{ with } g(x) = ( \oplus_{i=1}^n x_i )   + \sum_{i=1}^n x_i,
\end{equation}
where the $\oplus$ sum is performed modulo 2. The corresponding Clifford magic state $\ket{W_n}=W_n\ket{+}^{\otimes n}$, when expressed as a density matrix $\rho=\vert W_n \rangle \! \langle W_n \vert$, has entries in $\mathbb{Q}[i]$ by virtue of the fact that $g(x) \equiv 0 \mod 2$ for all $x$. By \theo{number-field-constraint} this implies that no $T$-states can be derived from $\ket{W_n}$ in the absence of a catalyst. In \fig{Wn_Catalysis} we give an explicit circuit for converting  $\ket{W_n}$ to $\ket{T}^{\otimes n-1}$ using catalysis. This matches \theo{PhasePolyCat} by virtue of the following lemma, which is proved in \app{phase-polynomial}. 
% \begin{lem}
% \label{lem:tau-for-Wn}
% $\tau(W_n)=n+1$.
% \end{lem}

\begin{restatable*}{lem}{tauforWn}
\label{lem:tau-for-Wn}
$\tau(W_n)=n+1$.
\end{restatable*}

% \begin{restatable*}{lem}{tauforWn}
% \label{lem:tau-for-Wn}
% $\tau(W_n)=n+1$.
% \end{restatable*}

\begin{figure}[ht!]
\center
% 	\hspace*{-0.65cm}
	\includegraphics[scale=0.62]{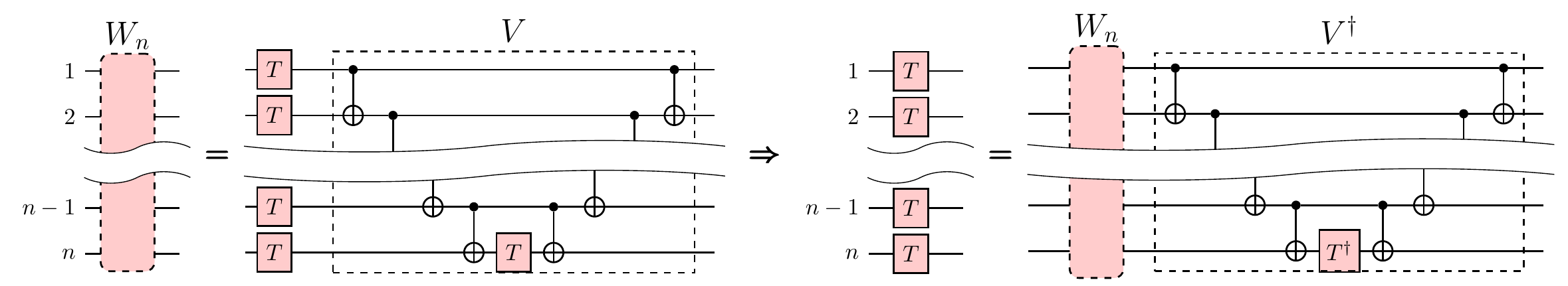}
	\caption{The first circuit identity can be verified explicitly, and follows from reasoning presented in \cite{amy2016t}.
	By applying the rightmost circuit to the $\ket{+}^{\otimes n}$ state, $\ket{W_n}$ is converted into $n$ copies of $\ket{T}$, using an additional $T^\dagger$-gate. 	In terms of resource states $\ket{W_n}\ket{T} \longrightarrow \ket{T}^{\otimes{n}}$, or equivalently $\ket{W_n} \implies \ket{T}^{\otimes{n-1}}$. The leftmost circuit identity depicts how the unitary $W_n$-gate can be implemented using a minimal (i.e., $\tau(W_n)=n+1$) number of $T$-gates. }
	\label{fig:Wn_Catalysis}
\end{figure}

% The $\{CNOT,T\}$ circuit that optimally synthesizes $W_n$ is in \fig{Wn_Catalysis}.

% Notice that every row of $P_n$ in \eq{P_for_W} is even weight and so by Claim~\ref{EvenWeightClaim} we know $\ket{W_n} \nrightarrow  \ket{T}$.  Since $P$ has $n+1$ columns, we have $\tau(W_n)=n+1$.  Since $P$ is a full rank matrix, we have $\nu(\ket{W_n})=n$. Therefore, $2\nu( \ket{U} ) - \tau(U) =n-1$ and by Thm.~\ref{thm:PhasePolyCat} we conclude that 
% \begin{equation}
%  \ket{W_n} \implies \ket{T}^{\otimes n-1}.
% \end{equation}  
% These are the most illuminating examples that one can obtain from Thm.~\ref{thm:PhasePolyCat} because assuming $U \neq T^{\otimes n}$ we know $\tau(U)>\nu(U)$ and then $\tau(U)=\nu(U)+1$ leads to the best possible catalysis protocols. 

% At first glance, the $W_n$ unitaries may look unfamiliar.  However, $W_2$ has the same non-Clifford part as $CS$ and so they are equivalent up to Cliffords.  The $W_2$ example is also equivalent to the catalysis protocol first observed by Campbell~\cite{Campbell2010}. For $W_3$, we have that the state $\ket{W_3}$ is Clifford equivalent to $\ket{CCZ}$ and so $\ket{CCZ} \implies \ket{T}^{\otimes 2}$, which is the catalysis protocol observed by Gidney and Fowler~\cite{gidney2018efficient}.  The Clifford equivalence of $\ket{W_3}$ and $\ket{CCZ}$ may be not obvious and so we comment further on this.  We have that $CNOT_{3,2} W_3 CNOT_{3,2}$ has the same phase polynomial matrix as $V=CCZ_{1,2,3}CS_{2,3}$.  Furthermore, Cody Jones~\cite{Jones2012} showed that $V$ can be used to synthesize $CCZ$, which establishes the equivalence.
\begin{figure}[ht!]
\center
	\includegraphics[scale=0.85]{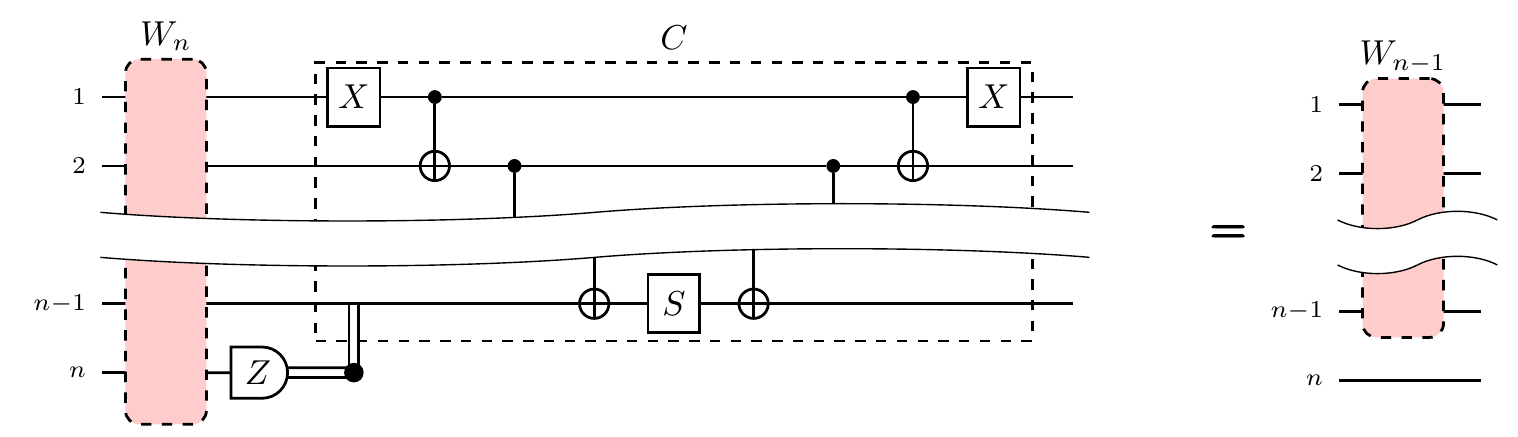}
	\caption{By applying the circuit on the left to the $\ket{+}^{\otimes n}$ state, $\ket{W_n}$ is converted into $\ket{W_{n-1}}$, where, depending on the $Z$ measurement result, a Clifford correction $C$ as in \eq{Correction} may be required.
	This equality can be understood as follows. 
	First note that the gate is symmetric with respect to permutations of qubits, so measuring the last qubit is equivalent to meausring the first. 
	Since a pair of CNOT gates are applied in $W_n$ controlled on the first qubit, if the outcome is zero, then those CNOTs can be removed, and $W_{n-1}$ is applied directly. 
	On the other hand, if the outcome is one, then instead of $W_{n-1}$, we have $W_{n-1}$ sandwiched between a pair of $X$ gates applied to the target of those CNOTs.
	This can be fixed by the content of the dashed box, which can be verified by taking the product of the gate which is applied with $W_{n-1}$, propagating the $X$s through the circuit and making use of $XTX^\dagger = T^\dagger$ before cancelling adjacent CNOTs.
	}
	\label{fig:Wn_level_reduction}
\end{figure}

We also have that the $W_n$ states can be reduced in the sense that
\begin{equation}
\label{WnReduction}
   \ket{ W_n} \rightarrow  \ket{ W_{n-1} }.
\end{equation}
This conversion is achieved by first measuring the last qubit in the computational basis as shown in \fig{Wn_level_reduction}. If one obtains the ``0" outcome, we immediately have the state $\ket{W_{n-1}} \ket{0}$ and so just discard the last qubit.  In the case of a ``1" outcome, then a Clifford correction 
\begin{equation}
 C = \sum_x  i^{1  \oplus_{j=1}^n x_j} \kb{x}{x} \label{eq:Correction},
\end{equation}
is required. Performing this correction and discarding the last qubit we again obtain $\ket{W_{n-1}}$.

\subsection{One-bit adder conversion protocols}
\label{sec:adder-conversion-protocols}

In this subsection, we present a class of protocols that use catalysis to convert resource states for the third level of the Clifford hierarchy (i.e. Clifford magic states) to resource states for the higher levels of the Clifford hierarchy.
It is beneficial to apply some of the protocols directly at the gate level too.
The main building block for the protocols in this subsection is the circuit shown on~\fig{three-sqrt-t-gates},
which is a special case of an idea described
on \href{https://arxiv.org/pdf/1709.06648v3.pdf#page=4}{Page~4} in \cite{Gidney2018}. 
This circuit implements three $\sqrt{T}$ gates (which are in the $4^{\text{th}}$ level of the Clifford heirarchy) using one $\sqrt{T}$ gate along with a few gates from the third level of the Clifford hierarchy. 
The key difference in our approach from that in Ref.~\cite{Gidney2018} is that to scale this small example to parallel rotations on an $n$-qubit register, we use recursion, whereas in Ref.~\cite{Gidney2018} a Hamming weight generalization is used.
Compared with the Hamming weight construction, our recursive construction amortizes the cost of the correction operations associated with injecting gates from higher levels of the Clifford hierarchy. 
Later, in \sec{prob-half-conversion-bounds}, we show that our construction is asymptotically optimal under the assumption that only measurements with probability one-half are used. 

\begin{figure}[ht]
\centering
  \begin{subfigure}{0.8\columnwidth}
	\centering
	\includegraphics[scale=0.66]{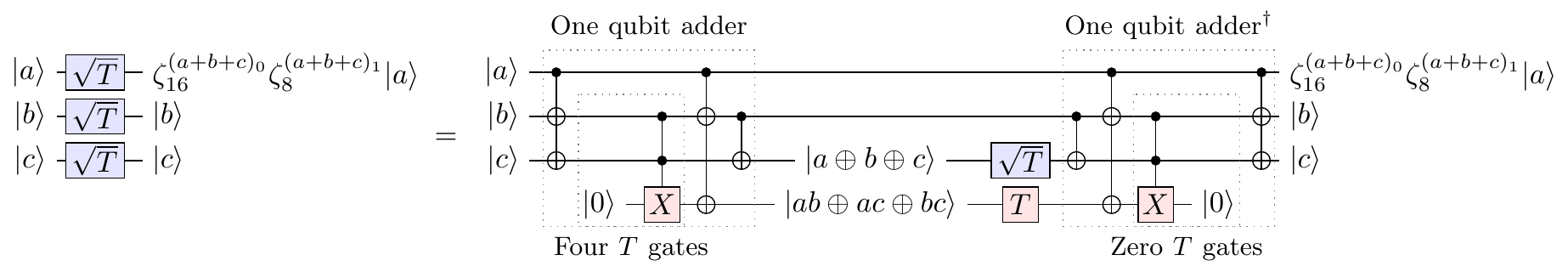}
	\caption{Three $\sqrt{T}$ gates can be applied using a circuit with just one $\sqrt{T}$ gate and other gates in the third level of the Clifford hierarchy.  
	This uses the Hamming weight register idea from~\cite{Gidney2018} along with the adder from Figure~4 in~\cite{Gidney2018}.}
	\label{fig:three-sqrt-t-gates}
	\end{subfigure}
  \vspace{1em}
    \begin{subfigure}{0.8\columnwidth}
  	\centering
  	\includegraphics[scale=1]{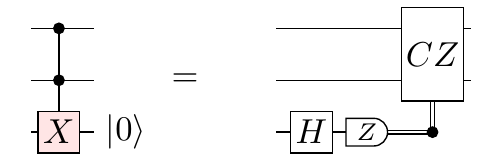}
  	\hspace{2em}
  	\includegraphics[scale=1]{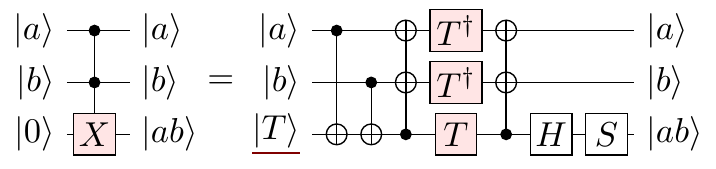} 
  	\caption{Simplified circuits for the Toffolli (i.e. the doubly-controlled-X gate) when the target qubit ends in the $\ket{0}$ state, and when the target qubit starts in the $\ket{0}$ state.
  	See {Figure~3}~in~\cite{Gidney2018}.
  	Using the first of these circuits as a subroutine, \fig{three-sqrt-t-gates} implements three $\sqrt{T}$ gates using one $\sqrt{T}$ gate, one $T$ gate and one $CCX$ gate.
  	Additionally making use of the second of these circuits, \fig{three-sqrt-t-gates} implements three $\sqrt{T}$ gates using one $\sqrt{T}$ gate and five $T$ gates.
    }
  	\label{fig:simpler-CCX}
  \end{subfigure}
  \caption{Circuits for applying three $\sqrt{T}$ gates using five $T$ gates and one $\sqrt{T}$ gate.}
  \label{fig:sqrtT-via-CCZ-and-T}
\end{figure}

To understand the circuit in~\fig{three-sqrt-t-gates}, first note that the gate
$\exp{\at{i\theta \ket{1}\bra{1}}}^{\otimes n}$
acting on an $n$-qubit register in the computational basis state $\ket{w}$ gives 
$e^{i \theta \cdot \mathrm{hw}(w)}\ket{w}$,
where $\mathrm{hw}(w)$ is the Hamming weight of the bit string $w$.
Therefore an alternative way of applying the gate $\exp{\at{i\theta \ket{1}\bra{1}}}^{\otimes n}$ is
to compute the binary representation of $\mathrm{hw}(w)$ and store it in a quantum register $\ket{x_k \ldots x_0 }$,
and for $j$ from $0$ to $k$ apply $\exp{\at{i 2^j \theta \ket{1}\bra{1}}}$ to qubit $j$ in the register. 
In~\fig{three-sqrt-t-gates} we use the adder circuit shown
in~\href{https://arxiv.org/pdf/1709.06648v3.pdf#page=4}{Figure~4} in~\cite{Gidney2018}
to compute the Hamming weight of the bit string $a,b,c$.
For bit strings of length three the Hamming weight can be represented using two bits. 
The lower bit is the parity $a\oplus b \oplus c$
and the higher bit is the majority function $ab\oplus bc\oplus ac$.
These are exactly the values computed by the adder. 
An important efficiency gain comes from the observation illustrated in the first circuit in \fig{simpler-CCX} 
that the one qubit adder can be un-computed by using Clifford gates and single qubit Pauli measurements only~\cite{Gidney2018}. 
With this trick, the circuit shown in \fig{three-sqrt-t-gates} applies three $\sqrt{T}$ gates using only one $\sqrt{T}$ gate,
and either one $T$ gate and one $CCX$ gate,
or five $T$ gates if the second circuit in \fig{simpler-CCX} is used.

\begin{figure}[ht]
  \centering
  \includegraphics[scale=1]{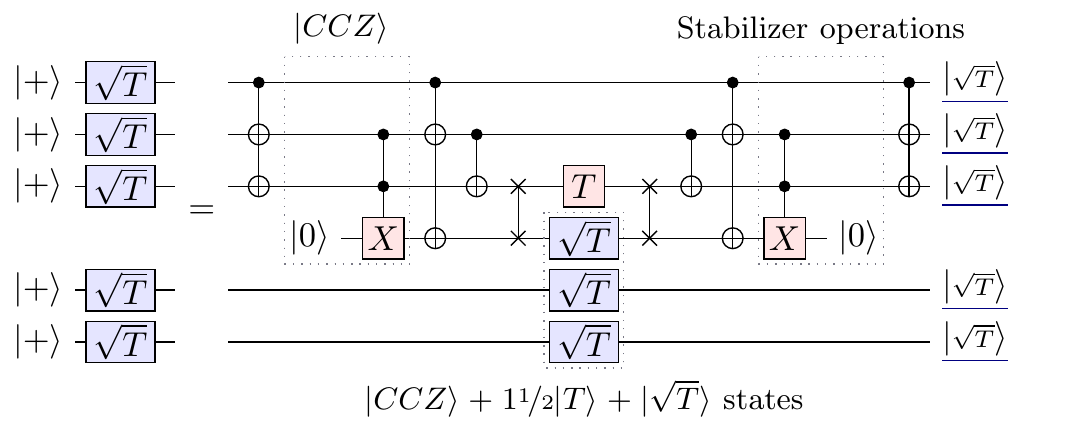}
  \caption{ \label{fig:many-sqrt-t-states} Conversion $k\ketsm{CCZ} + (k+\nicefrac{1}{2})\ket{T} \catalyse{\ketsm{\sqrt{T}}} 2k\ketsm{\sqrt{T}}$ for $k=2$.}
  
\end{figure}

\fig{three-sqrt-t-gates} can be adapted to form resource state conversion protocols. 
For example the protocol $\ketsm{\sqrt{T}} + 5.5\ket{T} \rightarrow 3\ketsm{\sqrt{T}}$ follows directly when $\ket{+}$ states are fed into \fig{three-sqrt-t-gates} when the $\sqrt{T}$ gate is implemented by injection of $\ket{\sqrt{T}}$, and when the third Clifford level gates are implemented with $\ket{T}$ resource states.
We need to use $5$ $\ket{T}$ states to implement the first $CCX$ gate and the $T$ gate in \fig{three-sqrt-t-gates}. 
The $\sqrt{T}$ gate is implemented with the injection circuit which requires an additional $T$ gate correction $50\%$ of the time, which adds $1/2$ to the number of $\ket{T}$ states consumed (on average). 

The extra $T$ gate can be amortized to give the conversion 
\[
k\ketsm{CCZ} + (k+\nicefrac{1}{2})\ket{T} \catalyse{\ketsm{\sqrt{T}}} 2k\ketsm{\sqrt{T}}
\]
valid for any positive integer $k$.
We see that asymptotically $\ketsm{\sqrt{T}}$ state costs
half of $\ketsm{T}$ state plus half of $\ketsm{CCZ}$ state. 
Using the circuit on \fig{three-sqrt-t-gates} we can reduce the parallel application of $2k+1$ $\sqrt{T}$ gates
to the parallel application of $2k-1$ $\sqrt{T}$ gates for any positive integer $k$. 
We use the circuit on \fig{three-sqrt-t-gates} to execute first three out of $2k+1$ $\sqrt{T}$ gates by only using one $\sqrt{T}$ gate. 
Then we observe that the rest of the $2k-2$ $\sqrt{T}$ gates can be executed in parallel with the newly introduced $\sqrt{T}$ gate. 
\fig{many-sqrt-t-states} shows how to reduce the parallel application of five $\sqrt{T}$ gates to the parallel application of three $\sqrt{T}$ gates. 
% This is slightly different from the approach in~\cite{Gidney2018} to applying many single qubit rotations by angle $\theta$ in parallel.
% The author suggests computing the Hamming weight of the whole register and then applying rotations by angles $2^p \theta$. 
% Our constructions leads to the lower $T$ count for the special case of applying many $\sqrt{T}$ gates in parallel.
We also note that all above results can be applied to conversion of $\ketsm{\sqrt{T}^3}$ states and application of $\sqrt{T}^3$ gates. The cost of applying $\sqrt{T}^3$ gate is the same as the cost of applying $\sqrt{T}$. This is an important observation for single qubit circuit synthesis applications.

A similar idea leads to the lower cost of applying many $T^{j/2^{d-2}} = \exp\atsm{ \pi i j / 2^d \ket{1}\bra{1}}$ for positive integer $d \ge 2$ and an odd $j$ proved in \app{dyadic-powers-of-t}. 

\begin{restatable*}{thm}{dyadicpowersconversion}
\label{thm:dyadic-powers-conversion}
Let $k, d \ge 1$ be positive integers and let $j$ be an odd integer and let $a_{d,k} = 2^{d-1}(k-1)+2$. Then $a_{d,k}$ gates $\exp\atsm{ \pi i j / 2^d \ketsm{1}\brasm{1}}$ can be executed in parallel by using stabilizer operations with measurements that have probability $50\%$,
$b_{d,k} = \atsm{2^{d-1}-1}\at{k-1}+d-1$ copies of $\ket{CCZ}$ state and 
using one copy of each of the states
$\ketsm{ \pi j / 2^{d}},\ketsm{ \pi j / 2^{d-1}},\ldots,\ketsm{\pi j / 2^2}$ as a catalyst.
Asymptotically, the state $\ketsm{ \pi j / 2^d }$ is produced using $1-1/2^{d-1}$ $\ketsm{CCZ}$ states.
\end{restatable*}

Note that the number of $\ket{CCZ}$ states used by protocols described in the theorem above is 
asymptotically the same as the lower bounds established later in~\sec{prob-half-bounds}:

\begin{restatable*}{lem}{dyadicpowersbound}
\label{lem:dyadic-powers-bound}
Consider a protocol that uses stabilizer operations with measurements probability $50\%$,
$\ket{CCZ}$ states and a multi-qubit state as a catalyst.
The catalyst has entries in $\mathcal{R}_{d'}$ for some $d'$.
% The catalyst can be written as $\ket{U}$ for some diagonal unitary that belongs to
% finite level of the Clifford hierarchy.
Suppose that such circuit uses $k$ $\ket{CCZ}$ states and 
produces $n$ states $\ketsm{ \pi j / 2^d }$ for odd $j$ and integer $d \ge 2$, then $k \ge n(1-1/2^{d-1})$.
Asymptotically, at least $1-1/2^{d-1}$ copies of $\ketsm{CCZ}$ state are needed to produce state $\ketsm{ \pi j / 2^d }$.
\end{restatable*}

These protocols are useful for reducing the cost of approximate unitary synthesis, as described in \app{seqeuntial-root-t}.

\subsection{Conversion bounds} 
\label{sec:montonebounds}

Suppose we can identify a monotone $\mathcal{M}$ such that $\mathcal{M}(\ket{\psi})$ is real for any state $\ket{\psi}$, and is non-increasing under stabilizer operations. 
We say such a function $\mathcal{M}$ is a \textit{monotone}, and can use it to bound conversion processes since for example, a resource state $\ket{A}$ cannot be used to produce a resource state $\ket{B}$ with stabilizer operations if $\mathcal{M}(\ket{A})< \mathcal{M}(\ket{B})$, i.e.,
$$ \mathcal{M}(\ket{A})< \mathcal{M}(\ket{B}) ~~~ \text{implies} ~~~ \ket{A} \not\longrightarrow \ket{B}. $$
If the monotone is also \textit{additive}, such that $\mathcal{M}(\ket{\psi}\otimes \ket{\phi}) = \mathcal{M}(\ket{\psi}) + \mathcal{M}(\ket{\phi})$ for all $\ket{\psi}$ and $\ket{\phi}$, then we can say even more. 
For example we can rule out catalyzed conversions since $\mathcal{M}(\ket{A})< \mathcal{M}(\ket{B})$ implies that $\mathcal{M}(\ket{A} \otimes \ket{\text{cat.}})< \mathcal{M}(\ket{B} \otimes \ket{\text{cat.}})$ for any catalyzing state $\ket{\text{cat.}}$.
Tensor powers of states simplify, allowing us to make asymptotic implications, i.e.,
$$ \mathcal{M}(\ket{A})< \alpha \cdot \mathcal{M}(\ket{B}) ~~~ \text{implies} ~~~ \ket{A}^{\otimes n} \not\Longrightarrow \ket{B}^{\otimes \alpha n} ~~~\forall n.$$
using the arrow notation described in \defn{conversion-equations}.
In other words this would put an upper bound of $\alpha$ on the catalytic rate of conversion from $\ket{A}$ to $\ket{B}$.
Note that equivalent implications hold if the monotone is \textit{multiplicative} rather than additive, i.e., if $\mathcal{M}(\ket{\psi}\otimes \ket{\phi}) = \mathcal{M}(\ket{\psi}) \cdot \mathcal{M}(\ket{\phi})$ for all $\ket{\psi}$ and $\ket{\phi}$.

For example, consider the states $\ket{T}$ and $\ket{CCZ}$ for which the best known conversion algorithms are:
\begin{eqnarray}
4 \ket{T} & \rightarrow & \ket{CCZ}, \nonumber \\
\ket{CCZ} & \catalyse{\ket{T}} & 2 \ket{T}. \nonumber
\end{eqnarray}
Clearly these algorithms would have loss if feeding the output of one into the other.
The best possible conversion algorithms have (for any $n$ and any catalyst) the minimum $r$ and maximum $r'$ in 
\begin{eqnarray}
r n \ket{T} & \catalyse{} & n \ket{CCZ}, \nonumber \\
n \ket{CCZ} & \catalyse{} & r' n \ket{T}. \nonumber
\end{eqnarray}

It is straightforward to compute the stabilizer nullity values $\nu(\ket{T})=1$ and $\nu(\ket{CCZ})=3$.
As described above, the fact that $\nu$ is additive immediately implies $r \geq 3$ and $r' \leq 3$.
It is also possible to compute the extent values $\xi(\ket{T})= (\text{sec}~\pi/8)^2= 1.17157$ and $\xi(\ket{CCZ})=16/9 =1.77778$.
Moreover, $\log \xi$ is an additive monotone with respect to collections of $\ket{T}$ states and $\ket{CCZ}$ states (which satisfy \lemm{stabilizer-extent-mult})\footnote{Note that the bounds from the extent are only guaranteed to hold for catalysts which satisfy \lemm{stabilizer-extent-mult}, but if (as conjectured) the extent is multiplicative for all states then it will hold in general.} and therefore $r \geq \log[1.77778]/\log[1.17157] = 3.63356$ and $r' \leq \log[1.77778]/\log[1.17157] = 3.63356$.
We therefore have that $r \geq 3.63356$ and $r' \leq 3$.
From these bounds we see there is a gap: the best possible algorithm would require at least $3.63$ $\ket{T}$ states to produce a $\ket{CSS}$ state, which can then be converted back into at most 3 $\ket{T}$ states.
In \tab{T_to_Targ} and \tab{CCZ_to_Targ} we show these conversion bounds along with those for many other pairs of states.

\begin{table}[ht]
    \centering
    {\renewcommand{\arraystretch}{1.25}
    \begin{tabular}{c|c|c}
     $\ket{\psi}$
     &
     \begin{tabular}{@{}c@{}} 
        Best algo. (lower bound) [Ref.] \\ 
        $r n \ket{T}    \catalyse{} n \ket{\psi}$ 
     \end{tabular}
     &
     \begin{tabular}{@{}c@{}}
        Best algo. (upper bound) [Ref.]\\ 
        $n \ket{\psi} \catalyse{} r' n \ket{T}$
     \end{tabular} 
     \\
     \hline
        $\ket{\sqrt{T}}$ &
        2.5 (1) [\figsm{many-sqrt-t-states}] & 0.25 (0.754933*)  \\
        $\ket{T}$ & 1 (1) & 1 (1)\\
        $\ket{CS} = \ket{W_2}$ &
        3 (2.96818*) [\figsm{CStoT}] & 1 (2) [\figsm{CStoT}] \\
        $\ket{CCS}$ & 7 (4.53328*) \cite{Jones2012} & 0.5 (3) [\propsm{measure-control}] \\
        $\ket{C^3S}$  & 11 (4) \cite{Jones2012} & 0.25 (3.82743*) [\propsm{measure-control}] \\
        $\ket{CCZ}$  & 4 (3.63356*) \cite{Jones2012} &
        2 (3) \cite{gidney2018efficient} \\
        $\ket{C^3Z}$  & 6 (5.12122*) \cite{Jones2012} & 1 (4) [\tabsm{CCZ_to_Targ}] \\
        $\ket{C^4Z}$  & 12 (5) \cite{Jones2012} & 0.5 (3.8233*) [\tabsm{CCZ_to_Targ}] \\
        $\ket{CCZ_{123,145}}$ & 8 (5) [\tabsm{CCZ_to_Targ}] & 2 (4.37739*) [\tabsm{CCZ_to_Targ}]  \\
        $\ket{W_3}$ & 4 (3.63356*) [\tabsm{CCZ_to_Targ}] & 2 (3) [\tabsm{CCZ_to_Targ}] \\
        $\ket{W_4}$ &
        5 (4.99907*) [\figsm{Wn_Catalysis}] & 3 (4) [\figsm{Wn_Catalysis}]\\ 
        $\ket{W_5}$ &
        6 (5.93637*) [\figsm{Wn_Catalysis}] & 4 (5) [\figsm{Wn_Catalysis}]
    \end{tabular}
    }
    \caption{Catalytic conversion rates to and from $\ket{T}$ states. 
    In the first column, the produced or consumed state is specified. 
    The second and third columns list the conversion rates ($r$ to consume, and $r'$ to produce)  for the best known algorithm, along with the tightest bound implied by stabilizer extent or nullity in parenthesis.
    The references to particular results from \app{explicit-circuits} are provided in square brackets.
    Note that bounds from the stabilizer extent, marked here by an asterisk, are not known to hold for arbitrary catalysts since the stabilizer extent is currently not known to be multiplicative for all states.
    The results that reference \tab{CCZ_to_Targ} are direct consequence of corresponding result from the table together with the inter-conversion between $\ketsm{T}$ and $\ketsm{CCZ}$. 
    The values of the extent are calculated in \app{extent-values}.
    }
    \label{tab:T_to_Targ}
\end{table}

\begin{table}[ht]
    \centering
    {\renewcommand{\arraystretch}{1.25}
    \begin{tabular}{c|c|c}
 $\ket{\psi}$  &  \begin{tabular}{@{}c@{}} Best algo. (lower bound) \\ $r n \ket{CCZ} \catalyse{} n \ket{\psi}$ \end{tabular} & \begin{tabular}{@{}c@{}} Best algo. (upper bound) \\ $n \ket{\psi} \catalyse{} r' n \ket{CCZ}$ \end{tabular}\\ \hline
$\ket{\sqrt{T}}$  &
0.75 (0.33333) [\figsm{many-sqrt-t-states}] & 0.0625 (0.207767*) [\tabsm{T_to_Targ}] \\
$\ket{T}$   & 0.5 (0.33333) \cite{gidney2018efficient} & 0.25 (0.275212*) \cite{Jones2012} \\
$\ket{CS} = \ket{W_2}$  & 1 (0.81688*) [\figsm{Wn_level_reduction}, \figsm{ccz-to-cs}] & 0.5 (0.66666)  [\figsm{cs-to-ccz}] \\
$\ket{CCS}$ & 2 (1.24763*) \cite{Jones2012} & 0.25  (1) [\propsm{measure-control}] \\
$\ket{C^3S}$ & 3 (1.33333) \cite{Jones2012} & 0.125 (1.05336*) [\propsm{measure-control}] \\
$\ket{CCZ}$ & 1(1) & 1(1)  \\
$\ket{C^3Z}$ & 2 (1.40942*) \cite{Jones2012} & 0.5 (1.33333) [\propsm{measure-control}]   \\
$\ket{C^4Z}$ & 3 (1.66667) \cite{Jones2012} & 0.25 (1.05336*) [\propsm{measure-control}] \\
$\ket{CCZ_{123,145}}$ & 2 (1.66667) [\figsm{ccz123145}] & 1 (1.20471*) [\figsm{ccz123145}] \\
$\ket{W_3}$ & 1 (1) \cite{HowardCampbell2016}  & 1 (1)  \cite{HowardCampbell2016}  \\
$\ket{W_4}$ & 2.5 (1.3758*) [\tabsm{T_to_Targ}] & 1 (1.33333) [\figsm{Wn_level_reduction}] \\
$\ket{W_5}$  & 3 (1.66667) [\tabsm{T_to_Targ}] & 1 (1.63376*) [\figsm{Wn_level_reduction}]
    \end{tabular}
    }
    \caption{Catalytic conversion rates to and from $\ket{CCZ}$ states. 
    In the first column, the produced or consumed state is specified. 
    The second and third columns list the conversion rates ($r$ to consume, and $r'$ to produce)  for the best known algorithm, along with the tightest bound implied by stabilizer extent or nullity in parenthesis.
    The references to particular results from \app{explicit-circuits} are provided in square brackets.
    Note that bounds from the stabilizer extent, marked here by an asterisk, are not known to hold for arbitrary catalysts since the stabilizer extent is currently not known to be multiplicative for all states.
    The results that reference \tab{T_to_Targ} are direct consequence of corresponding result from the table together with the inter-conversion between $\ketsm{T}$ and $\ketsm{CCZ}$. 
    The values of the extent are calculated in \app{extent-values}.}
    \label{tab:CCZ_to_Targ}
\end{table}

\section{Computational task lower bounds}
\label{sec:computational-tasks}

\label{sec:lower-bounds}

In the previous section we established bounds on the resources required to
produce specific states. In this section we lower bound the non stabilizer resources needed to
implement some important computational tasks.
Specifically, we consider the multiply controlled $Z$ gate in \sec{CnZ-bound} and
the modular adder in \sec{adder-bound}.

Our strategy to lower bound the number of copies of a resource state $\ket{\psi}$
needed to implement a unitary $U$ (where $U$ corresponds to some computational task) is to bootstrap bounds on the resources required to
produce specific states. 
For example, note that a lower bound of the number of copies of  $\ket{\psi}$ needed to
produce a \textit{state} $U\ket{S}$, where $\ket{S}$ is a stabilizer state, 
also serves as a lower bound for applying $U$.
It is also useful to consider catalysis when establishing lower bounds for computational tasks.
For example, suppose $U$ maps the state $\ket{\Psi}\ket{S}$ to a state 
$\ket{\Psi}\ket{\Phi}$ for some non-stabilizer states $\ket{\Phi}$, $\ket{\Psi}$,
then the number of copies of $\ket{\psi}$ needed to (catalytically) produce $\ket{\Phi}$
also serves as a lower bound for applying $U$. 

\subsection{Lower bounds for the \texorpdfstring{$C^n Z$}{multiply-controlled-Z} gate}
\label{sec:CnZ-bound}

The multiply controlled $Z$ gate $C^n Z$ is a key component of many important algorithms, for example to implement the reflection step in Grover's search \cite{Grover1996}.
We can lower bound the resources required to implement $C^n Z$ as follows:

\begin{prop} \label{prp:CnZlowerbound}
	For $n\geq 3$, it is not possible to apply the multiply controlled $Z$ gate $C^{n-1} Z$ or produce the state $\ket{C^{n-1} Z} = C^{n-1} Z \ket{+}^{\otimes n}$ by Clifford gates and measurements using fewer than $n$ $\ket T$ states, or $n/2$ $\ket{CS}$ states, or $n/3$ $\ket{CCZ}$ states.
\end{prop}

\begin{proof}
	First note that proving that a bound holds for the state $\ket{C^{n-1} Z}$ implies that it holds for the gate $C^{n-1} Z$.
	The proof for each of the bounds is then very straightforward: we simply show that the stabilizer nullity of the input state is smaller than the output state unless the bound is satisfied.
	Direct verification shows that $\nu(\ket{T})=1, \nu(\ket{CS})=2$, and $\nu(\ket{CCZ})=3$.
	Finally, it is clear that $\nu(\ket{C^{n-1} Z})=n$ for all $n\geq 3$ from \propos{CnZPauliSpectrum} since we see that no non-trivial Pauli operator has expectation value $+1$ for the state $\ket{CCZ}$.	
\end{proof}

\begin{prop} \label{prop:CnZPauliSpectrum}
	For all $n \geq 3$, the Pauli spectrum of the state $\ket{C^{n-1} Z} = C^{n-1} Z \ket{+}^{\otimes n}$ has values (and multiplicities): $1 ~(1)$; $0~(-1 + 2^{n - 1} + 2^{2 n - 1})$; $1-2^{2-n}~(2^n-1)$; $2^{2-n}~(1 - 3 \cdot 2^{n - 1} + 2^{2 n - 1})$.
\end{prop}

\begin{proof}
	Consider the multiply controlled $Z$ state $\ket{C^{n-1}Z}$, defined as 
	\begin{eqnarray}
	\ket{C^{n-1}Z} = C^{n-1}Z |+\rangle^{\otimes n} =\frac{1}{\sqrt{2^n}} \sum_{b \in \{0,1\}^n} (-1)^{b_1 \cdot b_2 \cdot \dots \cdot b_n} \ket{b}.\nonumber
	\end{eqnarray}
	We are interested in the set of Pauli expectation values $\langle C^{n-1}Z | X^x Z^z  |C^{n-1}Z \rangle$ for arbitrary bit strings $x$ and $z$.
	Explicit calculation shows
	\begin{eqnarray}
	2^n \langle C^{n-1}Z | X^x Z^z  |C^{n-1}Z \rangle & = & \sum_{b,b' \in \{0,1\}^n} (-1)^{b_1 \cdot b_2 \cdot \dots \cdot b_n}(-1)^{b'_1 \cdot b'_2 \cdot \dots \cdot b'_n} \bra{b'}X^x Z^z \ket{b},\nonumber \\
	& = & \sum_{b,b' \in \{0,1\}^n} (-1)^{b_1 \cdot b_2 \cdot \dots \cdot b_n}(-1)^{b'_1 \cdot b'_2 \cdot \dots \cdot b'_n}(-1)^{z \cdot b} \bra{b'+x} b \rangle, \nonumber\\
	& = & \sum_{b \in \{0,1\}^n} (-1)^{b_1 \cdot b_2 \cdot \dots \cdot b_n}(-1)^{(b_1+x_1) \cdot (b_2+x_2) \cdot \dots \cdot (b_n+x_n)}(-1)^{z \cdot b}.\nonumber
	\end{eqnarray}
	When $x = 0^n$, we see that the sum simplifies to $\sum_{b \in \{0,1\}^n} (-1)^{z \cdot b}$, which is $2^n$ for $z=0^n$, and $0$ for any other $z$.
	For $x \neq 0^n$, note that the terms in the sum over $b$ differ from $\sum_{b \in \{0,1\}^n} (-1)^{z \cdot b}$ only for $b = 1^n$ and $b = 1^n+x$. Therefore, 
	\begin{eqnarray}
	2^n \langle C^{n-1}Z | X^x Z^z  |C^{n-1}Z \rangle 
	& = & \sum_{b \in \{0,1\}^n} (-1)^{b_1 \cdot b_2 \cdot \dots \cdot b_n}(-1)^{(b_1+x_1) \cdot (b_2+x_2) \cdot \dots \cdot (b_n+x_n)}(-1)^{z \cdot b}, \nonumber \\
	& = & - 2(-1)^{z \cdot 1^n}- 2(-1)^{z \cdot (1^n+x)}+\sum_{b \in \{0,1\}^n} (-1)^{z \cdot b}. \nonumber
	\end{eqnarray}
	When $z=0^n$, this is simply $2^n - 4$. When $z \neq 0^n$, it is $- 2(-1)^{z \cdot 1^n}- 2(-1)^{z \cdot (1^n+x)}$.
	Summarizing,
	\begin{eqnarray}
	| \langle C^{n-1}Z | X^x Z^z  |C^{n-1}Z \rangle | & = & \begin{cases} 
	1 &\mbox{if } x = 0^n \mbox{ and } z = 0^n, \\ 
	0 &\mbox{if } x \cdot z \mbox{ is odd, or if } x = 0^n \mbox{ and } z \neq 0^n, \\
	1-2^{2-n} &\mbox{if } x \neq 0^n \mbox{ and } z = 0^n, \\
	2^{2-n}  &\mbox{if } x \neq 0^n \mbox{ and } z \neq 0^n \mbox{ and } x \cdot z \mbox{ is even}.
	\end{cases}
	\end{eqnarray}
	We can count the number of each subset of binary vectors $x$ and $z$ to find the multiplicities.
\end{proof}

\subsection{Lower bounds for the modular adder}
\label{sec:adder-bound}

The adder circuit is one of the most fundamental quantum arithmetic operations, which implements addition on a pair of registers in superposition.
We can lower bound the resources required to implement it as follows:\footnote{After the first posting of this paper, Craig Gidney \cite{GidneyBlog2019} showed that the state $\ket{C^{n} Z}$ can be produced using the $n$-qubit modular adder. 
We reproduce his argument in \app{Gidney-bound-reduction} for completeness. 
Using our (slightly stronger) bounds for $\ket{C^{n-1} Z}$ the adder circuit cannot be implemented with fewer than $n+1$ copies of $\ket{T}$, or $(n+1)/2$ copies of $\ket{CS}$, or $(n+1)/3$ copies of $\ket{CCZ}$.}

\begin{prop} \label{prp:AdderLowerBound}
An adder circuit on two $n$-qubit registers acts on basis states as 
\[
A(\ket{i}\ket{j}) = \ket{i}\ket{i + j}
\]
with $i+j$ evaluated modulo $2^n$. It is not possible to implement the adder circuit with Clifford gates and measurements using fewer than $n-2$ $\ket{T}$ states, $(n-2)/2$ $\ket{CS}$ states or $(n-2)/3$ $\ket{CCZ}$ states.
\end{prop}

\begin{proof}
The proof proceeds in two steps. 
First we show that the adder circuit $A$ acting on the $n$-qubit quantum Fourier state $\ket{QFT^b_n}$ (defined below) and the stabilizer state $\ket{+}^{\otimes n}$ has the action $A(\ket{+}^{\otimes n} \ket{QFT^b_n}) = \ket{QFT^{-b}_n}  \ket{QFT^b_n}$.
This tells us that if $A$ is implemented by a set of Clifford gates and Pauli measurements along with some input resource state $\ket{\psi}$, it must be that $\nu(\ket{\psi} \ket{QFT^b_n}) \geq \nu(\ket{QFT^{-b}_n}  \ket{QFT^b_n})$, and hence $\nu(\ket{\psi}) \geq \nu(\ket{QFT^{-b}_n})$ by the additive property of the stabilizer nullity.
Second we show that $\nu(\ket{QFT^{-1}_n}) = n-2$, which then directly implies our bounds since if the bounds are not satisfied, $\nu(\ket{\psi}) \geq \nu(\ket{QFT^{-1}_n})$ would not be satisfied. 

Given this proof structure, it remains to show that 
\[
A(\ket{+}^{\otimes n}  \ket{QFT^b_n}) = \ket{QFT^{-b}_n} \ket{QFT^b_n},
\] and that $\nu(\ket{QFT^{-1}}) = n-2$.
First we recall the family of quantum Fourier states for each integer $a = 0,1, \dots ,2^n-1$: 
\begin{eqnarray}
	\ket{QFT^{a}_n} = \frac{1}{\sqrt{2^n}}\sum_{y=0}^{2^n-1} \exp\left[ \frac{i 2 \pi a y}{2^n}\right] \ket{y} = \otimes_{k=1}^{n}\frac{\ket{0}+e^{i 2\pi a/ 2^k}\ket{1}}{\sqrt{2}},
	\label{eq:QFTstate}
\end{eqnarray}
where $\ket{y}$ is an $n$-qubit basis state (with $y$ expressed in binary), and note that $\ket{QFT^0_n} = \ket{+}^{\otimes n}$, and $\ket{QFT^{a}_n} = \ket{a+2^n}$.
Consider applying the adder to a pair of such states:
\begin{eqnarray} \label{eq:qft-state}
A \left( \ket{QFT^{a}_n} \ket{QFT^b_n} \right) &=& \frac{1}{2^n}\sum_{y=0}^{2^n-1}\sum_{z=0}^{2^n-1}  \exp\left[ \frac{i 2 \pi (a y+bz)}{2^n}\right] \ket{y} \ket{z+y},\nonumber\\
&=& \frac{1}{2^n}\sum_{y=0}^{2^n-1}\sum_{x=0}^{2^n-1}  \exp\left[ \frac{i 2 \pi (a y+b(x-y))}{2^n}\right] \ket{y} \ket{x},\nonumber\\
&=& \ket{QFT^{a-b}_n} \ket{QFT^b_n}.\nonumber
\end{eqnarray}
By taking $a=0$, we have $A(\ket{+^{\otimes n}} \ket{QFT^b_n}) = \ket{QFT^{-b}_n}  \ket{QFT^b_n}$ as required.

Finally, to calculate the stabilizer nullity of $\ket{QFT^{-b}_n}$, we use the tensor product decomposition in Eq.~(\ref{eq:QFTstate}), and note that $\nu(\ket{QFT^{-b}_n})$ is the sum of the stabilizer nullity for each state in the tensor product.
When $b=1$, from prop.~\ref{prp:AdderLowerBound} the first two states in the tensor decomposition have $\nu=0$, whereas the remaining the $n-2$ states have $\nu=1$, such that $\nu(\ket{QFT^1_n}) = n-2$.
The bounds are then implied from the fact that $\nu(\ket{T}) = 1$, $\nu(\ket{CS}) = 2$ and $\nu(\ket{CCZ}) = 3$.
\end{proof}

% In all of the examples we have seen so far, we focused on the number of $\ket{CCZ}$ states required to produce some other resource state or perform some task. 
% However, it would be more standard to ask how many $\ket{T}$ states are required, and one could in principle be interested in how many $\ket{CS}$ states (or how many copies of any other non-stabilizer state) would be required. 
% To make sense of providing the cost in terms of a variety of standard input resource states, it makes sense to first consider the cost of converting between the input resource states themselves, which we turn to in the following section. 

The calculation of $\nu(\ket{QFT^1_n}) = n-2$ that we performed in the proof above
also implies that the Quantum Fourier Transform on $n$ qubits can not be performed 
using fewer than $n-2$ copies of $\ket{T}$.

\section{Lower bounds for approximate unitary synthesis}
\label{sec:unitary-synthesis}

%\magnote{Something that could be a little confusing in this section is the use of the word "protocol". In the current usage, sometimes a protocol produces a state, and sometimes it applies a unitary. VK: I feel like this is OK, why should protocol refer to something that only applies a unitary or only approximates the state? Usually if somebody wants to emphasise what kind of protocol it is, they put a noun in front of it: Distillation protocol, unitary synthesis protocol, state preparation protocol. }
In this section, we lower bound the number of resource states needed to approximate an arbitrary single-qubit unitary using Clifford gates and Pauli measurements.
Unlike the previously-known lower bounds, our bounds: 
(1) allow for Pauli measurements; 
(2) allow measurement outcomes to affect the subsequent parts of the protocol; and 
(3) do not depend on the number of ancillary qubits used in the protocol.

There are some subtleties to be addressed when analyzing a protocol containing measurements that can affect the operations applied in subsequent parts of the protocol.
In particular, the state the protocol outputs and
the number of resource states it consumes are random variables, which depend on the sequence of measurement outcomes obtained. 
The following definition is convenient for formulating lower bounds in this setting.

% ECs VERSION
% \begin{dfn}
% \label{defn:measurement-dependent-protocols}
% Consider a protocol which consists of alternating Clifford gates, ancilla initialization and measurement, where the measurement outcomes can affect subsequent parts of the protocol. Let $S$ denote the set of possible measurement outcomes, for a particular $m \in S$ we have a postselected quantum circuit 
% \begin{equation}
%     \rho \rightarrow A_m \rho A_m^\dagger = B_m \left( \rho \otimes \kb{\Psi}{\Psi}^{\otimes N(m)} \right) B_m^\dagger ,
% \end{equation}
% where $B_m$ is a Kraus operator composed of stabilizer operators and $N: S \rightarrow \mathbb{Z}$ is a function indicating how many copies of a magic state $\ket{\Psi}$ are used on that instance of the protocol.   We say the protocol approximates a one-qubit unitary $U$ to within precision $\varepsilon$ with probability $p$, if there exists a $S' \subseteq S$ such that for all $\rho$
% \begin{align}
%     \sum_{m \in S'} \mathrm{Tr}[ A_m \rho A_m^\dagger ] & \geq p \\ \nonumber
%   \left| \left|  \frac{\sum_{m \in S'} A_m \rho A_m^\dagger}{\sum_{m \in S'} \mathrm{Tr}[A_m \rho 
%      A_m^\dagger ]} - U \rho U^\dagger  \right| \right|_1  & \leq \epsilon . 
% \end{align}
% Furthermore, it achieves this using at most $\mathcal{N}_{\ket{\Psi}}(U,\varepsilon)$ resources if $\mathcal{N}_{\ket{\Psi}}(U,\varepsilon)\leq N(m)$ for all $m \in S$. 
% \end{dfn}

% VKs VERSION
\begin{dfn}
\label{defn:measurement-dependent-protocols}
Consider a protocol with measurement outcomes that can affect subsequent parts of the protocol.
Fixing a sequence of measurement outcomes in the protocol specifies an associated post-selected quantum circuit.
Every input state to such a protocol defines a probability distribution on the set of all measurement outcomes and on their associated post-selected quantum circuits.
We say that the protocol has some property $P$ with probability at least $p$ if, for all states input to the protocol, a sample drawn from the distribution of post-selected quantum circuits has the property $P$ with probability at least $p$.
\end{dfn}

For example, the property $P$ above could be \textit{the number of $\ket{T}$ states consumed is at least $M$}.
The primary goal of this section is to establish the following result:

\begin{thm} \label{thm:approximation-lower-bound}
Consider a protocol that uses $\mathcal{N}_{\ket{\Psi}}(U,\varepsilon)$ copies of the resource state $\ket{\Psi}$ and stabilizer operations to
approximate a one-qubit unitary $U$ to within precision $\varepsilon$ (measured by the diamond norm). For any positive $C > 1$ and $\varepsilon < 1/(2^8 C)$ there exists a unitary $U$ such that the following inequalities must hold
\begin{align*}
\mathcal{N}_{\ket{T}}(U,\varepsilon) 
& \ge
\frac{1}{6} \log_2\at{1/\varepsilon} - \frac{1}{6} \log_2\at{C} - 1,
\\
\mathcal{N}_{\ket{CCZ}}(U,\varepsilon)
& \ge
\frac{1}{8} \log_2\at{1/\varepsilon} - \frac{1}{8} \log_2\at{C} - \frac{3}{4},
\\
\mathcal{N}_{\ket{CS}}(U,\varepsilon)
& \ge
\frac{1}{6} \log_2\at{1/\varepsilon} - \frac{1}{6} \log_2\at{C} - 1.
\end{align*}
with probability at least $(C-1)/C$.
In particular, this is the case for all unitaries $U$ such that
$2\sqrt{C \varepsilon} \le \abs{\bra{0}U\ket{1}}^2 \le 6 \sqrt{C \varepsilon}$.
\end{thm}
The bounds in \theo{approximation-lower-bound} directly imply related bounds on the average case, such as:
\[
\mathbf{E} \mathcal{N}_{\ket{T}}(U,\varepsilon) 
\ge
\frac{C-1}{C}\at{\frac{1}{6} \log_2\at{1/\varepsilon} - \frac{1}{6} \log_2\at{C} - 1}.
\] 
In the rest of this section we put together the pieces to prove \theo{approximation-lower-bound}.
Our strategy to lower bound the number of resource states needed to approximate unitary $U$ to within diamond-norm precision $\varepsilon$ is to establish a relation between this and lower bounds on approximating the state $U\ket{1}$ to within trace norm $\varepsilon'$. 
Unfortunately, the associated resource requirement divergence is not captured by either the nullity or extent monotones we have discussed as they do not diverge for states approaching $\ket{0}$. 
Our unitary synthesis results do not hold when a catalyst state is allowed, in contrast to those bounds proven with the stabilizer nullity due to its additive property.
We make the relation between approximating the unitary $U$ to within diamond-norm precision $\varepsilon$ and approximating the state $U\ket{1}$ to within trace norm $\varepsilon'$ concrete in \sec{unitary-and-state-approx}, and then prove lower bounds for state approximation using different resource states in \sec{approx-with-CS} and \sec{approx-with-T}.
Before this, we present a theorem which we use to prove our lower bounds apply even with an arbitrary number of additional stabilizer ancillas:
\begin{thm} \label{thm:canonical-form}
Consider a post-selected stabilizer circuit with input $\ket{\psi_{\mathrm{in}}}$ and output $\ket{\psi_{\mathrm{out}}}$, where
$\ket{\psi_{\mathrm{in}}}$ is defined on no fewer qubits than $\ket{\psi_{\mathrm{out}}}$.
Then there exists a set of $k = \nu(\ket{\psi_{\mathrm{in}}})-\nu(\ket{\psi_{\mathrm{out}}})$ independent commuting Pauli operators $P_1,\ldots,P_{k}$ and
a Clifford unitary $C$ such that 
\[
\ket{\psi_{\mathrm{out}}}\otimes\ket{S}
\propto
C M_{P_1} \ldots M_{P_{k}} \ket{\psi_{\mathrm{in}}},
\]
where $\ket{S}$ is a stabilizer state and where $M_{P}$ is the projector on the $+1$ eigenspace of $P$.
\end{thm}
From \theo{canonical-form}, without loss of generality we can assume that there are only commuting measurements in the protocol and no ancillary qubits which simplifies our analysis.
However, note that this canonical form works for post-selected measurements. 
We highlight this theorem here because we expect that it may be of broader application and interest. 
We defer the proof to \app{canonical-form}.

\subsection{Approximate unitary synthesis with and without post-selection}
\label{sec:unitary-and-state-approx}

Our starting point addresses the order of taking averages for a protocol with measurement outcomes that can affect subsequent parts of the protocol. 
In particular, the following lemma shows that a protocol that has an average output density matrix which is close to a desired state also has, on average, an output density matrix which is close to the desired state \textit{on individual runs of the protocol}.

\begin{lem} \label{lem:from-channels-to-post-selection}
Consider a protocol that, when averaged over measurement outcomes, produces a density matrix $\rho$ that has fidelity $\bra{\psi} \rho \ket{\psi}$ at least $1 - \delta$ with a pure state $\ket{\psi}$.
Then, for any $C > 1$, with probability at least $(C-1)/C$ the fidelity between $\ket{\psi}$ and the protocol's output is at least $1- C\delta$ following the convention of \defn{measurement-dependent-protocols}.
\end{lem}

\begin{proof}
Suppose the protocol has $N$ possible sequences of measurement outcomes.
Let $p_k$ be the probability of the $k^\text{th}$ sequence of measurement outcomes occurring, and let
$\rho_k$ be the normalized density matrix of the output register for that sequence.

For fixed $C > 1$ we split the set of all fixed sequences of measurement outcomes into two subsets, $S$ and its complement $\overline{S}$. 
The set $S$ contains sequences that output good approximations of $\ket{\psi}$ such that for $k\in S$, $\bra{\psi}\rho_{k}\ket{\psi} \geq 1 - C \delta$, and $\overline{S}$ contains sequences that output worse approximations, such that for $k \in \overline{S}$, $\bra{\psi}\rho_{k}\ket{\psi} < 1 - C \delta$.
Because the overall average output $\rho$ has fidelity at least $1 - \delta$ with $\ket{\psi}$, the probability $p_S$ of all outcomes leading to a good approximation can not be small.
More explicitly, let $\rho_S$ and $\rho_{\overline{S}}$ be the normalized density matrices corresponding to averaging over the subsets $S$ and $\overline{S}$ respectively:
\[
\rho_S \propto \sum_{k \in S} p_k \rho_k \text{ and } \rho_{\overline{S}} \propto \sum_{k \in \overline{S}} p_k \rho_k.
\]
The density matrix of the output is then $\rho = p_S \rho_S + (1 - p_S) \rho_{\overline{S}}$. 
By construction $\bra{\psi}\rho_{\overline{S}}\ket{\psi} < 1 - C \delta$, therefore 
\[
1 - \delta
\le
\bra{\psi}\rho\ket{\psi}
= 
p_S \bra{\psi} \rho_S \ket{\psi} + (1-p_S) \bra{\psi} \rho_{\overline{S}} \ket{\psi}
\le
p_S + (1-p_S)(1-C \delta).
\]
By solving the inequality $1 - \delta \le p_S + (1-p_S)(1-C \delta)$ we derive the required lower bound on $p_S$.
\end{proof}
Thus far we have used fidelity to compare a state and its approximation, but we wish to deduce something about the diamond norm distance between channels. 
We can give bounds in both directions between the trace distance and the fidelity $\sqrt{\bra{\psi} \rho \ket{\psi}}$ using the Fuchs–van de Graaf inequalities:
\begin{align} \label{eq:fuchs-van-de-graaf}
\sqrt{\bra{\psi} \rho \ket{\psi}} & \ge 1 - \frac{1}{2}\nrm{ \ket{\psi}\bra{\psi} - \rho }_1, \\
\quad
\label{eq:fuchs-van-de-graaf-snd}
\nrm{ \ket{\psi}\bra{\psi} - \rho}_1 & \le 2\sqrt{ 1 - \bra{\psi} \rho \ket{\psi} }.
\end{align}

From the second of these inequalities and from \lemm{from-channels-to-post-selection}, the following is implied:
\textit{Consider a protocol which, when averaged over measurement outcomes, produces a density matrix $\rho$ that has fidelity at least $1 - \delta$ with a pure state $\ket{\psi}$.
Then, for any $C > 1$, with probability at least $(C-1)/C$ the trace distance between $\ket{\psi}$ and the protocol's output is at most $2\sqrt{C \delta}$.}
Note that the square root is necessary, as exemplified by randomized protocols~\cite{hastings2016turning,CampbellRandom17,childs2018faster,campbell2018random}. A corollary of these protocols is approximate state preparation protocols that achieve trace distance $\sim \delta$ by randomly choosing between different deterministic state preparation procedures, each with trace distance $\sim \sqrt{\delta}$.

The next lemma establishes connection between the lower bounds for state preparation protocols with post-selection and lower bounds for non-post-selected protocols for approximating unitaries.

\begin{lem} \label{lem:from-state-to-unitary}
Consider a protocol that uses $\mathcal{N}_{\ket{\Psi}}(U,\varepsilon)$ copies of the resource state $\ket{\Psi}$ and stabilizer operations to
approximate a one-qubit unitary $U$ to within precision $\varepsilon$ (measured by the diamond norm). 
For any $C > 1$,
let $N$ be the minimum number of copies of a resource state $\ket{\Psi}$ needed to approximate the state $\ket{\psi}=U\ket{1}$ to trace distance $2\sqrt{C \varepsilon}$ with any protocol composed of stabilizer operations and post-selection. 
Then $\mathcal{N}_{\ket{\Psi}}(U,\varepsilon) \ge N$ with probability at least $(C-1)/C$, following the convention of \defn{measurement-dependent-protocols}.
\end{lem}

\begin{proof}
Given a protocol that uses $\mathcal{N}_{\ket{\Psi}}(U,\varepsilon)$ copies of $\ket{\Psi}$ to
approximate $U$ to diamond-norm precision $\varepsilon$,
we could approximate the state $\ket{\psi} = U\ket{1}$ to within trace distance $\varepsilon$ with $\mathcal{N}_{\ket{\Psi}}(U,\varepsilon)$ copies of $\ket{\Psi}$.
By Fuchs-van de Graaf inequality~\eq{fuchs-van-de-graaf}, 
our protocol approximates $\ket{\psi}$ with fidelity at least $1- \varepsilon/2$.
We now have a statement regarding the fidelity of the protocol, averaged over all the protocol's possible measurement sequences, and we wish to connect this to post-selected protocols.
By direct application of \lemm{from-channels-to-post-selection}, the fidelity between the output of this protocol and $\ket{\psi}$ is at least $1 - C \varepsilon$ with probability at least $(C-1)/C$, following the convention of \defn{measurement-dependent-protocols}.
Finally, by Fuchs-van de Graaf inequality~\eq{fuchs-van-de-graaf-snd}, 
the output density matrix $\rho$ is within trace distance $ 2 \sqrt{C \varepsilon}$  with probability at least $(C-1)/C$.
Therefore $\mathcal{N}_{\ket{\Psi}}(U,\varepsilon) \ge N$ with probability at least $(C-1)/C$.
\end{proof}

In the next sub-sections we establish lower bounds on the number of $\ket{T}$ and $\ket{CS}$ states needed to approximate one qubit states when using post-selected stabilizer operations.
We first establish the lower bounds involving $\ket{CS}$ because it is simpler and
illustrates main ideas used for the lower bound in terms of $\ket{T}$ states.

\subsection{Lower bounds with \texorpdfstring{$\ket{CS}$}{CS} and \texorpdfstring{$\ket{CCZ}$}{CCZ} resource states}
\label{sec:approx-with-CS}

We start by establishing approximation lower bound using $\ket{CS}$ states because it is the simplest case sufficient to illustrate the main proof techniques.
The aim of this subsection is to prove the following result:

\begin{lem} \label{lem:cs-lower-bound}
Let $N_{\ket{CS}}\at{\ket{\psi},\varepsilon}$ be the minimum number of $\ket{CS}$ resource states required
to approximate the one-qubit state $\ket{\psi}$
to within trace distance $\varepsilon$ using stabilizer operations and post-selection.
When $\varepsilon < 1/8$, there exists a state $\ket{\psi}$ such that  $N_{\ket{CS}}\at{\ket{\psi},\varepsilon} \ge \nicefrac{1}{3}\cdot\log_2(1/\varepsilon)-\nicefrac{2}{3}$. 
For example, this is the case for all states such that
$\varepsilon < \abs{\ip{\psi | 0}}^2 < 3 \varepsilon$.
\end{lem}

% \magnote{Q: It seems confusing that we are using $\varepsilon$ for both diamond norm and trace distance. A: I think it's OK, as long as we don't have $\lim_{\varepsilon \rightarrow \infty} N_{\varepsilon} = 0$ \texttt{;)} If we think this is an issue, I would suggest using $\varepsilon$ and $\epsilon$ to distinguish two of them. }

\begin{proof}
Our proof has two main parts.
Firstly, we note that the existence of a protocol that uses $n$ copies of $\ket{CS}$ to approximately prepare a state $\ket{\psi}$ to within trace distance $\varepsilon$, where the target state satisfies
$\varepsilon < \abs{\ip{\psi|0}}^2 < 3 \varepsilon$,
implies that there must be a set of $k \le 2n$ commuting Pauli operators which, when measured on the input state $\ket{CS}^{\otimes n}$, have a probability of all giving $+1$ outcomes in the interval $(0,4\varepsilon)$. 
Secondly, we observe that the probability of a joint measurement of any $k \le 2n$ commuting Pauli operators
on the input state $\ket{CS}^{\otimes n}$ can either be zero, or must be at least $1/2^{k+n}$.
We then conclude that $4\varepsilon \ge 1/2^{n+k} \ge 1/2^{3n}$ and therefore 
$N_{\ket{CS}}\at{\ket{\psi},\varepsilon} \ge \nicefrac{1}{3}\cdot\log_2(1/(\varepsilon))-\nicefrac{2}{3}$.

% }
% \magnote{old:
% Let $\rho$ be the density matrix of the output qubit.
% The precision requirement is that $\nrm{ \ket{\psi}\bra{\psi} - \rho}_1 \le \varepsilon$.
% By \theo{canonical-form}, we can write the approximate preparation of $\rho$ in terms of a Clifford unitary $C$ and a set of $k - 1 < 2n$ independent commuting Pauli operators $P_1, \ldots, P_{k-1}$ such that}
% \[
% \ket{\phi}\otimes\ket{0}^{\otimes k'}
% \propto
% C M_{P_1} \ldots M_{P_{k-1}} \ket{CS}^{\otimes n},
% \]
% such that $\rho$ is the reduced density matrix of $\ket{\phi}$ on the output qubit.
%Let us construct commuting Pauli operators with the required property. 
Consider $\ket{\psi}$ such that $\varepsilon < \abs{\ip{\psi|0}}^2 < 3 \varepsilon$ and assume that 
the first qubit is the output qubit of the protocol.
Let $\rho$ be the density matrix of the output qubit.
By \theo{canonical-form}, we can write the approximate preparation of $\rho$ in terms of a Clifford unitary $C$ and a set of $k - 1 = \nu(\ket{CS}^{\otimes n}) - \nu(\ket{\Psi_{\text{out}}})$ independent commuting Pauli operators $P_1, \ldots, P_{k-1}$. 
Let us define $P_k = C^\dagger Z_1 C$ and show that $P_k$ commutes with $P_1,\ldots,P_{k-1}$.
Recall that if $P_k$ anti-commutes with one of $P_1,\ldots,P_{k-1}$, this implies that $p' =1/2$, where $p'=\mathrm{Tr}(\ket{0}\bra{0} \rho)$ is the probability of getting a $+1$ measurement of $Z_1$.
Next we estimate this probability based on the precision requirement $\nrm{ \ket{\psi}\bra{\psi} - \rho}_1 \le \varepsilon$.
Note that $p'$ satisfies the inequality: 
\[
\abs{ |\ip{{0}|{\psi}}|^2 - p' }
=
\abs{\mathrm{Tr}\left(\ket{0}\ip{0|\psi}\bra{\psi} \right) - \mathrm{Tr}(\ket{0}\bra{0} \rho) }
\le
\nrm{ \ket{\psi}\bra{\psi} - \rho}_1 \le \varepsilon,
\]
where we have used the inequality $\abs{\mathrm{Tr} AB } \le \nrm{A}_{\infty} \nrm{B}_1$, and that $\nrm{\ket{0}\bra{0}}_{\infty}=1$.
This implies that the probability $p'$ must belong to the interval $(0, 4\varepsilon)$. 
The condition $\varepsilon < 1/8$ implies that $p' \in (0,1/2)$ and therefore $P_k$ must commute with $P_1,\ldots, P_{k-1}$. 
Next we show that $k \le 2n$, by showing that $\nu(\ket{\Psi_{\text{out}}}) \ge 1$.
If $\nu(\ket{\Psi_{\text{out}}}) = 0$ this means that the output state in a stabilizer state and this would imply that probability 
of measuring $\ket{0}$ on output qubit must be $0,1$ or $1/2$ which is ruled out by our estimate $p' \in (0,4\varepsilon)$. 
The joint probability of measuring $P_1,\ldots,P_k$ is non-zero and less than the conditional probability $p'$ and therefore also belongs to interval $(0,4 \varepsilon)$, as required.

Next we show that if the joint probability of measuring any $k$ commuting Pauli operators $P_1,\ldots,P_{k}$ is non-zero, then it must be at least $1/2^{n+k}$.
Consider
\[
\bra{CS}^{\otimes n} \prod_{j=1}^k \frac{(I+P_{j})}{2} \ket{CS}^{\otimes n} = \frac{1}{2^{k}}\sum_{P \in \ip{P_1,\ldots,P_{k} }} \bra{CS}^{\otimes n} P \ket{CS}^{\otimes n}.
\]
The Pauli expectations of $\ket{CS}$ can only be $0$, $1$ or $\pm 1/2$.
Therefore, the value of the expression above can always be written as $a/2^{k+n}$ for some non-negative integer $a$ and its smallest non-zero value is $1/2^{k+n}$. 
\end{proof}

The key to generalizing the above result from $\ket{CS}$ states to an arbitrary $k$-qubit resource state $\ket{\Psi}$ is to establish a lower bound on the quantity:
\begin{equation} \label{eq:expectation}
 \frac{1}{2^{m}}\sum_{P \in \ip{P_1,\ldots,P_{m} }} \bra{\Psi}^{\otimes n} P \ket{\Psi}^{\otimes n},
\end{equation}
where $\set{P_1,\ldots,P_{m}}$ are independent commuting Pauli operators and $m \le k\cdot n$. 
Note that replacing $p$ with one in the statement of the lemma leads to a slightly weaker lower bound that does not require the knowledge of $p$.
For example, it is not too difficult to generalize the above result to use $\ket{CCZ}$ states in place of $\ket{CS}$ states, because their Pauli expectations also take values $0$, $1$ and $\pm 1/2$. The resulting lemma is

\begin{lem} \label{lem:ccz-lower-bound}
Let $N_{\ket{CCZ}}\at{\ket{\psi},\varepsilon}$ be the minimum number of $\ket{CCZ}$ resource states required to approximate the one-qubit state $\ket{\psi}$
to within trace distance $\varepsilon$ using stabilizer operations and post-selection probability $p$.
When $\varepsilon < 1/8$,
there exists a state $\ket{\psi}$ such that 
$N_{\ket{CCZ}}\at{\ket{\psi},\varepsilon} \ge \nicefrac{1}{4}\cdot\log_2(1/\varepsilon)-\nicefrac{1}{2}$. 
\end{lem}

\subsection{Lower bounds with \texorpdfstring{$\ket{T}$}{T} resource states} 
\label{sec:approx-with-T}

The goal of this subsection is to establish the lower bound on the probability of a sequence of measurements of $k$ independent commuting Pauli operators on input state $\ket{T}^{\otimes n}$ for $k \le n$ and then find the lower bound on the number of $\ket{T}$ states needed to approximate single a qubit unitary.
The following result is the missing piece needed to generalize \lemm{cs-lower-bound}.
\begin{prop} \label{prop:commuting-pauli-prob-on-t-states}
Let $\set{P_1,\ldots,P_{k}}$ be independent commuting Pauli operators and
let the probability of measuring the $+1$ eigenvalue of each be
\begin{equation} \label{eq:ts-expectations}
 p = \frac{1}{2^{k}}\sum_{P \in \ip{P_1,\ldots,P_{k} }} \bra{T}^{\otimes n} P \ket{T}^{\otimes n}.
\end{equation}
If the value of $p$ is non-zero, then  $ p \ge \frac{1}{2^{2k+n}}$.
\end{prop}

Before proceeding we need to introduce several concepts we are going to use in the proof \cite{cohen2007number}.
Consider the following set:
\[
\mathcal{R} = \set{ \frac{a+bi+\sqrt{2}(c+di)}{2^j} : \text{ for } a,b,c,d,j \text{ integers} }.
\]
Note that the set $\mathcal{R}$ is closed under addition, negation and multiplication, and contains $0$ and $1$.
Thus, the set $\mathcal{R}$ is an example of a ring. 
Also note that the set $\mathcal{R}$ is closed under complex conjugation.

Note that the state $\ket{T}$ can can be written as a vector with entries in $\mathcal{R}$ as $(\sqrt{2}/2,(1+i)/2)$. 
Similarly, all Pauli operators can be written as matrices with entries in $\mathcal{R}$.
For this reason, $p$ defined in Equation~\eq{ts-expectations} also belongs to $\mathcal{R}$. 
Moreover, as a real number, we can write $p=(a_p+c_p\sqrt{2})/2^k$ for some integers $a_p,c_p,k$. 
We cannot directly use the approach of lower bounding $p$ directly that we used in Sec.~\ref{sec:approx-with-CS}, because $\sqrt{2}$ is an irrational number and $a_p + c_p\sqrt{2}$ can be made arbitrary small.  To address this new complication, we use the bullet map that preserves $\mathcal{R}$ and is similar to complex conjugation:
\[
\at{\frac{a+bi+\sqrt{2}(c+di)}{2^k}}^{\bullet} = \at{\frac{a+bi-\sqrt{2}(c+di)}{2^k}}.
\]
One can directly check that for arbitrary elements of $r_1$ and $r_2$ of $\mathcal{R}$, the following holds:
\begin{align}
\at{r_1+r_2}^{\bullet} & = r_1^\bullet + r_2^\bullet ,\label{eq:addition}\\
\at{r_1\cdot r_2}^{\bullet} & = r_1^\bullet \cdot r_2^\bullet, \label{eq:multiplication} \\ 
(r_1 ^\bullet)^{\ast} & = (r_1 ^\ast)^{\bullet}.\label{eq:conjugation}
\end{align}
In addition, the map $(\cdot)^\bullet$ helps us convert numbers of the form $(a+c\sqrt{2})/2^k$ into numbers of the form $d/2^k$ because:
\begin{equation}\label{eq:automorphism-properties}
(a+c\sqrt{2})(a+c\sqrt{2})^\bullet = a^2 - 2 c^2
\end{equation}
Now we are ready to prove the proposition:
\begin{proof}[Proof of \propos{commuting-pauli-prob-on-t-states}]
We will show that if $p$ is non-zero, then $p^\bullet$ belongs to the interval $(0,1]$ and $p p^\bullet$ is a non-negative number of the form $n_p/2^{2k+n}$ for some integer $n_p$.
This implies that the smallest non-zero value of $p = (n_p/2^{2k+n})/p^\bullet$ is at least $1/2^{2k+n}$.

First note that the Pauli expectations of $\ket{T}$ can only be $0$, $1$ or $\pm 1/\sqrt{2}$.
For this reason, $p$ must be a number of the form $(a_p+c_p\sqrt{2})\sqrt{2}^n/2^{k}$.
Using \eq{addition}, \eq{multiplication}, \eq{conjugation} and \eq{automorphism-properties} we see that: 
\[
p^\bullet = \frac{1}{2^{n}}\sum_{P \in \ip{P_1,\ldots,P_{n} }} \bra{T^\bullet}^{\otimes n} P \ket{T^\bullet}^{\otimes n} \text{ where } \ket{T^\bullet} =  (-\sqrt{2}/2,(1+i)/2).
\]
Therefore $p^\bullet$ is the probability of measuring a projector on the state $\ket{T^\bullet}^{\otimes n} $ and must be less or equal to one.
By definition of $(\cdot)^\bullet$, $p^\bullet$ can be zero if and only if $p$ is zero. We conclude that $p^\bullet$ belongs to the interval $(0,1]$ as required. 

Finally let us compute 
\[
pp^{\bullet} = (a_p^2 - 2 c_p^2)\at{-1}^n / 2^{2k+n} = n_p / 2^{2k+n} \text{ for some integer } n_p,
\]
as required.
\end{proof}
Using the same techniques as in the proof of Lemma~\ref{lem:cs-lower-bound} we get the following result:

\begin{lem} \label{lem:t-lower-bound}
Let $N_{\ket{T}}\at{\ket{\psi},\varepsilon}$ be the minimum number of $\ket{T}$ resource states required to approximate the one-qubit state $\ket{\psi}$
to within trace distance $\varepsilon$ using stabilizer operations.
When $\varepsilon < 1/8$,
there exists a state $\ket{\psi}$ such that 
$N_{\ket{T}}\at{\ket{\psi},\varepsilon} \ge \nicefrac{1}{3}\cdot\log_2(1/\varepsilon)-\nicefrac{2}{3}$. 
For example, this is the case for all states such that
$\varepsilon < \abs{\ip{\psi | 0}}^2 < 3 \varepsilon$.
\end{lem}

We omit the proof here because it is very similar to the proof of Lemma~\ref{lem:cs-lower-bound}.
These can be generalized further to include states like $|\sqrt{T}\rangle$, $|\sqrt{T}^3\rangle$ as shown in \theo{root-t-lower-bound} in the Appendix and
other roots of $T$ using methods described in \app{general-approx-bounds} using the dyadic monotone introduced in the next section.
%The generalization requires more general version of the techniques used in this subsection and we will develop them in the appendix.

\begin{proof}[Proof of \theo{approximation-lower-bound}]
First note that setting $p=1$ on the right hand side of the inequalities in
\lemm{cs-lower-bound}, \lemm{ccz-lower-bound} and \lemm{t-lower-bound} form
new (weaker) inequalities which hold for all $p$.
Then apply \lemm{from-state-to-unitary} to each of these inequalities.
\end{proof}

\section{Tighter lower bounds with measurement probabilities one half}
\label{sec:prob-half-bounds}
The goal of this section is to introduce a quantity similar to the stabilizer nullity $\nu(\ket{\psi})$ that lets us establish stronger lower bounds on the number of resource states needed for certain tasks.
The drawback is that these tighter bounds are not for completely arbitrary sequences of Clifford gates and Pauli measurements, but only those in which each measurement outcome occurs with probability half.
However, as so many of the known circuits are of this class, we foresee these bounds being of interest and expect them to encourage researchers to turn to more rich classes of circuits to evade them.
In what follows, we first show that the well-known circuit~\cite{Jones2012} to implement the multiply-controlled-Z gate using $\ket{CCZ}$ states is optimal with probability half measurements. 
We then show that the best-known circuit for the modular adder~\cite{Gidney2018} using $\ket{CCZ}$ states with probability half measurements uses the number of $\ket{CCZ}$ states that differs by one from the lower bound.
 
\subsection{Lower bound with
\texorpdfstring{$CCZ$}{doubly-controlled-Z} gates for \texorpdfstring{$C^n Z$}{multiply-controlled-Z} gate} 
\label{sec:prob-half-ccz-bounds}

Consider quantum states which, when written in the computational basis, have entries in the following set: \[
\z\of{i,1/2} = \set{ \frac{a+ib}{2^k} : a,b,k \in \z }.
\]
Indeed, $\ket{C^n Z}$ can we written as vectors with entries in the above set.
Note that the set $\z\of{i,1/2}$ is a ring since it is closed under addition, negation, multiplication, and contains $0$ and $1$.

Observe that if a state $\ket{\psi}$ has entries in $\z\of{i,1/2}$
then for any Hermitian multi-qubit Pauli operator $P$, the expectation $\bra{\psi} P \ket{\psi}$ can be written as $a/2^k$ for integers $a,k$.
The expectation is in $\z\of{i,1/2}$ because the entries of the Pauli matrices are in $\z\of{i,1/2}$ and $\z\of{i,1/2}$ is closed under complex conjugation.  
The expectation is also a real number and all the real numbers in $\z\of{i,1/2}$ are of the form $a/2^k$ for integers $a,k$. 
Note that for stabilizer states Pauli expectations can only be $\pm 1$ and $0$. 
Roughly speaking, the power of $2$ in the denominator of the Pauli expectation lets us capture how non-stabilizer the state is.
Next we develop this intuition more rigorously.

First we need a more rigorous way to talk about the power of $2$ in the denominator.
Let $q$ be a non-zero rational number.
It can be written as a product of integer powers of prime numbers in a unique way:
\[
 q = \pm 2^k \cdot p_1 ^{k(1)} \cdots p_m^{k(m)}, \,p_k \text{ are odd primes,}\,k,k(1),\ldots,k(m)\text{ are integers}
\]
% Remark: the uniqueness of such decomposition follows from the definition of a prime number
% Remark: v_2 is called a 2-adic valuation of q
Let us define $v_2(q)$ to be $k$.
Note that function $v_2$ is somewhat similar to $\log\abs{\cdot}$ in that $v_2\at{q_1 q_2}=v_2\at{q_1}+v_2\at{q_2}$, $v_2\at{\pm1} = 0$ and $v_2\at{q} = v_2\at{-q}$.
For odd integer $a$ and integer $k$ the value is $v_2\at{a/2^k}=-k$.
Note also that $v_2$ is always non-negative for integer arguments.
It is convenient to extend $v_2$ to all rational numbers, by defining $v_2\at{0} = +\infty$.
Note that with this extension the multiplicative property still holds. 
Now we are ready to define the quantity of interest. 
\begin{dfn}[Dyadic monotone]
Let $\ket{\psi}$ be an $n$-qubit state with entries in $\z\of{i,1/2}$, the \emph{dyadic monotone} is
\[
\mu_2\ket{\psi} = \max\set{ -v_2\at{\bra{\psi}P\ket{\psi}} : P \in \set{I,X,Y,Z}^{\otimes n} }.
\]
\end{dfn}

The dyadic monotone is essentially the maximum power of two in the denominator over the Pauli spectrum (the set of all Pauli expectations).
It is invariant under Clifford unitaries because they map the set of all multi-qubit Pauli matrices to the set of all Pauli matrices up to a sign and $v_2$ is insensitive to the sign of its argument. 
In addition, Clifford unitaries map states with entries in $\z\of{i,1/2}$ to states with entries in $\z\of{i,1/2}$, because all Clifford unitaries can be written as matrices with entries in $\z\of{i,1/2}$, up to a global phase.

Similarly to the stabilizer nullity $\nu$, the dyadic monotone $\mu_2$ behaves nicely under taking tensor products.
\begin{prop}
Let $\ket{\phi}$ and $\ket{\psi}$ be states with entries in $\z\of{i,1/2}$, then 
\[
\mu_2\at{\ket{\phi}\otimes\ket{\psi}} = \mu_2{\ket{\phi}} + \mu_2{\ket{\psi}}.
\]
\end{prop}
\begin{proof}
The result follows from the fact that for Pauli matrices $P$ and $Q$ such that the expectations $\bra{\phi}P\ket{\phi}$ and $\bra{\psi}Q\ket{\psi}$ are non-zero it is the case that:
\[
v_2\at{\bra{\phi}\otimes\bra{\psi}\at{P \otimes Q}\ket{\phi}\otimes\ket{\psi}}
=
v_2\at{\bra{\phi}P\ket{\phi}}
+
v_2\at{\bra{\psi}Q\ket{\psi}}.
\]
\end{proof}

Another important property is that the dyadic monotone is minimal for stabilizer states:
\begin{prop} \label{prop:dyadic-non-negativity}
Let $\ket{\phi}$ be a state $\z\of{i,1/2}$, then $\mu_2\ket{\psi} \ge 0$, with equality achieved if and only if $\ket{\psi}$ is a stabilizer state.
\end{prop}
\begin{proof}
Consider a non-zero Pauli expectation $\bra{\psi}P\ket{\psi}$ and write it as $a/2^k$ for some odd integer $a$.
Note that $k$ must be non-negative because $\abs{\bra{\psi}P\ket{\psi}} \le 1$.
This shows that  $\mu_2\ket{\psi} \ge 0$. For stabilizer states, the only non-zero expectations can be $\pm 1$ and therefor $\mu_2$ is zero.
It remains to show that $\mu_2(\ket{\psi})=0$ implies that $\ket{\psi}$ is stabilizer state.
First note that $\mu_2(\ket{\psi})=0$ implies that all non-zero Pauli expectations are odd integers. 
Together with the condition $\abs{\bra{\psi}P\ket{\psi}} \le 1$ this implies that the expectations can only be $\pm 1$, 
in other words either $P$ or $-P$ is in $\mathrm{Stab}\ket{\psi}$.
Suppose that $\ket{\psi}$ is an $n$-qubit state and let us compute the size of $\mathrm{Stab}\ket{\psi}$. 
Note that the set $\set{I,X,Y,Z}^{\otimes n}$ is an orthogonal basis of the space of matrices with respect to the inner product $\ip{A,B} = \mathrm{Tr} A B^\dagger$.
The norm squared of the density matrix $\ket{\psi}\bra{\psi}$ is given by the following expression: 
\[
1
=
\ip{\psi | \psi}^2
=
\frac{1}{2^n}
\sum_{P \in \set{I,X,Y,Z}^{\otimes n} }
\abs{
  \mathrm{Tr}\at{\ket{\psi}\bra{\psi}P}
} ^2, 
\]
which implies that the size of $\mathrm{Stab}\ket{\psi}$ is $2^n$ and therefore that $\ket{\psi}$ is a stabilizer state.
\end{proof}

Now we show that Pauli measurements with probability half take states with entries in $\z\of{i,1/2}$ to states with entries in $\z\of{i,1/2}$ (allowing the dyadic monotone to be evaluated). 
Such measurements are used in magic-state injection protocols and play an important role in reducing state preparation using non-Clifford gates to state preparation using resource states, Clifford unitaries and Pauli measurements. 
Measuring a $\pm 1$ eigenvalue of a Pauli observable $P$ with probability $1/2$ is equivalent to multiplying the state by the matrix $(I\pm P)/\sqrt{2}$ which is equal to $(1+i)(I\pm P)/2$ up to a global phase.
The matrix $(1+i)(I\pm P)/2$ has entries in the ring $\z\of{i,1/2}$ and therefore the resulting state will also have entries in $\z\of{i,1/2}$. 

Next we show that $\mu_2$ is non-increasing under these measurements.
To do this, we need another property of the function $v_2$ given by the following proposition:
\begin{prop} For arbitrary rational numbers $a,b$ the following inequality holds
\begin{equation} \label{eq:ultrametric}
v_2(a_1+a_2) \ge \min(v_2(a_1),v_2(a_2)).
\end{equation}
\end{prop}
\begin{proof}
Let us first prove the inequality for non-zero $a_1, a_2$. 
Rewrite $a_j = 2^{k_j} p_j / q_j$ for integer $k_j$ and odd integers $p_j$ and $q_j$ such that
\[
a_1 + a_2 = 2^{\min(k_1,k_2)} \at{ 2^{k_1 - \min(k_1,k_2)} p_1 q_2 + 2^{k_2 - \min(k_1,k_2)} p_2 q_1 } / q_1 q_2.  
\]
Since $q_1$ and $q_2$ are odd, $v_2(a_1+a_2)$ is equal to 
\[
\min(k_1,k_2) + v_2\at{2^{k_1 - \min(k_1,k_2)} p_1 q_2 + 2^{k_2 - \min(k_1,k_2)} p_2 q_1}
\]
by the multiplicative property of $v_2$.
Since $2^{k_1 - \min(k_1,k_2)} p_1 q_2 + 2^{k_2 - \min(k_1,k_2)} p_2 q_1$ is an integer, its value of $\nu_2$ is non-negative. 
The case when at least one of $a_j$ is zero follows from the fact $\min(x,+\infty) = x$.
This concludes the proof of the inequality.
\end{proof}
Now we are ready to prove desired result:
\begin{prop} \label{prop:measure-half}
Let $\ket{\psi}$ be a state with entries in $\z\of{i,1/2}$,
let $P$ be a Pauli observable such that measuring its eigenvalue $+ 1$ has probability $1/2$
and let $\ket{\psi_+}$ be the normalized result of that measurement. 
Then $\mu_2 \ket{\psi} \ge \mu_2 \ket{\psi_+}$.
\end{prop}
\begin{proof}
Let us bound the value of $v_2$ for some Pauli operator $Q$ evaluated on the expectation $\bra{\psi_+} Q \ket{\psi_+}$.
The normalized state is $\ket{\psi_+} = \frac{(I+P)}{\sqrt{2}}\ket{\psi}$.
The expectation of $Q$ is therefore equal to: 
\[
\bra{\psi_+} Q \ket{\psi_+} = \bra{\psi} (I+P) Q (I+P) \ket{\psi}/2.
\]
If $P$ and $Q$ anti-commute, the expectation is zero and does not contribute to the calculation of $\mu_2$.
When $P$ and $Q$ commute, the expectation is equal to $\bra{\psi}Q\ket{\psi} + \bra{\psi}PQ\ket{\psi}$.
Next we use inequality $v_2(a+b) \ge \min(v_2(a),v_2(b))$, to see that: 
\[
v_2\at{\bra{\psi}Q\ket{\psi} + \bra{\psi}PQ\ket{\psi}} \ge  \min\at{\bra{\psi}Q\ket{\psi},\bra{\psi}PQ\ket{\psi}} \ge -\mu_2 \ket\psi.
\]
We have upper-bounded $-v_2\at{\bra{\psi_+} Q \ket{\psi_+}}$ by $\mu_2\ket\psi$ as required.
\end{proof}

Now we use these techniques to show the optimality of the well-known circuit \cite{Jones2012} to implement the multiply-controlled-Z gate 
using stabilizer operations with measurement probablities half and $\ket{CCZ}$ magic states.

\begin{lem}
At least $n-2$ $\ketsm{CCZ}$ states are needed to implement the $n$-qubit multiply controlled $Z$ gate $C^{n-1} Z$ by using stabilizer operations with measurement probabilities one half.
The optimal circuit follows from the construction for multiply-controlled unitaries described in~\cite{Jones2012}.
\end{lem}
\begin{proof}
The circuit for $C^{n-1} Z$ that follows from~\cite{Jones2012} uses $n-2$ $CCZ$ gates. 
By applying that circuit to $\ket{+}^{\otimes n}$ we can prepare $\ketsm{C^{n-1} Z}$.
If there existed a circuit that used $k$ $CCZ$ gates for $k<n-2$, we would be able to prepare states $\ketsm{C^{n-1} Z}$ starting from $k$ $\ket{CCZ}$ states and then using Clifford unitaries and Pauli observable measurements with probability half.
Let us show that this is impossible.
Indeed for the input state we would have value $\mu_2 \atsm{\ketsm{CCZ}^{\otimes k}} = k$.
For the output state we would have $\mu_2 \atsm{\ketsm{C^{n-1} Z}} = n-2$.
This follows from the calculation of Pauli spectrum of $\ketsm{C^n Z}$ in \propos{CnZPauliSpectrum}.
We have shown above that $\mu_2$ is non-increasing when we apply Clifford unitaries and measurements with probability $1/2$, 
therefore $k \ge n-2$ which concludes the proof.
\end{proof}

\subsection{Lower bounds for the modular adder}
\label{sec:prob-half-adder-bounds}

To establish lower bounds for adder circuits we will use the fact that adder can create a complex conjugate copy of a Fourier state.
Our strategy is to generalize $\mu_2$ to be defined on a wider set of states including Fourier states.
This is achieved by extending the domain of $v_2$ to a wider set of values. 
We postpone all the details of the construction of the generalization of $v_2$ to \app{number-theory}. 
Instead we list and discuss all the properties of $v_2$ needed for the lower-bound proof and prove the lower bound for the adder using them.
The properties are then proved in the appendix.

In the previous section, to establish the lower bounds we needed to define rings over which we can write coordinates of $\ket{C^n Z}$ states. 
We will need to define the rings we can use to write down the coordinates of Fourier states. 
We extend the domain of $\mu_2$ to the union of the following family of sets:
\[
\mathcal{R}_d = \z\of{\exp(i\pi/2^d),1/2} = \set{ \frac{1}{2^k} \sum_{j=0}^{2^d-1} a_j \exp(i\pi j/2^d) : \text{ where } a_j, k \text{ are integers} }. 
\]
Note that each of the sets $\mathcal{R}_d$ is closed under addition, negation, multiplication and therefore each of $\mathcal{R}_d$ is an example of a ring. 
In addition, ring $\mathcal{R}_d$ is closed under taking complex conjugate.
Note also that $\mathcal{R}_1$ is exactly the ring $\z\of{i,1/2}$ and $\mathcal{R}_{d} \subset \mathcal{R}_{d+1}$ for all positive $d$.

After we defined the rings, we extend the domain of function $v_2$ so it is defined on values of Pauli expectations of Fourier states. 
For this reason, $v_2$ must be defined at least on the real subsets of $\mathcal{R}_d$.
The proof of the lower bound for multiply-controlled-Z gate relied on additivity for a tensor product of states and monotonicity under measurements with probability $1/2$ of dyadic monotone $\mu_2$.
In turn, our proofs of the mentioned properties of dyadic monotone $\mu_2$ relied on the following two properties of $v_2$:
\begin{itemize}
  \item $v_2\at{a\cdot b} = v_2\at{a} + v_2 \at{b}$
  \item $v_2(a+b) \ge \min\at{v_2(a),v_2(b)}$
\end{itemize}
Above properties also hold for our extension of $v_2$.
We will also need to know some explicit values of $v_2$ to compute $\mu_2$ for Fourier states:
\begin{equation}
\text{For all odd integers } k, \text{ integers } d \ge 2 : v_2\at{ \cos(\pi k /2^d) } = v_2\at{ \sin(\pi k /2^d) } = \frac{1}{2^{d-1}} - 1.
\end{equation}
For example, using above we see that $\mu_2{\ket{T}} = 1/2$ because $v_2(1/\sqrt{2}) = -1/2$.
We can immediately conclude that $C^n Z$ gate requires at least $2(n-2)$ $\ket{T}$ states. 
Next we proceed to calculate $\mu_2$ for Fourier states:
\begin{prop} Consider Fourier state
\begin{eqnarray*}
	\ket{QFT^a_n} = \sum_{y=0}^{2^n-1} \exp\left[ \frac{i 2 \pi a y}{2^n}\right] \ket{y} = \otimes_{k=1}^{n}\left(\ket{0}+e^{i 2\pi a/ 2^k}\ket{1}\right),
\end{eqnarray*}
For all odd $a$, $\mu_2 \ket{QFT^a_n} = n - 3 + (1/2)^{n-2}$.
\end{prop}
\begin{proof}
Recall that Pauli expectations of $(\ket{0}+e^{i 2\pi a/ 2^k}\ket{1})/\sqrt{2}$ are 
\[
\{ 0, \cos\at{2\pi a/ 2^k}, \sin\at{2\pi a/ 2^k} \}.
\]
For this reason, for $k \ge 2 $ we have:
\[
\mu_2\at{(\ket{0}+e^{i 2\pi a/ 2^k}\ket{1})/\sqrt{2}} = v_2\at{\sin( \pi a /2^{k-1}) } = 1 - 1/2^{k-2}
\]
Using multiplicative property of $\mu_2$ we get:
\[
\mu_2\at{\ket{QFT^a_n}} = \sum_{ k = 2 }^{n} \at{ 1 - 1/2^{k-2} }= n - 3 + 1/2^{n-2}
\]
\end{proof}
Above leads to the following lower bound on the number of $\ket{CCZ}$ states needed to implement the modular adder:\footnote{
After the first posting of this paper, Craig Gidney \cite{GidneyBlog2019} showed that the state $\ket{C^{n} Z}$ can be produced using the $n$-qubit modular adder. 
We reproduce his argument in \app{Gidney-bound-reduction} for completeness. 
The requires at least $n-1$ copies of $\ket{CCZ}$ in this setting, which gives a tight lower bound of $n-1$ copies of $CCZ$ to implement the modular adder for a pair of $n$-qubit states.}

\begin{lem}
At least $n-2$ $\ket{CCZ}$ states are needed to implement the $n$-qubit modular adder for $n \ge 3$ by using stabilizer operations with measurement probabilities one half.
\end{lem}
\begin{proof} 
Recall that by applying a circuit for modular adder to $\ket{+}^{\otimes n} \otimes \ket{QFT^{1}_n}$ 
we can create a state  $\ket{QFT^{-1}_n} \otimes \ket{QFT^1_n}$.
If there existed a circuit that used $k$ $CCZ$ gates for $k<n-2$, we would be able to prepare states $\ket{QFT^{-1}_n}$ starting from $k$ $\ket{CCZ}$ and then using Clifford unitaries and Pauli observable measurements with probability half by using $\ket{QFT^1_n}$ as a catalyst.
Let us show that this this impossible.
Indeed for the input state we would have value $\mu_2$ equal to $k + \mu_2 \ket{QFT^1_n} $ and for the output state we would have $\mu_2 \ket{QFT^{-1}_n} + \mu_2 \ket{QFT^1_n}$.
We know that $\mu_2$ is non-increasing when we apply Clifford unitaries and measurements with probability $1/2$, 
therefore $k \ge \mu_2 \ket{QFT^{-1}_n} = n - 3 + (1/2)^{n-2}$ which implies that $k \ge n - 2$.
\end{proof}

The best known~\cite{Gidney2018} modular adder construction uses $n-1$ $\ket{CCZ}$ states, therefore our bound is one $\ket{CCZ}$ state short of the optimum.
Using the same techniques we can derive a lower bound of $2n - 5$ $\ket{T}$ states for $n \ge 3$.
This lower bound multiplicative constant is twice less than the best known construction.

It is also possible to show that the extension of $\mu_2$ to the union of $\mathcal{R}_d$ in non-negative and that its equality to zero implies that its argument is a stabilizer state. 
We defer prove of this fact to \propos{ext-dyadic-non-negativity} in the Appendix.

\subsection{Lower bounds for resource state conversion}
\label{sec:prob-half-conversion-bounds}

In \sec{adder-conversion-protocols} and \app{dyadic-powers-of-t}, we have introduced protocols for catalysis assisted conversion of $\ket{CCZ}$ states into states $\ketsm{ \pi j / 2^d }$. 
We have found that for odd $j$ and integer $d \ge 2$, asymptotically, one can create one $\ketsm{ \pi j / 2^d }$ state at the cost of $1-1/2^{d-1}$ $\ket{CCZ}$ states. 
Using the dyadic monotone we can show that this is optimal when only Pauli measurements with probability 50\% are allowed.

\dyadicpowersbound
\begin{proof}
Let $\ket{\mathrm{cat}}$ be a state used as a catalyst, then $\mu_2$ for the input of our protocol is $\mu_2 \ket{\mathrm{cat}} + k$ and for the output the value of $\mu_2$ is $n(1-1/2^{d-1})+ \mu_2 \ket{\mathrm{cat}}$. 
This is because for odd $j$ and integer $d \ge 2$, $\mu_2 \ketsm{\pi j/2^{d}} = 1 - 1/2^{d-1}$. 
Above implies that $k \ge n(1-1/2^{d-1})$.
\end{proof} 

It is possible to show the monotonicity of $\mu_2$ for a wider range of measurements, namely the Pauli measurements map the state defined over $\mathcal{R}_d$ to the state defined over $\mathcal{R}_d$. 
We provide more details on this in \propos{measure-extended} in the appendix.

\begin{table}[ht]
    \centering
    {\renewcommand{\arraystretch}{1.25}
    \begin{tabular}{c|c|c}
 $\ket{\psi}$  &  \begin{tabular}{@{}c@{}} Best algo. (lower bound) \\ $r n \ket{CCZ} \catalyse{} n \ket{\psi}$ \end{tabular} & \begin{tabular}{@{}c@{}} Best algo. (upper bound) \\ $n \ket{\psi} \catalyse{} r' n \ket{CCZ}$ \end{tabular}\\ \hline
$\ket{\sqrt{T}}$  &
0.75 (0.33333, 0.75$^\dagger$) [\figsm{many-sqrt-t-states}] & 0.0625 (0.207767*) [\tabsm{T_to_Targ}] \\
$\ket{T}$   & 0.5 (0.33333, 0.5$^\dagger$) \cite{gidney2018efficient} & 0.25 (0.275212*) \cite{Jones2012} \\
$\ket{CS} = \ket{W_2}$  & 1 (0.81688*,1$^\dagger$) [\figsm{Wn_level_reduction}, \figsm{ccz-to-cs}] & 0.5 (0.66666)  [\figsm{cs-to-ccz}] \\
$\ket{CCS}$ & 2 (1.24763*,2$^\dagger$) \cite{Jones2012} & 0.25  (1) [\propsm{measure-control}] \\
$\ket{C^3S}$ & 3 (1.33333,3$^\dagger$) \cite{Jones2012} & 0.125 (1.05336*) [\propsm{measure-control}] \\
$\ket{CCZ}$ & 1(1) & 1(1)  \\
$\ket{C^3Z}$ & 2 (1.40942*,2$^\dagger$) \cite{Jones2012} & 0.5 (1.33333) [\propsm{measure-control}]   \\
$\ket{C^4Z}$ & 3 (1.66667,3$^\dagger$) \cite{Jones2012} & 0.25 (1.05336*) [\propsm{measure-control}] \\
$\ket{CCZ_{123,145}}$ & 2 (1.66667,2$^\dagger$) [\figsm{ccz123145}] & 1 (1.20471*) [\figsm{ccz123145}] \\
$\ket{W_3}$ & 1 (1) \cite{HowardCampbell2016}  & 1 (1)  \cite{HowardCampbell2016}  \\
$\ket{W_4}$ & 2.5 (1.3758*,2$^\dagger$) [\tabsm{T_to_Targ}] & 1 (1.33333) [\figsm{Wn_level_reduction}] \\
$\ket{W_5}$  & 3 (1.66667,2$^\dagger$) [\tabsm{T_to_Targ}] & 1 (1.63376*) [\figsm{Wn_level_reduction}]
    \end{tabular}
    }
    \caption{Catalytic conversion rates to and from $\ket{CCZ}$ states. 
    This is an extended version of \tab{CCZ_to_Targ} that includes bounds based on dyadic monotone $\mu_2$.
    In the first column, the produced or consumed state is specified. 
    The second and third columns list the conversion rates ($r$ to consume, and $r'$ to produce)  for the best known algorithm, along with the tightest bound implied by stabilizer extent or nullity in parenthesis.
    The references are provided in square brackets.
    The bounds from the stabilizer extent, marked here by an asterisk, are not known to hold for arbitrary catalysts since the stabilizer extent is currently not known to be multiplicative for all states.
    The bounds from the dyadic monotone, marked here by $\dagger$, hold only for protocols that use measurements with outcome probabilities one half and for catalysts for which $\mu_2$ is defined.
    The results that reference \tab{T_to_Targ} are direct consequence of corresponding result from the table together with the inter-conversion between $\ketsm{T}$ and $\ketsm{CCZ}$. }
    \label{tab:CCZ_to_Targ-ext}
\end{table}

\section{Conclusion and open problems}

We have presented a number of resource lower bounds for a variety of scenarios including resource state conversion, unitary synthesis, and computational tasks. 
To do so, we have introduced a number of new tools, most notably the monotones that we call the stabilizer nullity and the dyadic monotone, along with a canonical form for post-selected stabilizer circuits.
We anticipate that these tools can be used much more broadly, and for example expect the following to be fruitful applications:
    \begin{itemize}
        \item Lower bounds for the multiply-controlled adder, used in multiplication,
        \item Lower bounds for the hamming weight-one state preparation,
        \item Lower bounds for the hamming weight computation circuit,
        \item Lower bounds for small circuits, such as the quantum Fourier transform on small number of qubits.
    \end{itemize}

There are a number of other questions which are raised by this work, which we feel are also deserving of further study:
\begin{enumerate}
    \item For what set of states is the stabilizer extent multiplicative? Although it is multiplicative for all the states that we apply it to, it is not known to be multiplicative for all states, such that not all of our inter-conversion bounds apply in the presence of arbitrary catalysts.
    \item We have found that the exact inter-conversion of stabilizer states is unavoidably lossy, even in the asymptotic limit. 
    In the setting of entanglement theory, exactly converting between different types of entangled state is not possible, but upon relaxing the exact requirement, loss free inter-conversion is possible in entanglement theory.
    It would be interesting to extend the unavoidably lossy resource inter-conversion results to the inexact setting.
    This has been done for odd-prime qudits \cite{wang2018efficiently}, but is not for qubits. 
    % \item Similar to entanglement theory, is there a sense in which we can understand convert-ability of $\ket A$ and $\ket B$ based on whether the Pauli Spectrum of $\ket A$ majorizes the Pauli Spectrum of $\ket B$?
    \item Is there a more efficient algorithm for the quantum adder which is outside the setting of probability half measurements? \item Is there a more efficient algorithm for the multiply controlled $Z$ which is outside the setting of probability 1/2 measurements?
    \item Studying resource state conversion protocols also informs us about possible values of arbitrary monotones. A related open question is the classification of all possible monotones for non-Clifford states that have certain properties, for example additivity (multiplicativity), faithfulness and strong convexity. 
\end{enumerate}

%% Is it possible to generalise phase-polynomial method to higher levels of Clifford hirerarchy and explain more conversion protocols with a common framework 

\section{Acknowledgements}

Circuit diagrams were created using $\langle q | pic \rangle$ \cite{Draper} and Quantikz \cite{Kay2018}. 
The correctness of many of the circuits was verified using Q\# and Microsoft's Quantum Development Kit \cite{QDK}.
We thank Craig Gidney who pointed out a strengthening of our bounds for the adder after the first edition of this paper was released. 
For completeness, we reproduce his argument, presented in \cite{GidneyBlog2019}, in \app{Gidney-bound-reduction}.

M.H.~is supported by a Royal Society--Science Foundation Ireland University Research Fellowship.

%% Use style below for journal version
% \bibliographystyle{plainurl}
%% Use style below to have arxiv links as well as DOI links
\bibliographystyle{plainurl}
\bibliography{resource-states-lower-bounds}

\begin{thebibliography}{10}

\bibitem{amy2016t}
Matthew Amy and Michele Mosca.
\newblock T-count optimization and {R}eed-{M}uller codes.
\newblock {\em IEEE Transactions on Information Theory}, Mar 2019.
\newblock \href {http://arxiv.org/abs/1601.07363} {\path{arXiv:1601.07363}},
  \href {http://dx.doi.org/10.1109/TIT.2019.2906374}
  {\path{doi:10.1109/TIT.2019.2906374}}.

\bibitem{Barenco1995}
Adriano Barenco, Charles~H. Bennett, Richard Cleve, David~P. DiVincenzo, Norman
  Margolus, Peter Shor, Tycho Sleator, John~A. Smolin, and Harald Weinfurter.
\newblock Elementary gates for quantum computation.
\newblock {\em Physical Review A}, 52:3457--3467, Nov 1995.
\newblock \href {http://arxiv.org/abs/quant-ph/9503016}
  {\path{arXiv:quant-ph/9503016}}, \href
  {http://dx.doi.org/10.1103/PhysRevA.52.3457}
  {\path{doi:10.1103/PhysRevA.52.3457}}.

\bibitem{Beverland2016}
Michael~E. Beverland, Oliver Buerschaper, Robert Koenig, Fernando Pastawski,
  John Preskill, and Sumit Sijher.
\newblock Protected gates for topological quantum field theories.
\newblock {\em Journal of Mathematical Physics}, 57(2):022201, 2016.
\newblock URL: \url{https://doi.org/10.1063/1.4939783}, \href
  {http://arxiv.org/abs/https://doi.org/10.1063/1.4939783}
  {\path{arXiv:https://doi.org/10.1063/1.4939783}}, \href
  {http://dx.doi.org/10.1063/1.4939783} {\path{doi:10.1063/1.4939783}}.

\bibitem{bravyi2018simulation}
Sergey Bravyi, Dan Browne, Padraic Calpin, Earl Campbell, David Gosset, and
  Mark Howard.
\newblock Simulation of quantum circuits by low-rank stabilizer decompositions.
\newblock {\em arXiv preprint arXiv:1808.00128}, 2018.
\newblock \href {http://arxiv.org/abs/1808.00128} {\path{arXiv:1808.00128}}.

\bibitem{bravyi2012}
Sergey Bravyi and Jeongwan Haah.
\newblock Magic-state distillation with low overhead.
\newblock {\em Physical Review A}, 86:052329, Nov 2012.
\newblock \href {http://arxiv.org/abs/1209.2426} {\path{arXiv:1209.2426}},
  \href {http://dx.doi.org/10.1103/PhysRevA.86.052329}
  {\path{doi:10.1103/PhysRevA.86.052329}}.

\bibitem{bravyi2005universal}
Sergey Bravyi and Alexei Kitaev.
\newblock Universal quantum computation with ideal clifford gates and noisy
  ancillas.
\newblock {\em Physical Review A}, 71(2):022316, 2005.
\newblock \href {http://arxiv.org/abs/quant-ph/0403025}
  {\path{arXiv:quant-ph/0403025}}, \href
  {http://dx.doi.org/10.1103/PhysRevA.71.022316}
  {\path{doi:10.1103/PhysRevA.71.022316}}.

\bibitem{Bravyi2013}
Sergey Bravyi and Robert König.
\newblock Classification of topologically protected gates for local stabilizer
  codes.
\newblock {\em Physical review letters}, 110:170503, 04 2013.
\newblock \href {http://dx.doi.org/10.1103/PhysRevLett.110.170503}
  {\path{doi:10.1103/PhysRevLett.110.170503}}.

\bibitem{Campbell2010}
Earl~T. Campbell.
\newblock Catalysis and activation of magic states in fault-tolerant
  architectures.
\newblock {\em Physical Review A}, 83:032317, Mar 2011.
\newblock \href {http://arxiv.org/abs/1010.0104} {\path{arXiv:1010.0104}},
  \href {http://dx.doi.org/10.1103/PhysRevA.83.032317}
  {\path{doi:10.1103/PhysRevA.83.032317}}.

\bibitem{CampbellRandom17}
Earl~T. Campbell.
\newblock Shorter gate sequences for quantum computing by mixing unitaries.
\newblock {\em Physical Review A}, 95:042306, Apr 2017.
\newblock \href {http://arxiv.org/abs/1612.02689} {\path{arXiv:1612.02689}},
  \href {http://dx.doi.org/10.1103/PhysRevA.95.042306}
  {\path{doi:10.1103/PhysRevA.95.042306}}.

\bibitem{campbell2018random}
Earl~T. Campbell.
\newblock A random compiler for fast {H}amiltonian simulation.
\newblock {\em arXiv preprint arXiv:1811.08017}, 2018.
\newblock \href {http://arxiv.org/abs/1811.08017} {\path{arXiv:1811.08017}}.

\bibitem{campbell2017unified}
Earl~T. Campbell and Mark Howard.
\newblock Unified framework for magic state distillation and multiqubit gate
  synthesis with reduced resource cost.
\newblock {\em Physical Review A}, 95(2):022316, 2017.
\newblock \href {http://arxiv.org/abs/1606.01904} {\path{arXiv:1606.01904}},
  \href {http://dx.doi.org/10.1103/PhysRevA.95.022316}
  {\path{doi:10.1103/PhysRevA.95.022316}}.

\bibitem{Campbell2018magicstateparity}
Earl~T. Campbell and Mark Howard.
\newblock Magic state parity-checker with pre-distilled components.
\newblock {\em {Quantum}}, 2:56, March 2018.
\newblock \href {http://arxiv.org/abs/1709.02214} {\path{arXiv:1709.02214}},
  \href {http://dx.doi.org/10.22331/q-2018-03-14-56}
  {\path{doi:10.22331/q-2018-03-14-56}}.

\bibitem{campbell2017roads}
Earl~T. Campbell, Barbara~M. Terhal, and Christophe Vuillot.
\newblock Roads towards fault-tolerant universal quantum computation.
\newblock {\em Nature}, 549(7671):172, 2017.
\newblock \href {http://arxiv.org/abs/1612.07330} {\path{arXiv:1612.07330}},
  \href {http://dx.doi.org/10.1038/nature23460}
  {\path{doi:10.1038/nature23460}}.

\bibitem{childs2018faster}
Andrew~M. Childs, Aaron Ostrander, and Yuan Su.
\newblock Faster quantum simulation by randomization.
\newblock {\em arXiv preprint arXiv:1805.08385}, 2018.
\newblock \href {http://arxiv.org/abs/1805.08385} {\path{arXiv:1805.08385}}.

\bibitem{Cleve1997}
Richard Cleve and Daniel Gottesman.
\newblock {Efficient computations of encodings for quantum error correction}.
\newblock {\em Physical Review A}, 56(1):76--82, jul 1997.
\newblock \href {http://arxiv.org/abs/9607030} {\path{arXiv:9607030}}, \href
  {http://dx.doi.org/10.1103/PhysRevA.56.76}
  {\path{doi:10.1103/PhysRevA.56.76}}.

\bibitem{cohen2007number}
H.~Cohen.
\newblock {\em Number Theory: Volume I: Tools and Diophantine Equations}.
\newblock Graduate Texts in Mathematics. Springer New York, 2007.
\newblock URL: \url{https://books.google.com/books?id=8zC8VPQV8psC}.

\bibitem{Draper}
Thomas~G. Draper and Samuel~A. Kutin.
\newblock { $\langle q | pic \rangle$: Quantum circuits made easy}.
\newblock {\em $\,$}, 2019.
\newblock URL: \url{https://github.com/qpic}.

\bibitem{Eastin2009}
Bryan Eastin and Emanuel Knill.
\newblock Restrictions on transversal encoded quantum gate sets.
\newblock {\em Phys. Rev. Lett.}, 102:110502, Mar 2009.
\newblock URL: \url{https://link.aps.org/doi/10.1103/PhysRevLett.102.110502},
  \href {http://dx.doi.org/10.1103/PhysRevLett.102.110502}
  {\path{doi:10.1103/PhysRevLett.102.110502}}.

\bibitem{Forest2015}
Simon Forest, David Gosset, Vadym Kliuchnikov, and David McKinnon.
\newblock {Exact synthesis of single-qubit unitaries over Clifford-cyclotomic
  gate sets}.
\newblock {\em Journal of Mathematical Physics}, 56(8):082201, aug 2015.
\newblock \href {http://arxiv.org/abs/1501.04944} {\path{arXiv:1501.04944}},
  \href {http://dx.doi.org/10.1063/1.4927100} {\path{doi:10.1063/1.4927100}}.

\bibitem{fowler2012}
Austin~G. Fowler, Matteo Mariantoni, John~M. Martinis, and Andrew~N. Cleland.
\newblock Surface codes: Towards practical large-scale quantum computation.
\newblock {\em Physical Review A}, 86:032324, Sep 2012.
\newblock \href {http://arxiv.org/abs/1208.0928} {\path{arXiv:1208.0928}},
  \href {http://dx.doi.org/10.1103/PhysRevA.86.032324}
  {\path{doi:10.1103/PhysRevA.86.032324}}.

\bibitem{Gidney2018}
Craig Gidney.
\newblock {Halving the cost of quantum addition}.
\newblock {\em Quantum}, 2:74, jun 2018.
\newblock \href {http://arxiv.org/abs/1709.06648} {\path{arXiv:1709.06648}},
  \href {http://dx.doi.org/10.22331/q-2018-06-18-74}
  {\path{doi:10.22331/q-2018-06-18-74}}.

\bibitem{GidneyBlog2019}
Craig Gidney.
\newblock {Producing an N+1 Qubit CCZ State with an N Qubit Adder}.
\newblock {\em $\,$}, 2019.
\newblock URL: \url{https://www.https://algassert.com/post/1906}.

\bibitem{gidney2018efficient}
Craig Gidney and Austin~G. Fowler.
\newblock Efficient magic state factories with a catalyzed {|CCZ>} to {2|T>}
  transformation.
\newblock {\em arXiv preprint arXiv:1812.01238}, 2018.
\newblock \href {http://arxiv.org/abs/1812.01238} {\path{arXiv:1812.01238}}.

\bibitem{Gosset2014}
D~Gosset, V~Kliuchnikov, M~Mosca, and V~Russo.
\newblock {An algorithm for the T-count}.
\newblock {\em Quantum Information $\backslash${\&} Computation},
  14(15$\backslash${\&}16):1261--1276, nov 2014.
\newblock \href {http://arxiv.org/abs/1308.4134} {\path{arXiv:1308.4134}}.

\bibitem{Gottesman1999}
Daniel Gottesman and Isaac~L. Chuang.
\newblock {Demonstrating the viability of universal quantum computation using
  teleportation and single-qubit operations}.
\newblock {\em Nature}, 402(6760):390--393, nov 1999.
\newblock \href {http://arxiv.org/abs/quant-ph/9908010}
  {\path{arXiv:quant-ph/9908010}}, \href {http://dx.doi.org/10.1038/46503}
  {\path{doi:10.1038/46503}}.

\bibitem{Grover1996}
Lov~K. Grover.
\newblock A fast quantum mechanical algorithm for database search.
\newblock In {\em Proceedings of the Twenty-eighth Annual ACM Symposium on
  Theory of Computing}, STOC '96, pages 212--219, New York, NY, USA, 1996. ACM.
\newblock \href {http://arxiv.org/abs/quant-ph/9605043}
  {\path{arXiv:quant-ph/9605043}}, \href
  {http://dx.doi.org/10.1145/237814.237866} {\path{doi:10.1145/237814.237866}}.

\bibitem{Haah2017magicstate}
Jeongwan Haah, Matthew~B. Hastings, D.~Poulin, and D.~Wecker.
\newblock Magic state distillation with low space overhead and optimal
  asymptotic input count.
\newblock {\em {Quantum}}, 1:31, October 2017.
\newblock \href {http://arxiv.org/abs/1703.07847} {\path{arXiv:1703.07847}},
  \href {http://dx.doi.org/10.22331/q-2017-10-03-31}
  {\path{doi:10.22331/q-2017-10-03-31}}.

\bibitem{hastings2016turning}
Matthew~B. Hastings.
\newblock Turning gate synthesis errors into incoherent errors.
\newblock {\em arXiv preprint arXiv:1612.01011}, 2016.
\newblock \href {http://arxiv.org/abs/1612.01011} {\path{arXiv:1612.01011}}.

\bibitem{heyfron2018efficient}
Luke~E. Heyfron and Earl~T. Campbell.
\newblock An efficient quantum compiler that reduces {T} count.
\newblock {\em Quantum Science and {T}echnology}, 4(1):015004, Sep 2018.
\newblock \href {http://arxiv.org/abs/1712.01557} {\path{arXiv:1712.01557}},
  \href {http://dx.doi.org/10.1088/2058-9565/aad604}
  {\path{doi:10.1088/2058-9565/aad604}}.

\bibitem{HowardCampbell2016}
Mark Howard and Earl~T. Campbell.
\newblock Application of a resource theory for magic states to fault-tolerant
  quantum computing.
\newblock {\em Physical Review Letters}, 118:090501, Mar 2017.
\newblock \href {http://arxiv.org/abs/1609.07488} {\path{arXiv:1609.07488}},
  \href {http://dx.doi.org/10.1103/PhysRevLett.118.090501}
  {\path{doi:10.1103/PhysRevLett.118.090501}}.

\bibitem{Iten2016}
Raban Iten, Roger Colbeck, Ivan Kukuljan, Jonathan Home, and Matthias
  Christandl.
\newblock Quantum circuits for isometries.
\newblock {\em Physical Review A}, 93:032318, Mar 2016.
\newblock \href {http://arxiv.org/abs/1501.06911} {\path{arXiv:1501.06911}},
  \href {http://dx.doi.org/10.1103/PhysRevA.93.032318}
  {\path{doi:10.1103/PhysRevA.93.032318}}.

\bibitem{Jones2012}
Cody Jones.
\newblock Low-overhead constructions for the fault-tolerant toffoli gate.
\newblock {\em Physical Review A}, 87:022328, Feb 2013.
\newblock \href {http://arxiv.org/abs/1212.5069} {\path{arXiv:1212.5069}},
  \href {http://dx.doi.org/10.1103/PhysRevA.87.022328}
  {\path{doi:10.1103/PhysRevA.87.022328}}.

\bibitem{jones2013multilevel}
Cody Jones.
\newblock Multilevel distillation of magic states for quantum computing.
\newblock {\em Physical Review A}, 87(4):042305, 2013.
\newblock \href {http://arxiv.org/abs/1210.3388} {\path{arXiv:1210.3388}},
  \href {http://dx.doi.org/10.1103/PhysRevA.87.042305}
  {\path{doi:10.1103/PhysRevA.87.042305}}.

\bibitem{karzig17}
Torsten Karzig, Christina Knapp, Roman~M. Lutchyn, Parsa Bonderson, Matthew~B.
  Hastings, Chetan Nayak, Jason Alicea, Karsten Flensberg, Stephan Plugge,
  Yuval Oreg, Charles~M. Marcus, and Michael~H. Freedman.
\newblock Scalable designs for quasiparticle-poisoning-protected topological
  quantum computation with {{Majorana}} zero modes.
\newblock {\em Phys. Rev. B}, 95(23):235305, June 2017.
\newblock \href {http://arxiv.org/abs/1610.05289} {\path{arXiv:1610.05289}},
  \href {http://dx.doi.org/10.1103/PhysRevB.95.235305}
  {\path{doi:10.1103/PhysRevB.95.235305}}.

\bibitem{Kay2018}
Alastair Kay.
\newblock {Quantikz}.
\newblock 9 2018.
\newblock URL:
  \url{https://royalholloway.figshare.com/articles/Quantikz/7000520}, \href
  {http://dx.doi.org/10.17637/rh.7000520.v3}
  {\path{doi:10.17637/rh.7000520.v3}}.

\bibitem{KBRY}
Vadym Kliuchnikov, Alex Bocharov, Martin Roetteler, and John Yard.
\newblock {A Framework for Approximating Qubit Unitaries}.
\newblock {\em arXiv preprint arXiv:1510.03888}, oct 2015.
\newblock \href {http://arxiv.org/abs/1510.03888} {\path{arXiv:1510.03888}}.

\bibitem{Knill1995}
Emanuel Knill.
\newblock {Approximation by Quantum Circuits}.
\newblock {\em arXiv preprint arXiv:1812.10145}, pages 1--23, aug 1995.
\newblock \href {http://arxiv.org/abs/quant-ph/9508006}
  {\path{arXiv:quant-ph/9508006}}.

\bibitem{knill2005quantum}
Emanuel Knill.
\newblock Quantum computing with realistically noisy devices.
\newblock {\em Nature}, 434(7029):39, 2005.
\newblock \href {http://arxiv.org/abs/quant-ph/0410199}
  {\path{arXiv:quant-ph/0410199}}, \href
  {http://dx.doi.org/10.1038/nature03350} {\path{doi:10.1038/nature03350}}.

\bibitem{aqft2018}
Yunseong Nam, Yuan Su, and Dmitri Maslov.
\newblock Approximate quantum fourier transform with o(nlog(n)) {T} gates.
\newblock {\em arXiv preprint arXiv:1803.04933}, 2018.
\newblock \href {http://arxiv.org/abs/1803.04933} {\path{arXiv:1803.04933}}.

\bibitem{Pastawski2015}
Fernando Pastawski and Beni Yoshida.
\newblock Fault-tolerant logical gates in quantum error-correcting codes.
\newblock {\em Phys. Rev. A}, 91:012305, Jan 2015.
\newblock URL: \url{https://link.aps.org/doi/10.1103/PhysRevA.91.012305}, \href
  {http://dx.doi.org/10.1103/PhysRevA.91.012305}
  {\path{doi:10.1103/PhysRevA.91.012305}}.

\bibitem{QDK}
Microsoft Quantum.
\newblock {Microsoft Quantum Development Kit}.
\newblock {\em $\,$}, 2019.
\newblock URL: \url{https://www.microsoft.com/en-us/quantum/development-kit}.

\bibitem{regula2017convex}
Bartosz Regula.
\newblock Convex geometry of quantum resource quantification.
\newblock {\em Journal of Physics A: Mathematical and Theoretical},
  51(4):045303, 2017.
\newblock \href {http://arxiv.org/abs/1707.06298} {\path{arXiv:1707.06298}},
  \href {http://dx.doi.org/10.1088/1751-8121/aa9100}
  {\path{doi:10.1088/1751-8121/aa9100}}.

\bibitem{RossSelinger}
Neil~J. Ross and Peter Selinger.
\newblock {Optimal ancilla-free Clifford+T approximation of z-rotations}.
\newblock {\em Quantum Information and Computation}, 16(11-12):901--953, mar
  2016.
\newblock \href {http://arxiv.org/abs/1403.2975} {\path{arXiv:1403.2975}}.

\bibitem{Selinger2012}
Peter Selinger.
\newblock Quantum circuits of $t$-depth one.
\newblock {\em Physical Review A}, 87:042302, Apr 2013.
\newblock \href {http://arxiv.org/abs/1210.0974} {\path{arXiv:1210.0974}},
  \href {http://dx.doi.org/10.1103/PhysRevA.87.042302}
  {\path{doi:10.1103/PhysRevA.87.042302}}.

\bibitem{Shende2006}
V.~V. Shende, S.~S. Bullock, and I.~L. Markov.
\newblock Synthesis of quantum-logic circuits.
\newblock {\em Trans. Comp.-Aided Des. Integ. Cir. Sys.}, 25(6):1000--1010,
  June 2006.
\newblock \href {http://arxiv.org/abs/quant-ph/0406176}
  {\path{arXiv:quant-ph/0406176}}, \href
  {http://dx.doi.org/10.1109/TCAD.2005.855930}
  {\path{doi:10.1109/TCAD.2005.855930}}.

\bibitem{shende2004minimal}
Vivek~V. Shende, Igor~L. Markov, and Stephen~S. Bullock.
\newblock Minimal universal two-qubit controlled-not-based circuits.
\newblock {\em Physical Review A}, 69(6):062321, 2004.
\newblock \href {http://arxiv.org/abs/quant-ph/0308033}
  {\path{arXiv:quant-ph/0308033}}, \href
  {http://dx.doi.org/10.1103/PhysRevA.69.062321}
  {\path{doi:10.1103/PhysRevA.69.062321}}.

\bibitem{shende2004smaller}
Vivek~V. Shende, Igor~L. Markov, and Stephen~S. Bullock.
\newblock Smaller two-qubit circuits for quantum communication and computation.
\newblock In {\em Proceedings Design, Automation and Test in Europe Conference
  and Exhibition}, volume~2, pages 980--985. IEEE, 2004.

\bibitem{wang2018efficiently}
Xin Wang, Mark~M Wilde, and Yuan Su.
\newblock Efficiently computable bounds for magic state distillation.
\newblock {\em arXiv preprint arXiv:1812.10145}, 2018.
\newblock \href {http://arxiv.org/abs/1812.10145} {\path{arXiv:1812.10145}}.

\end{thebibliography}

\newpage
\appendix

\section{Appendices}

\subsection{Generic circuits for injecting diagonal gates}
\label{app:diagonal-gate-injection}
In this section we provide an algorithm to implement any $n$-qubit diagonal unitary $U$ using the corresponding resource state $\ket{U}=U \ket{+}^{\otimes n}$, as mentioned in \sec{intro-computational-tasks}.
It was pointed out in \cite{HowardCampbell2016} that when $U$ belongs to the third level of the Clifford hierarchy it can be implemented using one resource state $\ket{U}$ via a half-teleportation circuit (see Figure~1(a) in \cite{HowardCampbell2016}). 
Here we make this protocol more explicit as well as slightly more general.
The algorithm is as follows:
\begin{alg}[Apply the diagonal $n$-qubit unitary $U$ using $\ket{U}$] $\,$ \newline
\label{alg:inject}
\textbf{Input:} $2n$ qubits, with the first $n$ qubits in the state $\ket{U}$, and the last $n$ qubits in an arbitrary state $\ket{\alpha}$.
\begin{enumerate}
	\item apply $\mathrm{CNOT}_{n+1,1}\ldots \mathrm{CNOT}_{2n,n}$.
	\item measure the first $n$ qubits; the measurement outcomes are $m(1),\ldots,m(n)$. 
	\item for each $k$ in $\set{1,\ldots,n}$: if $m(k)$ is $1$ apply $X_{n+k}$.
	\item for each $k$ in $\set{1,\ldots,n}$: if $m(k)$ is $1$ apply $U X_{n+k} U^{\dagger}$.
\end{enumerate}
\textbf{Output:} The first $n$ qubits are in a known computational basis state, and the last $n$ qubits are in the state $U\ket{\alpha}$.
\end{alg}

Above we use the notation $\mathrm{CNOT}_{a,b}$ for CNOT gate with the control qubit $a$ and target qubit $b$. 
Note that step 3 must be completed for all $k$ in $\{ \dots n\}$ before proceeding to step 4.
Next we prove the correctness of above protocol. 

\begin{prop} Algorithm~\ref{alg:inject} is correct. 
If the diagonal unitary $U$ belongs to level $k$ of the Clifford hierarchy, then the corrections applied in step (4) are unitaries that belong to at most level $k-1$ of the hierarchy.
\end{prop}
\begin{proof}
Let us first show the correctness. We will use the following notation for $U$ and $\ket\alpha$:
\[
U = \sum_{k \in \set{0,1}^n} e^{i \varphi\at{k}} \ket{k}\bra{k},\quad
\ket{\alpha} = \sum_{k \in \set{0,1}^n }\alpha_k\ket{k}.
\]
The initial state can be written as: 
\[
U \ket{+}^{\otimes n} \ket{\alpha}=\frac{1}{2^{n/2}}\sum_{k,j \in \set{0,1}^n } e^{i \varphi\at{k} }\alpha_j\ket{k,j}. 
\]
After applying the CNOT gates in step (1) the state becomes: 
\[
\frac{1}{2^{n/2}}\sum_{k,j \in \set{0,1}^n } e^{i \varphi\at{k} }\alpha_j\ket{k\oplus j,j} = \frac{1}{2^{n/2}}\sum_{k,j \in \set{0,1}^n } e^{i \varphi\at{k \oplus j } }\alpha_j\ket{k,j}.
\]
For measurement outcome $m = (m(1),\ldots,m(n))$, the state after step (2) will be
\[
\ket{m} \otimes \sum_{j \in \set{0,1}^n } e^{i \varphi\at{j \oplus m } }\alpha_j\ket{j}.
\]
After applying the last two steps of the protocol the state of qubits $n+1,\ldots,2n$ will be:
\[
 U \at{X^{m(1)}\otimes\ldots \otimes X^{m(n)}} U^{\dagger} \at{X^{m(1)}\otimes\ldots \otimes X^{m(n)}} \sum_{j \in \set{0,1}^n } e^{i \varphi\at{j \oplus m } }\alpha_j\ket{j}.
\]
Note that:
\begin{eqnarray}
U^{\dagger} \at{X^{m(1)}\otimes\ldots \otimes X^{m(n)}} \sum_{j \in \set{0,1}^n } e^{i \varphi\at{j \oplus m } }\alpha_j\ket{j} &=& U^{\dagger} \sum_{j \in \set{0,1}^n }{e^{i \varphi\at{j \oplus m } }\alpha_j\ket{j \oplus m}},\nonumber \\
 &=& U^{\dagger} \sum_{j \in \set{0,1}^n } e^{i \varphi\at{j} }\alpha_{j\oplus m}\ket{j},\nonumber \\
 &=& \sum_{j \in \set{0,1}^n } \alpha_{j\oplus m}\ket{j}.\nonumber
\end{eqnarray} 
Finally, we see that
\[ 
U \at{X^{m(1)}\otimes\ldots \otimes X^{m(n)}} \sum_{j \in \set{0,1}^n} \alpha_{j\oplus m} \ket{j} = U\ket{\alpha}.
\]
as required. 

Finally, we note that if $U$ belong to the $\ell^{\mathrm{th}}$ level of the Clifford hierachy, we have by definition that all $U X_k U^\dagger$ belong to the $(\ell-1)^{\mathrm{th}}$ level. 

\end{proof}

We finish this section with the expression for some explicit corrections $U X_k U^\dagger$ in the following list and in \fig{ccz-inject} and \fig{gateteleportation}.
\begin{itemize}
	\item $U = \exp\at{i \pi \ket{1}\bra{1} /2^k}$, correction: $U X U^\dagger = e^{-i \pi /2^k} \exp\at{i \pi \ket{1}\bra{1} /2^{k-1}} X$.
	\item $U = CS = \exp\at{\frac{\pi i }{2}\ket{11}\bra{11}}$, corrections: 
	\begin{itemize}
		\item $U X_1 U^\dagger = \exp\at{\frac{-\pi i }{2}Z\otimes\ket{1}\bra{1}} X_1 =  \text{CNOT}_{1,2} S_1 S_2^\dagger \text{CNOT}_{1,2} X_1$,
		\item $U X_2 U^\dagger = \text{SWAP}_{1,2} U X_1 U^\dagger \text{SWAP}_{1,2} = \text{CNOT}_{2,1} S_2 S_1^\dagger \text{CNOT}_{2,1} X_2$.
	\end{itemize}
	\item $U = CCZ = \exp\at{i \pi \ket{111}\bra{111}}$, corrections: 
	\begin{itemize}
		\item $U X_1 U^\dagger = \exp\at{i\pi I\otimes\ket{11}\bra{11}} X_1 = CZ_{2,3} X_1$,
		\item $U X_2 U^\dagger = \text{SWAP}_{1,2} U X_1 U^\dagger \text{SWAP}_{1,2} = CZ_{1,3} X_2$,
		\item $U X_3 U^\dagger = \text{SWAP}_{1,3} U X_1 U^\dagger \text{SWAP}_{1,3} = CZ_{1,2} X_3$.
	\end{itemize}
\end{itemize}

\begin{figure}
    \centering
 	\includegraphics{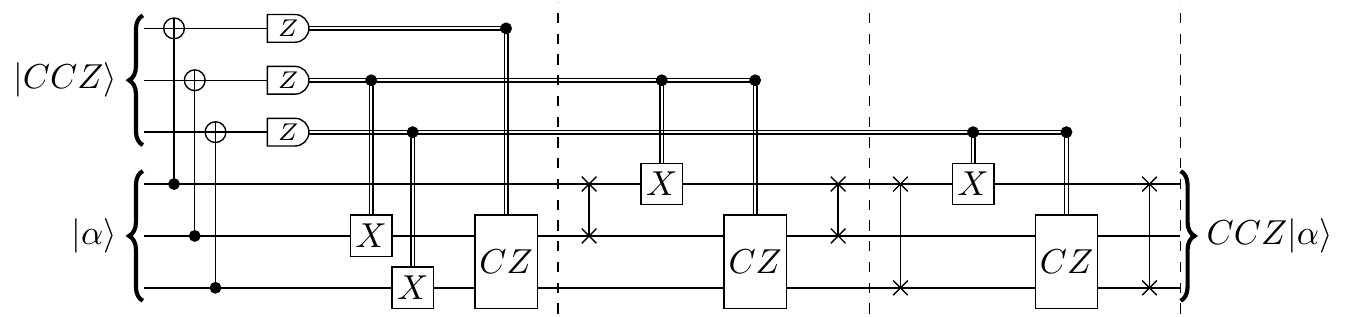}
 	\caption{Implementing $CCZ$ using $\ket{CCZ}$.}
 	\label{fig:ccz-inject}
\end{figure}

\begin{figure}[ht!]
  \begin{subfigure}[b]{0.45\textwidth}
    \centering
    \includegraphics{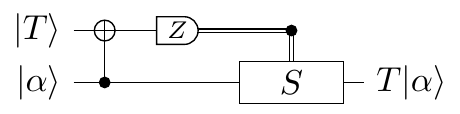}
    \caption{Implementing $T$ using $\ket{T}$.}
    \label{fig:t-inject}
  \end{subfigure}
  \begin{subfigure}[b]{0.45\textwidth}
    \centering
    \includegraphics{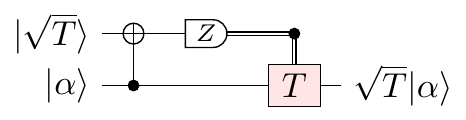}
    \caption{Implementing $\sqrt{T}$ using $|\sqrt{T}\rangle$ and $\ket{T}$.}
    \label{fig:sqrt-t-inject}
  \end{subfigure}

  \vspace{1em}
  \begin{subfigure}[b]{0.9\textwidth}
  	\centering
    \includegraphics{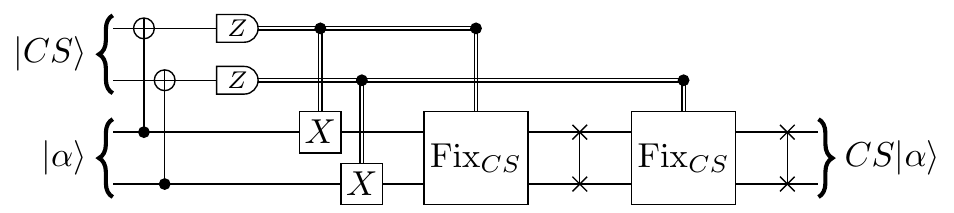}
  \end{subfigure}

  \vspace{1em}
  \begin{subfigure}[b]{0.9\textwidth}
  	\centering
 	\includegraphics{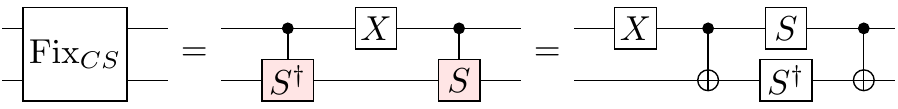}
 	\caption{Implementing $CS$ using $\ket{CS}$.}
 	\label{fig:cs-inject}
  \end{subfigure}

 	\caption{Gate injection circuits to apply some non-Clifford gates using resource ancilla states and Clifford operations.
 	 Shaded boxes represent gates in the third level of the Clifford hierarchy, which if necessary could in turn be implemented using a resource state and a Clifford circuit.}
 	\label{fig:gateteleportation}
\end{figure}

\subsection{Reducing the cost of unitary synthesis using \texorpdfstring{$\sqrt{T}$}{√̅T} gates}
\label{app:seqeuntial-root-t}
\begin{figure}[ht]
  \centering
  \includegraphics[scale=0.8]{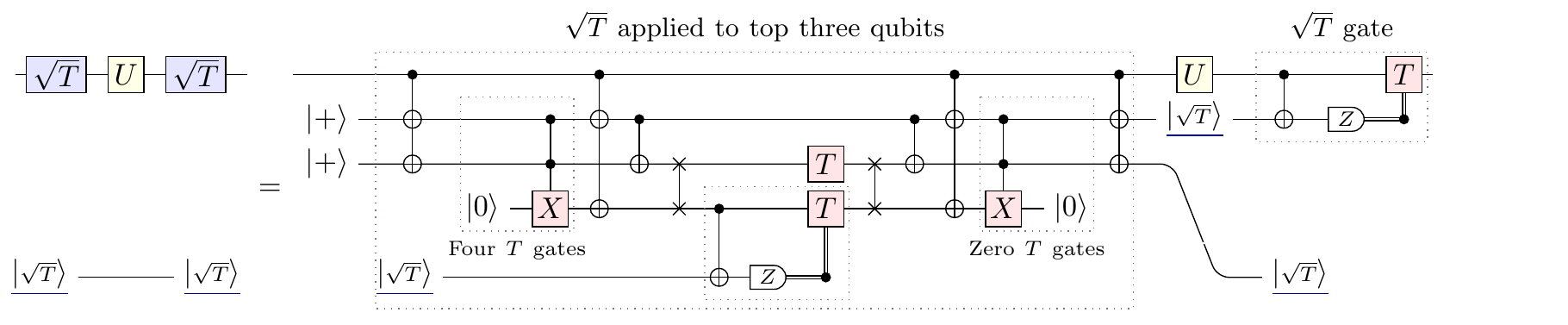}
  \caption{Application of $\sqrt{T}U\sqrt{T}$ catalyzed by a $\ketsm{\sqrt{T}}$ state. 
  This uses six $\ketsm{T}$ states on average, and always uses at least five and at most seven $\ketsm{T}$ states.
  The $\sqrt{T}$ gate can be injected using the $\ket{\sqrt{T}}$ state as in \fig{sqrt-t-inject}.}
  \label{fig:two-sqrt-t-gates}
\end{figure}

In this section we describe how to reduce the cost of approximate unitay synthesis using $\sqrt{T}$ states as mentioned in \sec{Conversions}.
We also make use of a trick to reduce the injection cost when applying $\sqrt{T}$ gates sequentially.

Applying a $\sqrt{T}$ gate using magic state injection uses an extra $T$ gate with probability one half. 
Using the family of conversion protocols $\ketsm{\sqrt{T}} + (5k+\frac{1}{2})\ket{T} \rightarrow (2k+1)\ketsm{\sqrt{T}}$ to create $\ketsm{\sqrt{T}}$ states, applying one $\sqrt{T}$ gate uses on average $3 + 1/(4k)$ $T$ gates. 
In the worst case, this method will use $3.5 + 1/(2k)$ $T$ gates. 
We further reduce the number of $T$ gates needed to apply $\sqrt{T} U \sqrt{T}$. 
This situation is common when $\sqrt{T}$ gates are used for the synthesis of single qubit $Z$ rotations by an arbitrary angle. 
The circuit shown in \fig{two-sqrt-t-gates} uses on average three $T$ gates per $\sqrt{T}$ gate and $3.5$ $T$ gates in the worst case. 
In addition, applying $\sqrt{T}$ gates using the protocol in \fig{two-sqrt-t-gates} requires less ancillary qubits in comparison to using conversion protocols $\ketsm{\sqrt{T}} + (5k+\frac{1}{2})\ket{T} \rightarrow (2k+1)\ketsm{\sqrt{T}}$ for $k > 1$.

A significant application of the above is to reduce the overhead of circuit synthesis by giving access to a larger gate set.
We therefore take an aside here to explain the context and describe how our results imply overhead reduction.
Approximating the single qubit rotation $\exp\at{i\theta\ket{1}\bra{1}}$ to within $1$-norm accuracy $\ve$ using Clifford and $T$ gates requires less then $3 \log_2\at{1/\varepsilon} + O(\log(\log_2(1/\ve)))$ $T$ gates~\cite{RossSelinger} in the typical case and less then $4\log_2(1/\ve) + O(1)$ in the worst case.
Consider now expanding the gate set to Clifford, $T$, $\sqrt{T}$ and $\sqrt{T}^3$ gates.
If $N_{T}$, $N_{\sqrt{T}}$ and $N_{\sqrt{T}^3}$ denote the number of $T$, $\sqrt{T}$ and $\sqrt{T}^3$ gates used to approximate the rotation, then the algorithm 
described in \cite{KBRY} finds gate sequences with the number of gates satisfying:
\begin{equation*}\label{eq:cost}
2 N_{T} + 3\at{N_{\sqrt{T}} + N_{\sqrt{T}^3}} < 4 \log_2\at{1/\ve} + O(1).
\end{equation*}
For this algorithm it was also empirically observed that $N_{T} \approx N_{\sqrt{T}} + N_{\sqrt{T}^3}$. Assuming that applying $\sqrt{T}$ and $\sqrt{T}^3$ gates consumes $\alpha$ $\ketsm{T}$ states, we see that using Clifford, $T$, $\sqrt{T}$ and $\sqrt{T}^3$ gates for rotation synthesis will use less than 
\[
\frac{1 + \alpha }{5} \cdot 4 \log_2\at{1/\ve} + O(1)
\]
$T$ gates. When the same algorithm uses only Clifford and $T$ gate set it finds sequences with the at most  $4\log_2(1/\ve) + O(1)$ $T$ gates. Therefore, we achieve break-even point with Clifford and $T$ synthesis when applying $\sqrt{T}$ gate consumes four $T$ gates. In the best protocol we find so far three $T$ gates are consumed for each $\sqrt{T}$ gate applied on average and $3.5$ $T$ gates are consumed in the worst case. This results in an average-case $20 \%$ reduction and worst case $10 \%$ reduction in the number of $T$ gates used to synthesize single qubit rotation. 

\subsection{Explicit circuits for some common resource conversions}
\label{app:explicit-circuits}
In this section we include a number of explicit constructions which provide conversion upper bounds that appear in \tab{T_to_Targ} and \tab{CCZ_to_Targ} in \sec{Conversions}.

\begin{figure}[ht!]
	\begin{subfigure}[b]{0.9\textwidth}
		\centering
		\includegraphics[scale=1]{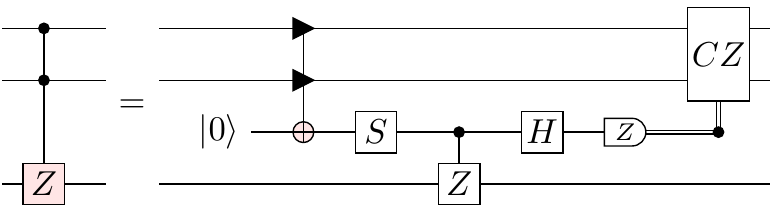}
	 	\caption{$CCZ$ gate using Toffoli* gate \cite{Jones2012}.}
	 	\label{fig:cczgate-from-t}
	\end{subfigure}
	\begin{subfigure}[b]{0.9\textwidth}
	\centering
	\includegraphics[scale=1]{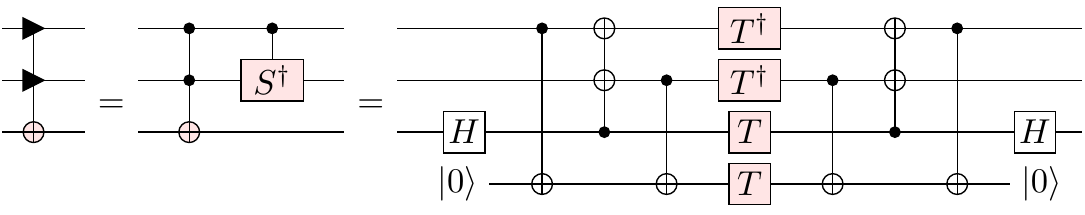}
	\caption{The Toffoli* gate, which differs from the Toffoli gate (as defined in \cite{Selinger2012}), uses four $\ket{T}$ states.}
	\label{fig:toffoli-star}
	\end{subfigure}
	\caption{Known circuits for implementing $CCZ$ gate using four $T$ gates.}
	\label{fig:gatesynthesisfromT}
\end{figure}
\begin{figure}[ht!]
	\hfill
	\begin{subfigure}[b]{0.45\textwidth}
		\centering
		\includegraphics[scale=0.9]{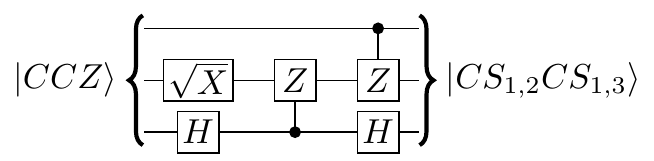}
		\caption{Conversion $\ket{CCZ} \leftrightarrow \ket{CS_{1,2} CS_{2,3}}$}
		\label{fig:ccz-to-cscs}
	\end{subfigure}
	\hfill
	\begin{subfigure}[b]{0.45\textwidth}
		\centering
		\includegraphics[scale=0.9]{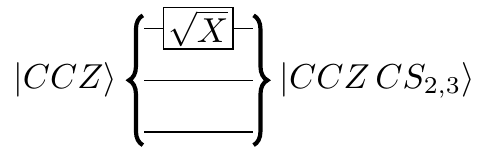}
		\caption{Conversion $\ket{CCZ} \leftrightarrow \ket{CCZ\,CS_{2,3}}$}
		\label{fig:ccz-to-cczcs}
	\end{subfigure}
	\hfill
	\caption{Two way conversion of resource states from \cite{HowardCampbell2016}. These circuits are useful subroutines for some of our results.}
	\label{fig:two-way}
\end{figure}
\begin{figure}[ht!]
	\hfill
	\begin{subfigure}[b]{0.45\textwidth}
    \centering\includegraphics[scale=0.9]{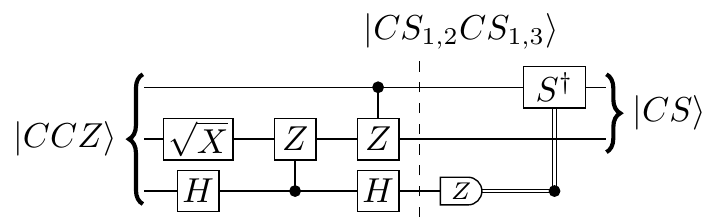}
      \caption{Conversion $\ket{CCZ} \rightarrow \ket{CS}$.}
      \label{fig:ccz-to-cs}
	\end{subfigure}
	\hfill
	\begin{subfigure}[b]{0.45\textwidth}
		\centering\includegraphics[scale=0.9]{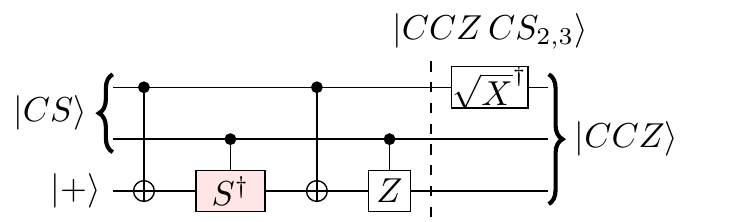}
		\caption{Conversion $2\ket{CS} \rightarrow \ket{CCZ}$.}
		\label{fig:cs-to-ccz}
	\end{subfigure}
	\hfill
	\caption{Conversion between $\ket{CCZ}$ and $\ket{CS}$.}
	\label{fig:gatesynthesis}
\end{figure}
\begin{figure}[ht!]
    \centering\includegraphics[scale=0.9]{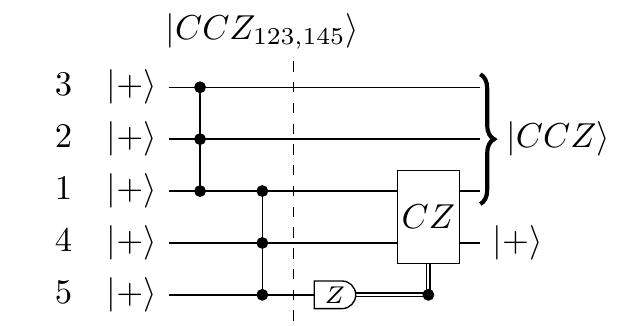}
	\caption{Conversion between $\ket{CCZ}$ and $\ket{CCZ_{123,145}}$.}
	\label{fig:ccz123145}
\end{figure}
\begin{figure}[ht!]
    \centering\includegraphics[scale=0.9]{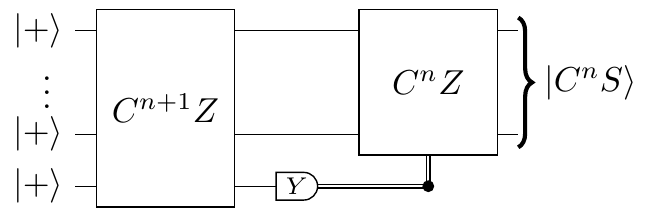}
	\caption{Conversion from $\ket{C^{n+1}Z}$ to $\ket{C^{n-1}S}$ or $\ket{C^{n-1}S^\dagger}$ with probability one half. See \propos{multiCCZtoCS} for the correctness proof.}
	\label{fig:multiCCZtoCS}
\end{figure}

\begin{prop} \label{prop:measure-control}
Let $U$ be a diagonal $n$-qubit unitary, and let $CU$ be a controlled version of $U$,
then measuring the fist qubit of the state $\ket{CU} = CU \ket{+}^{\otimes (n+1)}$ in $Z$ basis sets the rest of the qubits into the state $\ket{U}$ with probability one half and into the state $\ket{+}^{\otimes n}$ otherwise.

In particular, this implies the following conversion protocols: 
\begin{itemize}
    \item $\ketsm{C^n Z} \rightarrow \frac{1}{2} \ketsm{C^{n-1} Z} \rightarrow \ldots  \rightarrow  \frac{1}{2^{n-2}} \ketsm{CCZ}  $
    \item $\ketsm{C^n S} \rightarrow  \frac{1}{2} \ketsm{C^{n-1} S} \rightarrow \ldots \rightarrow  \frac{1}{2^{n-1}} \ketsm{CS}  $
\end{itemize}
\end{prop}
\begin{proof}
Note that projectors $(I\pm Z)/2$ commute with $CU$ and therefore applying $Z$ measurement to the first qubit is the same as measuring $Z$ on the first qubit of $\ket{+}^{\otimes (n+1)}$ and then applying $CU$ to $\ket{0}\otimes \ket{+}^{\otimes n}$ or $\ket{1}\otimes \ket{+}^{\otimes n}$ depending on the measurement outcome. We get $\ket{0}$ or $\ket{1}$ on the first qubit with probability one half and therefore we get $\ket{U}$ or $\ket{+}^{\otimes n}$ on the rest of the qubits with probability one half.
\end{proof}

\begin{prop} \label{prop:multiCCZtoCS}
For $n\ge 0$, the probability of measuring the eigenvalue $m = \pm 1$ of $Y$ on the first qubit of $\ket{C^{n+1} Z}$ is $1/2$. After the measurement, the state of the rest of the qubits is $\ket{C^n S^{m}}$. 
\end{prop}
\begin{proof}
Let us first show that the probability of measurement outcome is $1/2$.
Let us write 
\[
\ket{C^{n+1}Z} = \ket{0}\otimes\ket{+}^{\otimes n}/\sqrt{2} + \ket{1}\otimes\ket{C^n Z}^{\otimes n}/\sqrt{2}
\]
The probability of measuring $+1$ eigenvalue of $Y$ is: 
\[
\bra{C^{n+1}Z}I+Y\ket{C^{n+1}Z}/2 =
\bra{0}I+Y\ket{0}/4 + 
\bra{0}I+Y\ket{0}/4 +  
\alpha\bra{0}I+Y\ket{1}/4 +
\alpha^\ast \bra{1}I+Y\ket{0}/4,
\]
where $\alpha = \bra{+}^{\otimes n} \ket{C^n Z}$. 
The probability is half because $\alpha$ is a real number and $\bra{0}I+Y\ket{1} = - \bra{1}I+Y\ket{0}$. 

We prove the second part of the proposition by induction on $n$.
When $n=0$, and the measurement outcome is $+1$, the second qubit will be in the state 
\[
\frac{I+Y_1}{\sqrt{2}}\ket{CZ} 
=
\ket{i} \otimes \at{\ket{+}/\sqrt{2} -i Z\ket{+}\sqrt{2}}
= 
e^{- i \pi/4} \ket{i} \otimes \ket{S},
\]
where $\ket{i} = (1,i)/\sqrt{2}$.
Suppose we now we have shown that  
\[
\frac{I+Y_1}{\sqrt{2}}\ket{C^{n}Z} = e^{- i \pi/4} \ket{i} \otimes \ket{C^{n-2} S}
\]
Let us now observe that 
\[
\frac{I+Y_1}{\sqrt{2}}\ket{C^{n+1}Z} =
\frac{I+Y_1}{\sqrt{2}} \ket{+}^{\otimes n}\otimes\ket{0} 
+ \frac{I+Y_1}{\sqrt{2}} \ket{C^{n} Z} \otimes \ket{1} 
\]
By induction hypothesis and the fact that $(I+I)/\sqrt{2} \ket{+} = e^{-i \pi /4} \ket{i}$ it follows that:
\[
\frac{I+Y_1}{\sqrt{2}}\ket{C^{n+1}Z} = e^{- i \pi/4} \ket{i} \otimes \at{ \ket{+}^{\otimes n-1}\otimes\ket{0} + \ket{C^{n-1} S} \ket{1} } = e^{- i \pi/4} \ket{i} \otimes \ket{C^n S}
\]
By applying element-wise complex conjugation to all the equations above we get the proof for the $-1$ outcome of the measurement, because $Y^\ast = -Y$.
\end{proof}

\subsection{Extent values} 
\label{app:extent-values}
Here in \tab{extent_values} we list the extent values for some common resource states, which are used to produce some of the bounds in \tab{T_to_Targ} and \tab{CCZ_to_Targ} in \sec{montonebounds}.
To rigorously find the exact value of the extent one can perform the following steps:

\begin{enumerate}
    \item Find approximate numerical solutions to the primal and dual linear programs (that is a decomposition into a linear combination of stabilizer states and a witness state).
    \item Guess exact expressions close to the approximate solutions (or use algebraic number reconstruction tools to find them).
    \item Plug the (guessed) exact solutions into the linear program and see that min/max for primal/dual problem are equal thereby confirming they are the true solutions.
\end{enumerate}

Note that we did not perform the rigorous extent calculation for some of the multi-qubit states in \tab{extent_values}.
Instead, we computed extent value up to eight digits of precision and reconstructed the exact expression that matches found approximation.

\begin{table*}[ht]

	\hfill
	\begin{subfigure}[b]{0.45\textwidth}
\centering
{\renewcommand{\arraystretch}{1.25}
\begin{tabular}{c|c}
$\ket{\psi}$ & $\xi(\ket{\psi})$ \\ 
\hline 
$\ket{\sqrt{T}}$ & $2-\sqrt{2} + 1/\sqrt{2+\sqrt{2}}$ \\ 
$\ket{T}$  & $\frac{4}{2+\sqrt{2}}$ \\ 
$\ket{CS}$  & $\frac{8}{5}$ \\ 
$\ket{CCS}$  & $\frac{41}{20}$ \\ 
$\ket{C^3S}$  & $\frac{9}{8}+\frac{1}{\sqrt{2}}$ \\ 
$\ket{CCZ}$  & $\frac{16}{9}$ 
\end{tabular}
}
	\end{subfigure}
	\hfill
	\begin{subfigure}[b]{0.45\textwidth}
\centering
{\renewcommand{\arraystretch}{1.25}
\begin{tabular}{c|c}
$\ket{\psi}$ & $\xi(\ket{\psi})$ \\ 
\hline 
$\ket{C^3Z}$  & $\frac{9}{4}$ \\ 
$\ket{C^4Z}$  & $\frac{9}{8}+\frac{1}{\sqrt{2}}$  \\ 
$\ket{CCZ_{123,145}}$  & $2$ \\ 
$\ket{W_3}$  & $\frac{16}{9}$ \\
$\ket{W_4}$  & $\frac{64}{29}$ \\ 
$\ket{W_5}$  & $\frac{64}{25}$ \\
\end{tabular}
}
	\end{subfigure}
	\hfill
	
\caption{Exact expressions for the extent of the states used in \tab{T_to_Targ} and \tab{CCZ_to_Targ}.
All values are accurate to within eight digits of precision. Note that $\ket{C^4Z}$ and $\ket{C^3S}$ have the same extent but they are not Clifford-equivalent; measuring $\ket{C^4Z}$ on the last qubit in the $Y$ basis produces $\ket{C^3S}$ or $\ket{C^3S^\dagger}$~(\propos{multiCCZtoCS}), but conversion in the reverse direction is ruled out by the stabilizer nullity.}
\label{tab:extent_values}
\end{table*}

\subsection{Further details on phase polynomial protocols}
\label{app:phase-polynomial}
Here we provide the proofs for \theo{PhasePolyCat} and \lemm{tau-for-Wn} in \sec{polynomialcatalysis}. 
% \begin{thm*}
% %\label{thm:PhasePolyCat}
% 	Let $\ket{U}=U\ket{+}^{\otimes n}$ be an $n$-qubit magic state for a unitary $U$ from the 3$^{\mathrm{rd}}$ level of the Clifford hierarchy, and let $\tau(U)$ be the minimum number of $T$ gates needed to implement $U$ using the gate set $\{CNOT, S,T \}$.
%     The following resource conversion is possible
% 	\begin{equation}
% 	\ket{U}   \catalyse{\ket{T}^{\otimes \tau(U)-\mu( \ket{U} )}} \ket{T}^{\otimes (2\mu( \ket{U} ) - \tau(U))}.
% \end{equation}	
% \end{thm*}	

\PhasePolyCat*
In this theorem, we follow the conversion notation of \defn{conversion-equations} and use $\nu$ that was defined earlier as the stabilizer nullity (recall \defn{stabilizer-monotone}). 

\begin{proof} The proof of the theorem uses the phase polynomial formalism, which we quickly review here and the reader can learn more about in Refs.~\cite{amy2016t,campbell2017unified,heyfron2018efficient}.

For any diagonal unitary in the 3$^{\mathrm{rd}}$ level of the Clifford hierarchy we have
\begin{equation}
	U_f = \sum_x \exp( i f(x) \pi / 4) 	\kb{x}{x} ,
\end{equation}	
where $f:\mathbb{Z}_2^n \rightarrow \mathbb{Z}_8$ is a cubic form.  That is, $f$ can be decomposed as the phase polynomial
\begin{equation}
		f(x) = \sum_{a_k \neq 0} a_k \lambda_k(x) \pmod{8},
\end{equation}
where $a_k \in \mathbb{Z}_8$ and each $\lambda_k$ is a $ \mathbb{Z}_2$ linear function. That is, each $\lambda_k$ has the form
\begin{equation}
		\lambda_k(x) = (P_{1,k} x_1) \oplus  (P_{2,k} x_2)  \ldots (P_{n,k} x_n) \pmod{2} ,
\end{equation}	
where $P_{j,k}$ are binary.  Therefore, the function can be described by a binary matrix $P$ and vector $a$.  We only define columns of $P$ for nonzero $a_k$, so it has a number of columns equal to the number of terms in $f$.

For a function with a single term $f(x)=a_k \lambda_k(x)$, an easily verified circuit decomposition is
\begin{equation}
  U_{\lambda_k}= \sum_{x} \exp( i \lambda_k(x) \pi / 4)   \kb{x}{x}= V_{CNOT (\lambda_k)}^{\dagger} T_1^{a_k} V_{CNOT (\lambda_k)}
  \end{equation}
where $T_1$ is a $T$ gate acting on qubit 1 and $V_{CNOT (\lambda_k)}$ is a cascade of CNOT gates such that 
\begin{equation}
   V_{CNOT (\lambda_k)} \ket{ x } = V_{CNOT (\lambda_k)} \ket{ x_1, x_2, \ldots x_n} = \ket{ \lambda_k(x) , x_2, \ldots x_n}.
\end{equation}
We note that if $a_k$ is even then $T_1^{a_k}=S_1^{a_k/2}$ is a Clifford and the whole circuit is Clifford.  Whereas if $a_k$ is odd then $T_1^{a_k}=T_1 S_1^{(a_k-1)/2}$ and only a single $T$ gate is used.  For a phase polynomial $f$ with many terms we have
\begin{equation}
    U_f = \prod_k U_{\lambda_k}
\end{equation}
and so the $T$-count for the associated circuit is equal to the number of odd valued $a_k$.  If all values are even, then the unitary is Clifford.

We use this insight to split the unitary $U_f$ into a Clifford and non-Clifford part.  For each $a_k$ coefficient, we define $b_k \in \mathbb{Z}_4$ and $c_k \in \mathbb{Z}_2$ such that $a_k = 2 b_k + c_k$. Notice that $c_k=1$ if and only if $a_k$ is odd valued.  Then we have that $f = g + 2 h$ where $g$ and $h$ are the functions
\begin{align}
		g(x) & = \sum_{c_k \neq 0} c_k \lambda_k(x) \pmod{8}, \\
		h(x) & = \sum_{b_k \neq 0} b_k \lambda_k(x) \pmod{8}.
\end{align}
We see that $U_f=U_{g+2h}=U_{g}U_{2h}$ where $U_{2h}$ is a Clifford unitary. The non-Clifford part is $U_{g}$ and all the terms have odd valued co-coefficients, so the number of terms in $g$ gives an upper bound on $\tau(U_g)$ as discussed earlier. It follows that if the function $g$ has $m$ (odd-valued) terms then the state can be prepared using $m$ many $T$ gates or states.  For any given unitary $U_g$ there is an equivalence class of different functions $g$ that all result in the same unitary but with different numbers of terms. Herein we assume that $g$ is the optimal representative with the fewest number of terms, which we denote $\tau(U_g)$.  Design of compilers for finding this optimal function is an ongoing research area with several useful heuristics~\cite{amy2016t,campbell2017unified,heyfron2018efficient}.  Furthermore, there is a binary matrix $P$ description of $g$ (as defined above) with a number of columns also equal to $\tau(U_g)$. A trivial, but relevant, example is $U=T^{\otimes n}$ for which $P=\id_n$ and $\tau(T^{\otimes n})=n$.  

The next important step is that given a unitary $U_g$ we may also be able to remove terms from $g$ by applying inverse $T$ gates.  
More generally, given two such unitaries $U_g$ and $U_{g'}$ with phase polynomials $g$ and $g'$, we have that $U_{g'}=U_{g}U_{\Delta}$ where $\Delta= g - g'$.  Therefore,
\begin{equation}
    \ket{U_{g'}} = U_{\Delta} \ket{U_g} ,
\end{equation}
and
\begin{equation}
    \ket{T}^{\otimes \tau(U_\Delta)} \ket{U_{g'}} \rightarrow  \ket{U_g}.
\end{equation}
The number of $T$ states needed is equal to $\tau(U_\Delta)$, which in turn is equal to the number of terms where $g$ and $g'$ differ.

Given any $P$ we can always bring it into row-reduced echelon form using a CNOT circuit, by virtue of the arguments presented in Sec. III of Ref.~\cite{HowardCampbell2016}. Then
\begin{equation} 
	 P =\left( \begin{array}{cc}
	 \id_r & A \\
	 0 & 0 \end{array}
	   \right) ,
\end{equation}	
where $\id_r$ is an identity matrix of size equal to $r:=\mathrm{rank}(P)$.  If $P$ is full rank the additional 0 padding is not present.  Note that if $P$ has any 0 rows then the unitary acts trivially on the corresponding qubits leaving them in the $\ket{+}$ state and so $r \geq \mu(U)$.  Using our earlier argument, we can always remove from $P$ the columns corresponding to the matrix $A$ using a number of $T$ states equal to the number of columns in $A$.  Since $A$ has $\tau(U_g)-r$ columns, this requires the same quantity of $T$ states.  The resulting $U_{g'}$ has $P' = \id_r$ (with possibly some 0 row padding) which corresponds to $r$ copies of $T$ states.   Therefore,  we can perform
\begin{equation}
	\ket{U_g}\ket{T}^{\otimes (\tau(U_g)-r)} \rightarrow \ket{T}^{\otimes r}.
\end{equation}
If $r = \mu(U_g)$ then we have the result of the theorem.  If $r>\mu(U_g)$ then we actually have a stronger result and the statement of the theorem still follows.
\end{proof}

The interesting cases of \theo{PhasePolyCat} are those where $\ket{U} \rightarrow  \ket{T}^{\otimes r}$ is forbidden by virtue of the ring argument as presented in \theo{number-field-constraint}. We make the following observation
\begin{claim}\label{EvenWeightClaim}
Let $U$ be a diagonal unitary from the 3$^\mathrm{rd}$ level of the Clifford hierarchy with phase polynomial matrix $P$.  If all rows of $P$ have even Hamming weight then $U\ket{+}^{\otimes n} \nrightarrow  \ket{T}$.
\end{claim}
To see this, note that every diagonal unitary from the 3$^\mathrm{rd}$ level of the Clifford hierarchy is (up to Cliffords) a product of $T$, $CS$ and $CCZ$ gates~\cite{campbell2017unified}. In the special case that $U$ has phase polynomial matrix with even rows, then the unitary is a product of $CS$ and $CCZ$ gates (see App.D of Ref.~\cite{HowardCampbell2016}).  Such a unitary has elements in the ring $\q\at{i}$ and so $U\ket{+}^{\otimes n} \nrightarrow  \ket{T}$ follows (as discussed in \sec{catalysis}).  Though this transform is impossible without a catalyst, \theo{PhasePolyCat} gives a recipe for designing catalytic protocols and we next discuss some concrete examples.

 For any $n \geq 2$, we define $W_n$ as the unitary with phase polynomial matrix 
 \begin{equation}
     P_n =( \id_n, 1) = \left(  \begin{array}{ccccc}
     1 & 0 & & 0 & 1\\
     0 & 1 & & 0 & 1\\
     &   & \ddots & 0 & 1 \\
     0 & 0 & 0 & 1  & 1
     \end{array} 
     \right) ,
\end{equation}
which is the identity matrix padded with an all-one column. More explicitly, we have $W_n$ 
\begin{equation}
     W_n = \sum_{x}  \exp(i \pi g(x) / 4   ) \kb{x}{x} ,
\end{equation}
with
\begin{equation}
     g(x) = ( \oplus_{i=1}^n x_i )   + \sum_{i=1}^n x_i,
\end{equation}
where the $\oplus$ sum is performed modulo 2.

% \begin{figure}[ht!]
% \center
% 	\hspace*{-0.65cm}\includegraphics[scale=0.7]{Wn_catalysis.pdf}
% 	\caption{By applying the rightmost circuit to the $\ket{+}^{\otimes n}$ state, $\ket{W_n}$ is converted into $n$ copies of $\ket{T}$, using an additional $T^\dagger$-gate. 	In terms of resource states $\ket{W_n}\ket{T} \longrightarrow \ket{T}^{\otimes{n}}$, or equivalently $\ket{W_n}\catalyse{\ket{T}} \ket{T}^{\otimes{n-1}}$. The leftmost circuit identity depicts how the unitary $W_n$-gate can be implement using a minimal (i.e., $\tau(W_n)=n+1$) number of $T$-gates. }
% 	\label{fig:Wn_Catalysis_b}
% \end{figure}
With the machinery of phase polynomials and $P$ matrices established, it is now straightforward to prove \lemm{tau-for-Wn}, 
% \begin{lem*}
% %\label{lem:tau-for-Wn}
% $\tau(W_n)=n+1$.
% \end{lem*}

\tauforWn
\begin{proof}
Since $P$ has a width of $n+1$ columns, we have $\tau(W_n) \leq n+1$.  The only full rank phase polynomial matrices that are square give a unitary that is Clifford equivalent to $T^{\otimes n}$, since this is not the case we conclude $\tau(W_n) = n+1$.
\end{proof}
Notice that every row of $P_n$ is even weight and so by Claim~\ref{EvenWeightClaim} we know $\ket{W_n} \nrightarrow  \ket{T}$. Since $P$ is a full rank matrix, we have $\mu(\ket{W_n})=n$. Therefore, $2\mu( \ket{U} ) - \tau(U) =n-1$ and by \theo{PhasePolyCat} we conclude that 
\begin{equation}
 \ket{W_n} \implies \ket{T}^{\otimes n-1}.
\end{equation}  
These are the most illuminating examples that one can obtain from \theo{PhasePolyCat} because assuming $U \neq T^{\otimes n}$ we know $\tau(U)>\mu(U)$ and then $\tau(U)=\mu(U)+1$ leads to the best possible catalysis protocols. 

At first glance, the $W_n$ unitaries may look unfamiliar.  However, $W_2$ has the same non-Clifford part as $CS$ and so they are equivalent up to Cliffords.  The $W_2$ example is also equivalent to the catalysis protocol first observed by Campbell~\cite{Campbell2010}. For $W_3$, we have that the state $\ket{W_3}$ is Clifford equivalent to $\ket{CCZ}$ and so $\ket{CCZ} \implies \ket{T}^{\otimes 2}$, which is the catalysis protocol observed by Gidney and Fowler~\cite{gidney2018efficient}.  The Clifford equivalence of $\ket{W_3}$ and $\ket{CCZ}$ may be not obvious and so we comment further on this.  We have that $CNOT_{3,2} W_3 CNOT_{3,2}$ has the same phase polynomial matrix as $V=CCZ_{1,2,3}CS_{2,3}$.  Furthermore, Cody Jones~\cite{Jones2012} showed that $V$ can be used to synthesize $CCZ$, which establishes the equivalence.

% We also have that the $W_n$ states can be reduced in the sense that
% \begin{equation}
% \label{WnReduction}
%   \ket{ W_n} \rightarrow  \ket{ W_{n-1} }.
% \end{equation}
% This conversion is achieved by first measuring the last qubit in the computational basis. If one obtains the ``0" outcome, we immediately have the state $\ket{W_{n-1}} \ket{0}$ and so just discard the last qubit.  In the case of a ``1" outcome, then a Clifford correction 
% \begin{equation}
%  C = \sum_x  i^{1  \oplus_{i=1}^n x_i} \kb{x}{x} \label{eq:Correction},
% \end{equation}
% is required. Performing this correction and discarding the last qubit we again obtain $\ket{W_{n-1}}$.
% \begin{figure}[ht!]
% \center
% 	\includegraphics[scale=0.85]{Wn_level_reduction.pdf}
% 	\caption{By applying the circuit on the left to the $\ket{+}^{\otimes n}$ state, $\ket{W_n}$ is converted into $\ket{W_{n-1}}$, where, depending on the $Z$ measurement result, a Clifford correction $C$ as in \eq{Correction} may be required. }
% 	\label{fig:Wn_Catalysis_b}
% \end{figure}
% Many of the interconversions described here are depicted in \fig{ConversionGraphic}.

\subsection{Lower bound reduction from the modular adder to the multiply-controlled Z state}
\label{app:Gidney-bound-reduction}
Here we reproduce the argument from \cite{Gidney2018} that the modular adder which acts on a pair of $n$-qubit registers can be used to produce the $n$-controlled $Z$ state $\ket{C^n Z}$.

The argument proceeds in two steps. 
First we show that the controlled modular increment circuit $C\text{Inc}_{n}$ can be used to produce a $\ket{C^nZ}$ state using Clifford operations.
Second, we show that the controlled modular increment circuit can be implemented using the modular adder.
The controlled modular increment circuit $\text{Inc}_{n}$ acts as follows on computational basis states
\begin{eqnarray}
C\text{Inc}_{n}:~\ket{j}\ket{a+j} &\mapsto& \ket{j}\ket{a + j \mod 2^n},
\end{eqnarray}
where $j =0,1$ and $a = 0, 1, \dots 2^n-1$ is stored using binary on an $n$-qubit state.
One can implement the controlled modular increment circuit as show in \fig{IncrementInduction}(b).

\begin{figure}[ht!]
    \centering (a)\includegraphics[scale=1.1]{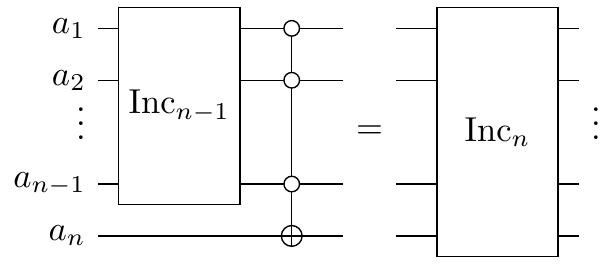}
    (b)\includegraphics[scale=0.9]{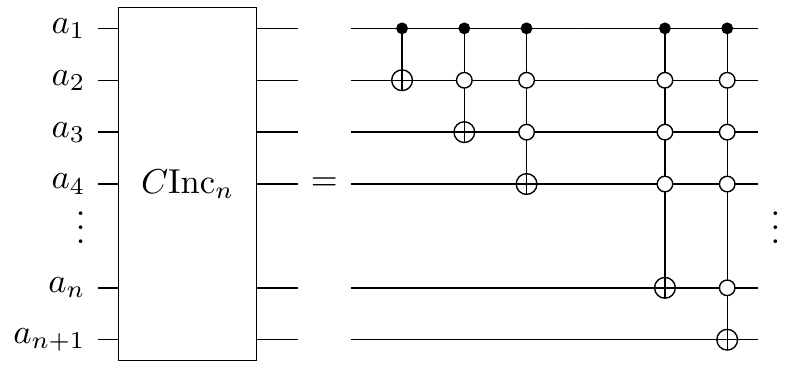}
    (c)\includegraphics[scale=1.0]{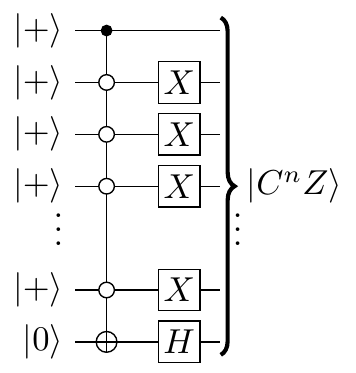}
    (d)\includegraphics[scale=0.9]{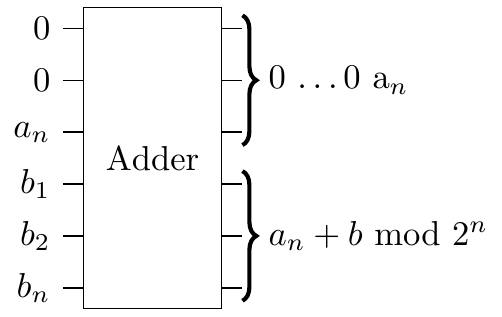}
	\caption{(a) An inductive argument shows that the modular increment circuit $\text{Inc}_n$ is built from a sequence of multiply-controlled not gates (where the control is activated on the $\ket{0}$ state of the control qubits rather than the $\ket{1}$ state). 
	(b) The controlled modular increment circuit $C\text{Inc}_n$ is then implemented by including an additional control for each gate on the additional qubit which controls whether or not $\text{Inc}_n$ is applied.
	(c) When applied to the state $\ket{+}^{\otimes n}\ket{0}$, all but the last of the gates in the circuit for $C\text{Inc}_n$ annihilate, and the resulting state is Clifford-equivalent to $\ket{C^n Z}$.
	(d) One can implement $C\text{Inc}_n$ with the modular adder by using the last qubit of the first input of the adder as the control and setting the other qubits of the first input to $\ket{0}$.
	}
	\label{fig:IncrementInduction}
\end{figure}

Consider applying $C\text{Inc}_n$ to the state $\ket{+}^{\otimes n}\ket{0}$.
Since the target of all the controlled gates is the $X$ gate, those which have a target qubit in the state $\ket{+}$ (an eigenstate of $X$) have no action and can be removed from the circuit so that only the last $n$-controlled not gate remains.
The resulting state is clearly Clifford-equivalent to $\ket{C^n Z}$ as shown in \fig{IncrementInduction}(c).

We have seen that if one can implement $C\text{Inc}_n$, it is possible to produce the state $\ket{C^n Z}$ by applying it to a stabilizer state and using Clifford gates.
Now note that one can implement $C\text{Inc}_n$ with the modular adder in \fig{IncrementInduction}(d) by using the last qubit of the first input of the adder as the control and setting the other qubits of the first input to $\ket{0}$.

\subsection{Canonical form for post-selected stabilizer computations}
\label{app:canonical-form}
%auto-ignore
%!TEX root = resource-states-lower-bounds.tex

The goal of this section is to establish the canonical form for post-selected stabilizer computations described in \theo{canonical-form}.

First we show we can assume both the input and output states of \theo{canonical-form} have trivial stabilizer, i.e. defined on a number of qubits equal to their nullity, due to the following proposition:
\begin{prop} \label{prop:trivialize-stabilizer}
Let $\ket\phi$ be an $m$-qubit state and let $\abs{\Stab\at{\ket{\phi}}} = 2^r$ for $r>0$,
then there exist a Clifford unitary $C$ such that 
$\ket{\phi} = C\at{\ket{0}^{r} \otimes \ket{\phi'} }$ 
where $\ket{\psi'}$ has a trivial stabilizer.
\end{prop}
\begin{proof} 
Recall that for any commutative sub-group $\G$ of the Pauli group that does not contain $-I$
there exist a Clifford $C$ such that 
$C \G C^\dagger = \ip{Z_1 , \ldots, Z_m}$ \cite{Cleve1997}. 
Choosing $\G = \Stab\at{\ket{\phi}}$ lets us find the required Clifford, 
because $\Stab\at{C\ket{\phi}} = C\,\Stab\at{\ket{\phi}} C^{\dagger}$. 
\end{proof}
Given this, \theo{canonical-form} is inferred from the following theorem (identical to \theo{canonical-form} but in which the input and output states have trivial stabilizer) which we prove in the remainder of this section:

\begin{thm} \label{thm:canonical-form-app}
Consider a post-selected stabilizer circuit with $n$-qubit input state $\ket{\psi_\mathrm{in}}$ and $m$-qubit output state $\ket{\psi_{\mathrm{out}}}$, where
$m \leq n$ and where $n=\nu(\ket{\psi_{\mathrm{in}}})$ and $m=\nu(\ket{\psi_\mathrm{out}})$. 
Then there exists a set of $k = n-m$ independent commuting Pauli operators $P_1,\ldots,P_{k}$ and
a Clifford unitary $C$ such that 
\[
\ket{\psi_\mathrm{out}}\otimes\ket{S}
\propto
C M_{P_1} \ldots M_{P_{k}} \ket{\psi_\mathrm{in}},
\]
where $\ket{S}$ is a stabilizer state and where $M_{P}$ is the projector on the $+1$ eigenspace of $P$.
\end{thm}

Note that if $m=n$, the states $\ket{\psi_{\mathrm{in}}}$ and $\ket{\psi_{\mathrm{out}}}$ can be obtained from one another by applying a Clifford unitary. 
An interesting feature of \theo{canonical-form-app} is that if we wish to enumerate all possible stabilizer circuits that can act on a particular input state, we need only to consider Pauli measurements that commute with each other. 
The following proposition gives some intuition for why this is the case.

\begin{prop}\label{prop:measurement-as-unitary}
Let $\ket{\psi}$ be an $n$-qubit state and let $P$ be a $n$-qubit Pauli operator such that there exists $Q \in \Stab\at{\ket{\psi}}$ that anti-commutes with $P$. 
Then the measurement of $P$ is equivalent to randomly applying the Clifford unitaries $\at{I+PQ}/\sqrt{2}$ or $\at{I-PQ}/\sqrt{2}$ with equal probability.
\end{prop}
\begin{proof}
First check that measuring $P$ gives outcome $+1$ or $-1$ with probability $1/2$. 
Indeed the probability of measuring $+1$ is $\bra{\psi}\at{I+P}\ket{\psi}/2$ and it is equal to:
\[
\bra{\psi}Q\at{I+P}Q\ket{\psi}/2 = \bra{\psi}\at{I+QPQ}\ket{\psi}/2 = \bra{\psi}\at{I-P}\ket{\psi}/2
\]
Therefore the probability of measuring $+1$ and $-1$ is the same and their sum is one. 
Therefore the probability of each measurement outcome in $1/2$.
This means what in case of $+1$ outcome the state becomes $\at{I+P}\ket{\psi}/\sqrt{2}$ which is equal to $\at{I+PQ}\ket{\psi}/\sqrt{2}$ which is a Clifford unitary. Similarly in case of $-1$ outcome we have applied $\at{I-PQ}\ket{\psi}/\sqrt{2}$.
\end{proof}

The next step towards the proof of \theo{canonical-form-app} is to rewrite an arbitrary quantum circuit consisting of 
Clifford unitaries and post-selected Pauli measurements into a canonical form. 
This is the subject of the next lemma: 

\begin{lem} \label{lem:weak-canonical-form}
Let $\ket{\psi_\mathrm{out}}$ be a non-zero $n$-qubit state that can be obtained from an $n$-qubit state $\ket{\psi_\text{in}}$ using Clifford unitaries and post-selected Pauli measurements.
Then there exists a Clifford unitary $C$ and 
a commutative sub-group $\G$ of the Pauli group that does not contain $-I$ with generators $P_1,\ldots,P_{m}$ such that:
\begin{itemize}
  \item $\ket{\psi_\mathrm{out}} \propto C M_{P_m} \ldots M_{P_1} \ket{\psi_\mathrm{in}}$,
  \item the group generated by $\G$ and $\Stab \at{\ket{\psi_\mathrm{in}}}$ is a commutative sub-group of $n$-qubit Pauli group
  and does not contain $-I$,
  \item none of the $P_k$'s are in $\Stab \at{\ket{\psi_\mathrm{in}}}$.
\end{itemize}
\end{lem}
\begin{proof}
We write the circuit of Clifford unitaries and post-selected Pauli measurements as: 
\[
 \ket{\psi_\text{out}} \propto C_{m'+1} M_{P'_{m'}} C_{m'} M_{P'_{m'-1}} C_{m'-2} \ldots C _2 M_{P'_1} C_{1} \ket{\psi_\text{in}},
\]
where $P'_k$ are $n$-qubit hermitian Pauli operators and $C_k$ are $n$-qubit Clifford unitaries. 
Next we observe that the projector $M_{P} = (I+P)/2$ transforms into another Pauli projector under conjugation by a Clifford unitary:
$
 C^{\dagger} M_{P} C = M_{C^{\dagger} P C} = M_{P'}
$
where $P'$ is an $n$-qubit hermitian Pauli operator because Clifford unitaries map Pauli matrices to Pauli matrices. 
By repeatedly applying this observation we can push each Clifford unitary to the end of computation and therefore:
\begin{equation*}
 \ket{\psi_\text{out}} = C' M_{P''_m} M_{P''_{m-1}} \ldots M_{P''_1} \ket{\psi_\text{in}},
\end{equation*}
where each $P''_k$ is an $n$-qubit hermitian Pauli operator and $C'$ is an $n$-qubit Clifford unitary.

Next we describe how to construct $P_1,\ldots,P_m$ out of $P''_1,\ldots,P''_{m'}$.
Suppose $P''_1$ anti-commutes with some $Q$ from $\Stab \at{\ket{\psi_\text{in}}}$. 
In this case we can replace $M_{P''_1}$ with the Clifford unitary $\at{I + P_1''Q}/\sqrt{2}$
as shown in \propos{measurement-as-unitary}.
Then we pull this Clifford unitary through the following measurements and absorb it into the Clifford gate applied at the end. 
If $P''_1$ commutes with  $\Stab \at{\ket{\psi_{in}}}$, there are several cases we need to consider. 
If $P''_1$ is in $\Stab \at{\ket{\psi_{in}}}$ than $M_{P''_1}$ can be removed from the canonical form,
if $-P''_1$ is in $\Stab \at{\ket{\psi_{in}}}$ then $\ket{\psi_\text{out}}$ is the zero state.
The remaining case is that $P''_1$ commutes with $\Stab \at{\ket{\psi_{in}}}$ but does not belong to it.
In this case we set $P_1$ to be $P''_1$.
We have ensured that $P_1$ and $\Stab \at{\ket{\psi_{in}}}$ generate commutative sub-group of a Pauli group that does not contain $-I$ and that $P_1$ is not in $\Stab \at{\ket{\psi_{in}}}$.
We repeat the described procedure for $P''_2, \ldots, P''_{m'}$ and get the required result.
\end{proof}
\lemm{weak-canonical-form} implies that if the state $\ket{\psi_\text{in}}$ can be transformed into the state $\ket{\psi_\text{out}}$
by post-selected stabilizer operations, then for some $m$, $m'$ and $n$:
\[
\ket{0}^{\otimes m'} \otimes\ket{\psi_{out}} \propto C M_{P_1} \ldots M_{P_{n}} \ket{0}^{\otimes m} \otimes \ket{\psi_{in}}.
\] 
To prove \theo{canonical-form-app} it remains to get rid of the ancillary qubits on the right side of this equation.
The following result is a key to this.
\begin{lem} \label{lem:ancilla-free}
Let $\ket{\phi}$ and $\ket{\psi}$ be two states such that for a Clifford unitary $C$ and $n > 0$:
\begin{equation} \label{eq:clifford-equivalence-with-anccilla}
\ket{0}^{\otimes n}\otimes\ket{\psi} = C (\ket{0}^{\otimes n}\otimes\ket{\phi}),
\end{equation}
then there exists a Clifford unitary $C_0$ such that $\ket{\psi} = C_0 \ket{\phi}$.
\end{lem}
We postpone the proof of \lemm{ancilla-free} and first complete the proof \theo{canonical-form-app} using \lemm{ancilla-free}.7\begin{proof}[Proof of \theo{canonical-form-app}]
Using \lemm{weak-canonical-form} we conclude that there exist a Clifford unitary $C$ and commuting Pauli operators $P_1, \ldots, P_s$ such that: 
\begin{equation} \label{eq:canonical-form-with-zeros}
\ket{0}^{\otimes m'}\otimes \ket{\psi_\mathrm{out}}
\propto
C' M_{{P_{s}}} \ldots M_{{P_1}} \ket{0}^{\otimes m} \otimes \ket{\psi_\mathrm{in}}.
\end{equation}
Next we show that operators $P_k$ can be replaced with operators $Q_k$ supported only on the last $n-m$ qubits.
Indeed, each of operators $P_k$ must commute with the stabilizer of $\ket{0}^{\otimes m}\otimes\ket{\psi_{in}}$ which consists of all possible operators $Z^{a_1} \otimes \ldots \otimes Z^{a_m}\otimes I_{2^{n-m}}$ for $a_j \in \set{0,1}$. This implies that each $P_k$ can be written as a tensor product 
\[
Z^{a_{k,1}}\otimes\ldots\otimes Z^{a_{k,m}}\otimes Q_k, \text{ for some } a_{k,j} \in \set{0,1}
\]
For this reason, applying $M_{P_k}$ to a state $\ket{0}^{\otimes m}\otimes\ket{\psi}$ is equivalent to
 applying $I_{2^m} \otimes M_{{Q_k}}$. 
Let $\G$ be a group generated by $Q_1, \ldots, Q_{s}$.
We rewrite Equation~(\ref{eq:canonical-form-with-zeros}) as:
\[
\ket{0}^{\otimes m'}\otimes \ket{\psi_\text{out}}
=
C' \at{ \ket{0}^{\otimes m} \otimes \at{ M_{Q_s} \ldots M_{Q_1} \ket{\psi_\text{in}}} }
\]
We remove first $m$ qubits initialized to $\ket{0}$ from the equation above by using \lemm{ancilla-free}.

Note that after measuring $Q_1,\ldots,Q_s$ on $\ket{\psi_\text{in}}$ the stabilizer of the result can be strictly
bigger then the group generated by $Q_1,\ldots,Q_s$.
We can just add remaining generators to the list of $Q_1,\ldots,Q_s$ to make sure that there are $m-m'$ of them.
If $m-m'$ is zero, then $s$ must be zero and input and output states must be Clifford equivalent.
This completes the proof.
\end{proof}

\subsubsection{Decoupling stabilizer states}
Here we prove \lemm{ancilla-free}.
It relies on several simpler results, which we separate into propositions and lemmas after the main proof of \lemm{ancilla-free}.
\begin{proof}[Proof of \lemm{ancilla-free}]
Note that it is sufficient to consider the case when $\ket{\phi}$ and $\ket{\psi}$ have a trivial stabilizer.
Our proof strategy consists of two steps.
First we show that in the Equation~(\ref{eq:clifford-equivalence-with-anccilla})
we can replace unitary $C$ with a Clifford unitary $C_n$
such that $C_n$ commutes with Pauli matrices $Z_k$ for $k$ from $1$ to $n$. 
Second we show that the commutation of $C_n$ and $Z_1$ implies that 
\begin{equation} \label{eq:clifford-decomposition}
C_n = \ket{0}\bra{0}\otimes C_{n-1} + \ket{1}\bra{1} C'_{n-1}, \text{ where } C_{n-1} \text{ is a Clifford}.
\end{equation}
This implies that $\ket{\psi}\otimes\ket{0}^{n-1}$ and  $\ket{\phi}\otimes\ket{0}^{n-1}$ are Clifford equivalent.
Proceeding by induction completes the proof.

Let us now construct a Clifford $C_n$ with required properties.
Consider Pauli matrices $P_a = Z^{a(1)}_1\otimes \ldots \otimes Z^{a(n)}_n$ where each $a(j)$ is either zero or one.
These are exactly the matrices that stabilize $\ket{0}^{\otimes n}$.
For each $a$, there exist $b$ such that $C P_a C^{\dagger} = P_b$ because
the stabilizer of $\ket{0}^{n}\otimes\ket{\psi}$ and $\ket{0}^{n}\otimes\ket{\phi}$ is exactly the set 
$\set{P_a : a \in \set{0,1}^n }$.
There exist a Clifford $D$ composed only of CNOT gates acting on the first $n$ qubits such that 
$D C P_a C^{\dagger} D^{\dagger} = P_a$.
Defining $C_n = D C$ ensures that $C_n$ commutes with Pauli $Z_k$.
Because $D$ is composed only of CNOT gates acting on first $n$ qubits
$ \ket{0}^{n}\otimes\ket{\psi} = D \ket{0}^{n}\otimes\ket{\psi}$.
This shows that Equation~(\ref{eq:clifford-equivalence-with-anccilla}) holds with $C$ replaced by $C_n$.

Now let us show that $C_n$ is of the form given by Equation~(\ref{eq:clifford-decomposition}).
Note that $C_n$ commutes with Pauli $Z$ on the first qubit,
therefore by \propos{z-commutator}
unitary $C_n$ can be written as $\ket{0}\bra{0}\otimes C_{n-1} + \ket{1}\bra{1}\otimes C'_{n-1}$. 
To show that $C_{n-1}$ must be a Clifford unitary we rely on \lemm{clifford-property}.
Indeed, for any positive $d$, $C_{n-1}\otimes I_d$ maps stabilizer states to stabilizer states because $C_n \otimes I_d$
is a Clifford that maps stabilizer states of the from $\ket{0}\otimes\ket{\alpha}$ to stabilizer states of the form 
$\ket{0}\otimes\ket{\beta}$. 
\end{proof}

The following proposition is a well-known result from the linear algebra and we provide the proof for completeness.
\begin{prop} \label{prop:z-commutator}
Let $U$ be a unitary that commutes with a Pauli $Z$ matrix on the first qubit 
then $U = \ket{0}\bra{0}\otimes U_{00} + \ket{1}\bra{1} \otimes U_{11}$.
\end{prop}
\begin{proof}
Note that the fact that $U$ commutes with $Z_1$ implies that $U$ commutes with matrix
$M = \lambda_0 \ket{0}\bra{0} \otimes I + \lambda_1 \ket{1}\bra{1} \otimes I$
for arbitrary complex numbers $\lambda_0, \lambda_1$.
Let us write $U = \sum_{a,b \in \set{0,1} } \ket{a}\bra{b} \otimes U_{ab}$. 
Next expand $M U$ and $U M$ as: 
\begin{align*}
M U
=
\lambda_0 \ket{0}\bra{0} \otimes U_{00} + \lambda_0 \ket{0}\bra{1} \otimes  U_{01} + 
\lambda_1 \ket{1}\bra{0} \otimes U_{10} + \lambda_1 \ket{1}\bra{1} \otimes U_{11} \\
U M = 
\lambda_0 \ket{0}\bra{0} \otimes U_{00} + \lambda_0 \ket{1}\bra{0} \otimes  U_{10} + 
\lambda_1 \ket{0}\bra{1} \otimes U_{01} + \lambda_1 \ket{1}\bra{1} \otimes U_{11}
\end{align*}
We see that equality $UM = MU$ is only possible when $U_{01}$ and $U_{10}$ are both zero.
\end{proof}
The next proposition if a convenient characterization of Pauli matrices that we use 
to establish a necessary condition for a unitary to be a Clifford later in this section.
\begin{prop} \label{prop:pauli-matrix-property}
Let $M$ be an $n$-qubit matrix such that $\mathrm{Tr}\at{MM^\dagger} = 2^n$
and for every Pauli matrix $P$ from $\set{I,X,Y,Z}^{\otimes n}$ the trace $\mathrm{Tr}\at{M P}$ is
either $0$ or $\pm 2^n$, then $M$ or $-M$ is a Pauli matrix.
\end{prop}
\begin{proof}
Recall that the set $P_n = \set{I,X,Y,Z}^{\otimes n}$ is an orthogonal basis in the vector space of $n$-qubit matrices
with respect to inner product $\ip{A,B} = \mathrm{Tr}\at{AB^{\dagger}}$.
Matrix $M$ can be represented as a sum $\sum_{P \in P_n} P\ip{M,P}/\ip{P,P}$.
In particular the square norm of $M$ is $2^n = \ip{M,M} = \sum_{P \in P_n} \abs{\ip{P,M}}^2 / \ip{P,P}$. 
The equality is only possible when there is exactly one Pauli matrix $P$ such that $\ip{P,M} = \pm 2^n$.
\end{proof}
It is well-known that Clifford unitaries map stabilizer states to stabilizer states. 
One can show that this is also a necessary condition for unitary to be a Clifford.
Here we prove a slightly weaker result.
\begin{lem} \label{lem:clifford-property}
Let $U$ be an $n$ qubit unitary such that unitary $U \otimes I_n$ maps stabilizer states to stabilizer states
then $U$ is a Clifford unitary. 
\end{lem}
\begin{proof}
We will exploit the fact that the Choi state of $U$ must be a stabilizer state.
Recall that the Choi state of unitary $U$ is the result of applying $U\otimes I_n$ to $n$ Bell states.
Bell state $(\ket{00}+\ket{11})/\sqrt{2}$ is stabilized by $X\otimes X$ and $Z \otimes Z$ and
its density matrix is proportional to $\sum_{P \in \set{I,X,Y,Z}} P\otimes P$.
The density matrix of $n$ Bell states is proportional to $\sum_{P \in \set{I,X,Y,Z}^{\otimes n}} P\otimes P$.
The density matrix of the Choi state of $U$ is equal to: 
\[
\rho = \frac{1}{2^{2n}} \sum_{P \in \set{I,X,Y,Z}^{\otimes n}} U P U^{\dagger}\otimes P
\]
Let us now fix $P$ and show that $M = U P U^\dagger$ is a Pauli matrix
by using \propos{pauli-matrix-property}.
First note that $\mathrm{Tr}\at{M M^\dagger} = 2^n$. 
Next observe that for arbitrary Pauli matrix $Q$ the $\mathrm{Tr}\at{M Q}$ must be either $0$ or $2^n$.
Note that $\mathrm{Tr}\at{\rho \at{Q \otimes P}} = 2^{-n} \mathrm{Tr}\at{M Q}$.
On the other hand, because $\rho$ is a density matrix of a stabilizer state,
value $\mathrm{Tr}\at{\rho \at{Q \otimes P}}$ can only be $0,1,-1$.
\end{proof}

\subsection{Conversion protocols for dyadic rational powers of T gate}
\label{app:dyadic-powers-of-t}
In this section we look at the creation of many copies of states $\ketsm{\pi j/2^d}$ 
which include $\ket{T}$, $\ketsm{\sqrt{T}}$ when $j=1$, $d=3,4$.
We show that in the limit of creating many copies of the same state less than one $CCZ$ gate is required per state.
We start with generalizations of some of the results discussed in \sec{adder-conversion-protocols} in context of producing $\ketsm{\sqrt{T}}$ states.

\begin{prop} \label{prop:double-angle}
Let $\theta$ be a real number and let $k$ be a positive integer. 
The parallel application of $2k+1$ unitaries $\exp\at{i\theta \ket{1}\bra{1}}$ can be achieved by stabilizer operations with measurements that have probability $50\%$,
one unitary $\exp\at{i\theta \ket{1}\bra{1}}$,
$k$ $CCZ$ gates and 
the parallel application of $k$ unitaries $\exp\at{i 2\theta \ket{1}\bra{1}}$.
\end{prop}
\begin{proof}
We proof the proposition by induction on $k$. Let us start with the base case $k=1$. 
Using a circuit similar to 
\fig{three-sqrt-t-gates}
we can apply three unitaries $\exp\at{i\theta \ket{1}\bra{1}}$ in parallel by using one ancilla, one $CCZ$ gate, one unitary $\exp\at{i 2\theta \ket{1}\bra{1}}$ and one unitary  $\exp\at{i \theta \ket{1}\bra{1}}$. 
Suppose now that we have established the proposition for $k=j$. 
Let us prove the result for $k=j+1$. 
We need to apply $2j+3$ unitaries $\exp\at{i\theta \ket{1}\bra{1}}$ in parallel. 
We apply first three of them using a circuit similar to 
\fig{three-sqrt-t-gates}. 
The circuit will use one $\exp\at{i 2\theta \ket{1}\bra{1}}$ gate, one ancilla, one $CCZ$ gate and one $\exp\at{i\theta \ket{1}\bra{1}}$ gate. 
We notice that remaining $2j$ unitaries $\exp\at{i\theta \ket{1}\bra{1}}$ can be applied in parallel with the newly introduced one. 
A special case of the induction step is shown on \fig{many-sqrt-t-states}.
Using the induction hypothesis we see that in total we will need $j+1$ ancillary qubits, $j+1$ $CCZ$ gates, $j+1$ unitaries $\exp\at{i 2\theta \ket{1}\bra{1}}$ and one unitary $\exp\at{i\theta \ket{1}\bra{1}}$. 
This completes the proof.
\end{proof}

Next we apply above proposition to obtain a protocol that uses catalysis to apply rotations $R(\theta)=\exp\at{i\theta \ket{1}\bra{1}}$ by angle $\theta = \pi j / 2^d$ for positive integer $d \ge 3$ and odd integer $j$. 

\begin{prop} \label{prop:power-reduction}
Let $k, d$ be positive integers and let $j$ be an odd integer.
The parallel application of $2k$ unitaries
$R\atsm{\pi j / 2^d }$ can be achieved by stabilizer operations
with measurements that have probability $50\%$,
using resource state $\ketsm{ \pi j / 2^d}$ as a catalyst,
$k$ $CCZ$ gates and
the parallel application of $k+1$ unitaries  $R\atsm{\pi j / 2^{d-1}}$.
\end{prop}
\begin{proof}
To apply the required unitary transformation we use the protocol described in \propos{double-angle}
with the last input set to $\ket{+}$ state. 
This will ensure that we apply $2k$ unitaries $R\atsm{\pi j / 2^d}$ in parallel and
produce one resource state $\ketsm{ \pi j / 2^d}$.
To apply one gate $R\atsm{\pi j / 2^d}$ needed by protocol from \propos{double-angle} 
we use resource state injection protocol. 
The protocol consumes one state  $\ketsm{ \pi j / 2^d}$ and
with probability $50\%$ requires one application of $R\atsm{\pi j / 2^{d-1}}$. 
The gate $R\atsm{\pi j / 2^{d-1}}$ used in the injection protocol can be applied in parallel with the rest of $R\atsm{\pi j / 2^{d-1}}$ applied as a part of protocol from \propos{double-angle}.
Therefore in total we will need to apply at most $k+1$ unitaries $R\atsm{ \pi j / 2^{d-1}}$ in parallel. 
We use the same number of $CCZ$ gates as in \propos{double-angle} which is equal to $k$.
\end{proof}

In \sec{adder-conversion-protocols} we presented a special case of the above proposition for $j=1$ and $d=4$. 
Next we apply above proposition recursively to obtain a family of conversion protocols for resource states $\ketsm{ \pi j / 2^d}$
that use states $\ketsm{\pi j / 2^d},\ldots,\ket{\pi j / 2^2}$ as catalysts together with ${CCZ}$ gates. 
\dyadicpowersconversion
\begin{proof}
We prove the theorem by induction on $d$. 
The base case $d=1$ is true because $R\atsm{ \pi j / 2}$ are Clifford gates 
and require zero ${CCZ}$ gates to be applied.
Suppose now that we have shown the result for $d=d'$ and let us prove the theorem for $d=d'+1$. 
We need to apply $a_{d'+1,k}$ gates $R\atsm{j / 2^{d'+1}}$. 
According to \propos{power-reduction} we achieve this using 
$a_{d'+1,k}/2$ $CCZ$ gates,
one resource state $\ketsm{\pi j / 2^{d'+1}}$  used as a catalyst,
and the parallel application of $a_{d'+1,k}/2 + 1$ unitaries $R\atsm{ \pi j / 2^{d'}}$.
We observe that $a_{d'+1,k}/2 + 1 = a_{d',k}$. 
Therefore, by induction hypothesis, the parallel application of unitaries $R\atsm{ \pi j / 2^{d'}}$ can be achieved using the resources described in the statement of the theorem. 
The total number of $CCZ$ gates applied is $b_{d',k}+ a_{d'+1,k}/2$ which is equal to $b_{d'+1,k}$ as required. 
We also added state $\ketsm{\pi j /2^{d'+1}}$ to the list of the catalysts used in the protocol.
Finally we note that $\lim_{k \rightarrow \infty} b_{d,k}/a_{d,k} = 1 - 2^{-(d-1)}$.
\end{proof}

\subsection{Overview of some definitions and results from Number Theory}
\label{app:number-theory}
The goal of this appendix is to review the results from algebraic number theory
needed to define and calculate function $v_2$ used in \sec{prob-half-bounds}.
We aim for a pedagogical and as self-contained as possible exposition
of the needed results. 
The readers with a solid knowledge of algebraic number theory 
should proceed to \rema{expert-summary}.

\subsubsection{Definition of \texorpdfstring{$v_2$}{v₂} and additivity}

Recall that we have defined $v_2$ in the beginning of \sec{prob-half-ccz-bounds} for rational numbers as following.
If $q$ is a non-zero rational number then $v_2(q)$ is equal to the power of $2$ in the factorization of $q$ into prime numbers. If $q$ is zero, then $v_2(q) = +\infty$. 
We need to extend $v_2$ to the real subset of the following family of sets:
\[
\mathcal{R}_d = \z\of{\exp(i\pi/2^d),1/2} = \set{ \frac{1}{2^k} \sum_{j=0}^{2^d-1} a_j \exp(i\pi j/2^d) : \text{ where } a_j, k \text{ are integers} }.
\]
and show that $v_2$ has the following two properties:
\begin{itemize}
  \item additivity, that is $v_2\at{x\cdot y} = v_2\at{x} + v_2 \at{y}$,
  \item $v_2(x+y) \ge \min\at{v_2(x),v_2(y)}$.
\end{itemize}
We also need to calculate values $v_2\at{ \cos(\pi k /2^d) }$, $v_2\at{\sin(\pi k /2^d) }$ for integers $k,d$.

We will define $v_2$ on a larger family of sets that includes $\mathcal{R}_d$ and its real subsets
\[
\q\atsm{\exp(i\pi/2^d)} = \set{ \sum_{j=0}^{2^d-1} a_j \exp(i\pi j/2^d) : \text{ where } a_j \text{ are rational numbers} }
\]
and shown that it has the required properties. 
Observe that sets $\q\atsm{\exp(i\pi/2^d)}$ are closed under addition and multiplication similarly to sets $\mathcal{R}_d$, so $v_2\at{x\cdot y}$ and  $v_2(x+y)$ are well-defined.

Our strategy for extending $v_2$ is the following. 
Later in this section we will define family of functions $N_d$ on sets $\q\atsm{\exp(i\pi/2^d)}$ with four properties: 
\begin{itemize}
  \item value of $N_d$ is always rational,
  \item $N_d$ is multiplicative, that is $N_d\at{x \cdot y} = N_d\at{x} \cdot N_d\at{y}$,
  \item $N_0$ is trivial, that is $N_0\at{x} = x$,
  \item $N_d\at{x}^2 = N_{d+1} \at{x}$.
\end{itemize}
Using functions $N_d$ and the definition of $v_2$ on the set of rational numbers we extend $v_2$ to the family of sets $\q\atsm{\exp(i\pi/2^d)}$ as:
\begin{equation}
v_2\at{x} = v_2\at{N_d\at{x}} / 2^d
\end{equation}

Above mentioned properties of $N_d$ make sure that $v_2$ is
additive and well-defined.
We see that the additive property of $v_2$ follows immediately 
from the multiplicative property of $N_d$.
The definition of $v_2$ on rational number does not change because $N_0\at{x} = x$.
Finally, the definition of $v_2$ is consistent. 
Function $v_2$ is defined on the family of nested sets:
\[
 \q \subset \q\at{i} \subset \q\atsm{\exp(i\pi/2^2)} \subset \ldots 
 \subset \q\atsm{\exp(i\pi/2^d)} \subset \q\atsm{\exp(i\pi/2^{d+1})} \subset \ldots.
\]
If $x$ belongs to set $\q\atsm{\exp(i\pi/2^d)}$, then $x$ also belongs to all the sets $\q\atsm{\exp(i\pi/2^{d+k})}$ for all integer $k$. For $v_2$ to be defined consistently, $v_2\at{N_d\at{x}} / 2^d$ must be equal to $v_2\at{N_{d+k}\at{x}} / 2^{d+k}$. This follows, from property $N_d\at{x}^2 = N_{d+1} \at{x}$ and the fact that for rational $q$ value $v_2\at{q^n} = n v_2\at{q}$.

The rest of this section is dedicated to defining function $N_d$ known as norm functions of $\q\atsm{\exp(i\pi/2^d)}$ and establishing their four required properties.
Let us start with $\q(i)$ and $N_1$.
The following three properties of complex conjugation are useful for our purpose:
\begin{itemize}
    \item if $a+ b i$ is in $\q(i)$, then $(a+b i)^\ast$ is in $\q(i)$,
    \item $\at{x\cdot y}^\ast = x^\ast \cdot y^\ast$,
    \item $x = x^\ast$ if and only if $x$ is in $\q$,
\end{itemize}
We define $N_1(x) = x\cdot x^\ast$. First property of complex conjugation ensures that $N_1$ is well-defined, the second one ensures multiplicativity of $N_1$, the third property ensures that $N_1$ is rational and that $N_1(x) = N_0(x)^2$ when $x$ is rational. To define $N_d$ for $d > 1$ we will need more maps similar to complex conjugation defined on sets $\q\atsm{\exp(i\pi/2^d)}$:
\begin{equation}
\sigma_k : \q\atsm{\exp(i\pi/2^d)} \rightarrow \q\atsm{\exp(i\pi/2^d)},\, \sigma_k\at{\sum_{j=0}^{2^d-1} a_j \exp(i\pi j/2^d)} = \sum_{j=0}^{2^d-1} a_j \exp(i\pi j k/2^d)
\end{equation}
Note that $\sigma_{-1}$ is the complex conjugation and $\sigma_1$ is the identity map.
The next proposition established some well-known properties of maps $\sigma_k$. 
We provide proof for completeness.
\begin{prop} \label{prop:sigmas}
For all odd $k$, maps $\sigma_k$ have the following properties
\begin{enumerate}
    \item for all $x,y$, 
    $\sigma_k\at{x \cdot y} = \sigma_k\at{x} \cdot \sigma_k\at{y}$ and
    $\sigma_k\at{x + y} = \sigma_k\at{x} + \sigma_k\at{y}$
    \item $x$ from $\q\atsm{\exp(i\pi/2^d)}$ is rational if and only if for all odd $k$ $\sigma_k\at{x} = x$
    \item for all $x$ from $\q\atsm{\exp(i\pi/2^d)}$ and for all odd $k$, $\sigma_{k+2\cdot2^d}\at{x} = \sigma_{k}\at{x}$
    \item or all $x$, $\sigma_k\at{\sigma_j\at{x}} = \sigma_{kj}\at{x}$
\end{enumerate}
\end{prop}
\begin{proof}
Property three follows from $2\pi$ periodicity of $\exp\at{i\phi}$.
Property four is a direct consequence of the definition of $\sigma_k$.
Additivity also follows directly from the definition.
The fact that for all rational $a$, $\sigma\at{a} = a$ also follows from definition.
For rational $a$, $\sigma_k\at{a\cdot x} = \sigma_k\at{a}\cdot\sigma_k\at{x}$ again by definition of $\sigma_k$.

To establish multiplicativity it is sufficient to check that for all $j$ and $j'$ and rational $a,b$: 
\[
\sigma_k\at{ a \exp(i\pi j/2^d) \cdot b \exp(i\pi j'/2^d) } 
=
\sigma_k\at{ a \exp(i\pi j/2^d)} \cdot \sigma_k\at{ b \exp(i\pi j'/2^d) },
\]
and then use additivity. 

It remains to show that $\sigma_k\at{x} = x$ for all odd $k$ implies that $x$ is rational. 
Consider 
\[
x = \sum_{j=0}^{2^d-1} a_j \exp(i\pi j k/2^d) \in \q\atsm{\exp(i\pi/2^d)}
\]
and let us see what are the implications of the fact $\sigma_{2^d+1}\at{x} = x$.
Observe, that all $a_j$ for odd $j = 2j' + 1$ must be zero. 
Indeed, $\sigma_{2^d+1}\at{\exp(i\pi (2j'+1)/2^d)} = - \exp(i\pi (2j'+1)/2^d)$ and therefore $a_j = -a_j$. 
We have shown, that $x$ belongs to $\q\atsm{\exp(i\pi/2^{d-1})}$. 
Repeatedely applying above argument we conclude that $x$ must be rational.
\end{proof}
Now we can define $N_d$ as following
\begin{equation}
    N_d\at{x} = \prod_{k=0}^{2^d-1} \sigma_{2k+1}\at{x}
\end{equation}
and prove that $N_d$ has required properties:
\begin{prop} \label{prop:norm-prop}
Maps $N_d$ have the the following properties:
\begin{itemize}
  \item value of $N_d$ is always rational,
  \item $N_d$ is multiplicative, that is $N_d\at{x \cdot y} = N_d\at{x} \cdot N_d\at{y}$,
  \item $N_0$ is trivial, that is $N_0\at{x} = x$,
  \item $N_d\at{x}^2 = N_{d+1} \at{x}$.
\end{itemize}
\end{prop}
\begin{proof}
Multiplicativity of $N_d$ follows from the multiplicativity of $\sigma_k$.
Let us check that $\sigma_{j}\at{N_d\at{x}} = N_d\at{x}$ for all odd $j$ to establish that $N_d$ is rational using the second property of $\sigma_{j}$ established in \propos{sigmas}:
\[
\sigma_{j}\at{N_d\at{x}} 
= 
\prod_{k \in \set{1,3,\ldots,2\cdot2^d -1 }} \sigma_j\at{\sigma_k\at{x}} 
= 
\prod_{k \in \set{1,3,\ldots,2\cdot2^d -1 }} \at{\sigma_{kj\,\mathrm{mod} (2\cdot 2^d)}\at{x}}
\]
We used properties three and four from \propos{sigmas} to establish the last equality.
Above expression is equal to $N_d\at{x}$ because map $k \mapsto kj\,\mathrm{mod} (2\cdot2^d)$ maps set 
\[
\set{1,3,\ldots,2\cdot2^d -1 }\]
to itself when $j$ is odd.

Map $N_0$ is equal to $\sigma_1$ and therefore trivial. Consider now expression for $N_{d+1}$ for $x$ from $\q\atsm{\exp(i\pi/2^d)}$: 
\[
N_{d+1}\at{x} =
\prod_{k=0}^{2^{d+1}-1} \sigma_{2k+1 }\at{x} = 
\prod_{k=0}^{2^{d+1}-1} \sigma_{2k+1 \, \mathrm{mod} (2\cdot 2^d) }\at{x}
\]
Note that function $k \mapsto 2k+1\,\mathrm{mod} (2\cdot2^d)$ takes the same value for $k$ and $k+2^d$ and therefore the expression above equals to $N_d\at{x}^2$.
\end{proof}

\subsubsection{Certain values of \texorpdfstring{$v_2$}{v₂}}

To compute many useful values of $v_2\at{x}$ it is sufficient to know values of $N_d\at{x}$ given by the following proposition:
\begin{prop} \label{prop:norm-two}
For all $j$, $N_d\at{\exp\at{i\pi j/2^d}} = 1$ and $N_d\at{1 - \exp\at{i\pi (2j+1)/2^d}} = 2$.
\end{prop}
\begin{proof}
First note that
\[
N_d\at{\exp\at{i\pi j/2^d}} =
\exp\at{i\pi j \sum_{k=0}^{2^d-1}(2k+1)/2^d } = 
\exp\at{i\pi j 4^d/2^d } = 1
\]
Second recall that polynomial $\Phi_d\at{x} = x^{2^d} + 1$ can be written as 
\[
\Phi_d\at{x} = \prod_{k=0}^{2^d-1}\at{x - \exp\at{i(2k+1)\pi/2^d} } 
\]
because each of $\exp\at{i(2k+1)\pi/2^d}$ for $k = 0 ,\ldots, 2^{d-1}$ is a root of $\Phi_d\at{x}$.
Expression for $N_d\at{1 - \exp\at{i\pi (2j+1)/2^d}}$ coincides with the expression for $\Phi_d\at{1}=2$.
\end{proof}
Using above proposition and properties of $v_2$ we find that for odd $k$
\begin{align}
v_2\at{ \sin(\pi k /2^d) } 
= 
v_2\at{ 2\sin(\pi k /2^d) } - 1 
=
v_2\at{ \exp(i \pi k /2^d) - \exp(- i\pi k /2^d) } - 1 
= \\
= v_2\at{ 1 - \exp(\pi k /2^{d-1}) } - 1 = 1/2^{d-1} - 1
\end{align}
Similar calculation shows that $v_2\at{ \cos(\pi k /2^d) } = 1/2^{d-1} - 1$. 

In \propos{norm-two} we saw that $N_d$ takes integer values for two elements of a ring of cyclotomic integers
\[
\z\of{\exp(i\pi/2^d)} = \set{ \sum_{j=0}^{2^d-1} a_j \exp(i\pi j/2^d) : \text{ where } a_j \text{ are integers} }
\]
This is true more generally
\begin{prop} \label{prop:int-norm}
Let $x$ be and element of $\z\of{\exp(i\pi/2^d)}$, then $N_d\at{x}$ is an integer. 
If $x'$ is an element of
\[
\mathcal{R}_d = \z\of{\exp(i\pi/2^d),1/2} = \set{ \frac{1}{2^k} \sum_{j=0}^{2^d-1} a_j \exp(i\pi j/2^d) : \text{ where } a_j, k \text{ are integers} },
\]
then $N_d\at{x'} = a/2^K$ for integers $a,K$.
\end{prop}
\begin{proof}
Consider the case when $x$ is from $\z\of{\exp(i\pi/2^d)}$. 
The result follows from a proof technique similar to the proof of rationality of $N_d\at{x}$, when $x$ is from $\q\at{\exp(i\pi/2^d)}$ in \propos{norm-prop}.

The second case follows from representing $x' = x / 2^k$ for some $x$ from $\z\of{\exp(i\pi/2^d)}$ and some integer $k$. 
Next we notice that by properties of $N_d$ from \propos{norm-prop} of $N_d\at{x'} = N_d\at{x} / 2^K$ for $K = 2^d k$.
\end{proof}

The following proposition gives a necessary condition for an element of $\mathcal{R}_d$ to be equal to $\pm 1$, in terms of $v_2$. This is the key to the proof of the fact that dyadic monotone $\mu_2$ is positive and is zero if and only if the corresponding state is the stabilizer state. 

\begin{prop} \label{prop:non-positivity}
Let $x$ be an element of

such that for all odd $k$, $\abs{\sigma_k\atsm{x}} \le 1$, then $v_2\at{x} \le 0$ and the equality is achieved if and only if $x = \pm 1$.
\end{prop}
\begin{proof}
Let us first show that $v_2\at{x}$ is non-positive.
Condition $\abs{\sigma_k\atsm{x}} \le 1$ implies that $N_d\at{x} \le 1$. 
Because $x$ is an element of $\mathcal{R}_d$ it can be written as $z/2^k$ for $z$ from 

Following the proof~\propos{norm-prop} of rationality of $N_d$, one can show that $N_d\atsm{z} = n$ is an integer and therefore $N_d\at{x} = n / 2^K $ and has absolute value less or equal to $1$. 
For any number of the form $n / 2^K$ with absolute value less or equal to $1$ value of $v_2$ is non-positive and $v_2$ is zero if an only if $n/2^K = \pm 1$. 
We see that $v_2\at{x}$ is non-positive and is zero if and only if $N_d\at{x} = \pm 1$.

Let us now show that $v_2\at{x}$ equal zero implies that $x = \pm 1$.
We have already shown that $N_d\at{x} = \pm 1$.
We also have condition that $\abs{\sigma_k\atsm{x}} \le 1$ for all $k$. 
The only way $N_d\at{x}$ can be equal to $\pm 1$ is if  $\sigma_1\atsm{x} = x = 1$ which conclude the proof.
\end{proof}

\subsubsection{Inequality \texorpdfstring{$v_2(x+y) \ge \min(v_2(x),v_2(y))$}{v₂(x+y) ≥ min(v₂(x),v₂(y))}}

Recall, that in \sec{prob-half-ccz-bounds} the inequality $v_2(x+y) \ge \min(v_2(x),v_2(y))$ for rational $x,y$ was first established for integer $x$ and $y$ and then extended to rationals by using additivity of $v_2$. 
We will follow the same strategy in the general case and introduce cyclotomic integers:
\[
\z\of{\exp(i\pi/2^d)} = \set{ \sum_{j=0}^{2^d-1} a_j \exp(i\pi j/2^d) : \text{ where } a_j \text{ are integers} }
\]
Indeed, for arbitrary $x,y$ from $\q\at{\exp(i\pi/2^d)}$ there always exist an integer $C$, such that $x'= Cx, y' = Cy$ are both cyclotomic integers from $\z\of{\exp(i\pi/2^d)}$. 
The general inequality easily follows from the inequality for cyclotimic integers:
\begin{align*}
v_2\at{x+y} 
= v_2\at{x'+y'} + v_2\atsm{1/C} 
\ge
\min\at{v_2\at{x'},v_2\at{y'}} + v_2\atsm{1/C} 
= 
\min\at{v_2\at{x'/C},v_2\at{y'/C}}
\end{align*}

To complete the proof of the inequality we need the proposition below.
Once this proposition is established, we can follow the same proof idea as for the rational version of the inequality in \sec{prob-half-ccz-bounds} with $2$ replaced by $1-\exp(i\pi/2^d)$.
\begin{prop} 
Let $x$ be an element of $\z\of{\exp(i\pi/2^d)}$, then $v_2\at{x} \ge 0$. 
Moreover, for $k = 2^d v_2\at{x}$, $x$ can be written as $x' (1-\exp(i\pi/2^d))^k$ for $x'$ from $\z\of{\exp(i\pi/2^d)}$ such that $v_2\at{x'} = 0$.
\end{prop}
\begin{proof}
Let us denote $\alpha_d = (1-\exp(i\pi/2^d))$ and choose $k$ to be the biggest power of $\alpha_d$ that divides 
\footnote{For $x,y$ from $\z\of{\exp(i\pi/2^d)}$, we say that $x$ divides $y$ if
there exist $r$ from $\z\of{\exp(i\pi/2^d)}$ such that $y = rx$. }
 $x$. 
We can write $x = \alpha_d^k x'$ such that $x'$ is from $\z\of{\exp(i\pi/2^d)}$ such that  $\alpha_d$ does not divide $x'$. 
It remains to show that $2$ does not divide $N_d\at{x'}$, because this will establish that $v_2\at{x'}=0$ and $v_2\at{x} = k v_a\at{\alpha_d} = k/2^d$.

Let us show that $2$ does not divide $N_d\at{x'}$.
Recall that $N_d\at{x'}$ is an integer for any $x'$ from  $\z\of{\exp(i\pi/2^d)}$ according to \propos{int-norm}.
Suppose now that $2$ divides $N_d\at{x'}$. 
This implies that $\alpha_d$ divides $N_d\at{x'}$.
Because $\alpha_d$ is prime according to \propos{prime}, $\alpha_d$ must divide $\sigma_{2k+1}\at{x'}$ for some $k$. 
There exist $j$ such that $(2j+1)(2k+1)\mathrm{mod}\,2^{d+1} = 1$ and $\sigma_{2j+1}\atsm{\sigma_{2k+1}\atsm{x}} = x$ for such j.
Therefore $\sigma_{2j+1}\at{\alpha_d}$ divides $x'$.
However, $\alpha_d$ divides $\sigma_{2j+1}\at{\alpha_d}$ according to \propos{divisibility} and therefore $\alpha_d$ divides $x'$ which is a contradiction.
\end{proof}

\begin{prop} \label{prop:prime}
$\alpha_d = (1-\exp(i\pi/2^d))$ is a prime element of $\z\of{\exp(i\pi/2^d)}$.
That is, for $x$ or $y$ from $\z\of{\exp(i\pi/2^d)}$, if $\alpha_d$ divides $xy$ then $\alpha_d$ divides $x$ or $y$.
\end{prop}
\begin{proof}
The result follows from the fact that $N_d\at{\alpha_d} = 2$ is a prime number and the fact that every element of ring of integers of a number field with a prime norm is a prime element of the ring of integers.
\end{proof}

\begin{prop} \label{prop:divisibility}
Number $u_j = \at{1 - \exp\at{i \pi(2j-1)/2^d}}/\at{1 - \exp\at{i \pi/2^d}}$ is a unit in $\z\of{\exp(i\pi/2^d)}$. In other words, $u_j$ and $u_j^{-1}$ are both in $\z\of{\exp(i\pi/2^d)}$. 
\end{prop}
\begin{proof}
First, note that $u_j$ is an element of $\z\of{\exp(i\pi/2^d)}$ because the polynomial $1 - x^{2j-1}$ is divisible by $(1-x)$:
\[
(1 - x^{2j-1})/(1-x) = \sum_{k=0}^{2j-2} x^k \implies u_j = \sum_{k=0}^{2j-2} \exp\at{i k \pi/2^d},
\]
To show that the inverse of $u_j$ is an element of $\z\of{\exp(i\pi/2^d)}$ we first find an integer $j'$ such that $j'(2j-1) \equiv 1 ~\mathrm{mod}\,2\cdot2^d$ by using the extended Euclidean algorithm and the fact that $(2j-1)$ and $2^{d+1}$ are coprime, then the inverse is
\[
u^{-1}_j = \at{1 - \exp\at{i \pi(2j-1)j'\pi/2^d}}/\at{1 - \exp\at{i \pi(2j-1)/2^d}}.
\] 
Again using that the polynomial $1-x^{j'}$ divisible by the polynomial $1-x$, we conclude that $u^{-1}_j$ is an element of $\z\of{\exp(i\pi/2^d)}$.
\end{proof}

\begin{rem} \label{rem:expert-summary}
All the ring of integers of number fields $\q\atsm{\exp(i\pi/2^d)}$ have unique ramified prime ideal $\mathfrak{p}_d$ with norm $2$. Function $v_2$ is a $\mathfrak{p}$-adic valuation divided by $2^d$. The re-normalization makes sure that $v_2$ is defined consistently for the whole family of nested fields 
\[
 \q \subset \q\at{i} \subset \q\atsm{\exp(i\pi/2^2)} \subset \ldots 
 \subset \q\atsm{\exp(i\pi/2^d)} \subset \q\atsm{\exp(i\pi/2^{d+1})} \subset \ldots.
\]
All the properties of $v_2$ follow from the properties of $\mathfrak{p}$-adic valuations and the fact that $\mathfrak{p}_d = \mathfrak{p}^2_{d+1}$.
\end{rem}

\subsection{Some properties of dyadic monotone \texorpdfstring{$\mu_2$}{μ₂}}
\label{app:properties-of-dyadic-mtone}
In this appendix we prove properties of the dyadic monotone that use slightly more advanced techniques from number theory introduced in \app{number-theory}.

\begin{prop} \label{prop:ext-dyadic-non-negativity}
Let $\ket{\psi}$ be a state with entries in $\mathcal{R}_d$, then $\mu_2\ket{\psi} \ge 0$ and the equality achieved if and only if $\ket{\psi}$ is a stabilizer state. 
In addition, for every Pauli operator $P$, if the expectation $\bra{\psi} P \ket{\psi} \ne 0$,
 then $v_2\at{\bra{\psi} P \ket{\psi}} \le 0$.
\end{prop}
\begin{proof}
Consider Pauli $P$ expectation $\alpha = \bra{\psi} P \ket{\psi}$. 
Because $P$ has eigenvalues $\pm 1$, $|\bra{\psi} P \ket{\psi}| \le 1$. 
Consider now $\sigma_k\at{\alpha}$. 
Because $\sigma_k$ respects addition, multiplication and commutes with complex conjugation according to \propos{sigmas}, $\alpha_k$ can be written as expectation $\bra{\psi_k} P_k \ket{\psi_k}$ where $\ket{\psi_k}$ is the state obtained from $\ket{\psi}$ by applying $\sigma_k$ element-wise and $P_k$ some other Pauli operator obtained from $P$ by also applying $\sigma_k$ element-wise. 
We conclude that $|\alpha_k| \le 1$. 
Now using \propos{non-positivity} we conclude that $v_2\at{\alpha} \le 0$ and the equality is achieved if and only if $\alpha = \pm 1$.
This implies that $\mu_2$ is always non-negative and equality is achieved if and only if all non-zero Pauli expectations of $\ket{\psi}$ are $\pm 1$. 
This implies that $\ket{\psi}$ is a stabilizer state, similarly to the proof of \propos{dyadic-non-negativity}.
\end{proof}

Next we show that $\mu_2$ is non-increasing for a slightly more general class of measurement than Pauli measurements with outcome probabilities one half.

\begin{prop} \label{prop:measure-extended}
Let $\ket{\psi}$ be a state with entries in $\mathcal{R}_d$. 
Let $P$ be a multi-qubit Pauli observable and let $p = \bra{\psi} I + P \ket{\psi} /2 > 0$ be a probability of measuring $+1$ eigenvalue of $P$. 
Suppose there exist global phase $e^{i \phi}$ such that $\ket{\psi_+} = e^{i \phi} \frac{I+P}{\sqrt{p}} \ket{\psi}$ is the state with entries in $\mathcal{R}_d$, then $\mu_2 \ket{\psi_+} \le \mu_2\ket{\psi}$.
\end{prop}
\begin{proof}
We assume that $ \bra{\psi} P \ket{\psi} \ne 0$, because equality to zero case corresponds to $p=1/2$ and covered by \propos{measure-half}.

Consider Pauli matrix $Q$ and corresponding expectation $\alpha = \bra{\psi_+} Q \ket{\psi_+}$.
If $P$ and $Q$ anti-commute, the expectation $\alpha$ is zero, because 
$(I+P)Q(I+P) = (I+P)(I-P)Q = 0$.
It remains to consider the case when $P$ and $Q$ commute. 
In this case the expectation is 
\[
\alpha = \frac{ \bra{\psi} QP + P \ket{\psi} }{ 1 + \bra{\psi} P \ket{\psi} }
\]
Using multiplicative property of $v_2$ and inequality $v_2\at{a+b} \ge \min\atsm{v_2\atsm{a}, v_2\atsm{b}}$:
\[
v_2\at{\alpha} \ge \min\atsm{v_2\at{\bra{\psi} QP \ket{\psi}}, v_2\at{\bra{\psi} P \ket{\psi}}} -  v_2\atsm{1 + \bra{\psi} P \ket{\psi}}
\]
Recall, that by definition of $\mu_2$, $v_2\at{\bra{\psi} P' \ket{\psi}} \ge - \mu_2 \ket{\psi}$ for any Pauli $P'$ including $P$ and $PQ$.
It remains to show that $v_2\atsm{1 + \bra{\psi} P \ket{\psi}}$ is non-positive.
\[
v_2\at{\bra{\psi} P \ket{\psi}} 
=
v_2\at{\bra{\psi} P \ket{\psi} + 1 - 1 }
\ge
\min\at{v_2\at{1 + \bra{\psi} P \ket{\psi}}, 0}
\]
We have shown above in \propos{ext-dyadic-non-negativity} that $v_2\at{\bra{\psi} P \ket{\psi}}$ is non-positive when the expectation $\bra{\psi} P \ket{\psi}$ is non-zero.
We see that $v_2\at{1 + \bra{\psi} P \ket{\psi}} = \bra{\psi} P \ket{\psi} \le 0$.

We have shown that for arbitrary Pauli matrix $Q$, $-v_2\at{\bra{\psi_+} Q \ket{\psi_+}} \le \mu_2 \ket{\psi}$.
Inequality  $\mu_2 \ket{\psi_+} \le \mu_2\ket{\psi}$ follows from the definition of $\mu_2$.
\end{proof}

One might wonder if above result holds for two or more post-selected Pauli measurements. 
Below we provide an example showing that post-selecting on two commuting Pauli measurements can increase value of $\mu_2$: 
\[
\ket{\Psi} = (1, -i, -5i - 2, -i, -2i + 1, -i - 2, -i, -i + 2, 1, i, i + 2, i, 1, i, i + 2, i)/8
\]
By direct computation one can check that $\mu_2 \ket{\Psi} = 3$. Post-selecting on $+1$ outcome for observables $Z_1, Z_2$ results in the state $(1, -i, -5i - 2)/4\sqrt{2}$ with the value of $\mu_2$ equal to $4$.

\begin{table*}[ht]
\hfill
\begin{subfigure}[b]{0.45\textwidth}
\centering
{\renewcommand{\arraystretch}{1.25}
\begin{tabular}{c|c}
$\ket{\psi}$     & $\mu_2(\ket{\psi})$ \\ 
\hline                                 
$\ket{\sqrt{T}}$ & 3/4                 \\ 
$\ket{T}$        & 1/2                 \\
$\ket{CS}$       & 1                   \\
$\ket{CCS}$      & 2                   \\
$\ket{C^3S}$     & 3                   \\
$\ket{CCZ}$      & 1
\end{tabular}
}
\end{subfigure}
\hfill
\begin{subfigure}[b]{0.45\textwidth}
\centering
{\renewcommand{\arraystretch}{1.25}
\begin{tabular}{c|c}
$\ket{\psi}$          & $\mu_2(\ket{\psi})$ \\ 
\hline 
$\ket{C^3Z}$          & 2                   \\ 
$\ket{C^4Z}$          & 3                   \\ 
$\ket{CCZ_{123,145}}$ & 5                   \\ 
$\ket{W_3}$           & 1                   \\
$\ket{W_4}$           & 2                   \\ 
$\ket{W_5}$           & 2                   
\end{tabular}
}
\end{subfigure}
\hfill
	
\caption{Exact expressions for the dyadic monotone of the states used in \tab{T_to_Targ} and \tab{CCZ_to_Targ}.}
\label{tab:dyadic-monotone-values}
\end{table*}

\subsubsection{Connections between the dyadic monotone and maximum denominator exponent}

Here we show that the maximum denominator exponent of a unitary, which was used in \cite{Forest2015} for the exact synthesis and canonical form of Clifford-cyclotomic gate-sets, is proportional to the dyadic monotone of the corresponding Choi state. 
Recall, that the Choi state's density matrix can be written as 
\[
\frac{1}{4^n}\sum_{P \in {I,X,Y,Z}^{\otimes n} } U P U^\dagger \otimes P.
\]
The Pauli spectrum of the Choi state consists of values
\[
\abs{\mathrm{Tr}\at{ (Q \otimes Q') \sum_{P \in \set{I,X,Y,Z}^{\otimes n} } U P U^\dagger \otimes P }}/4^n,
\] 
for all possible Pauli matrices $Q,Q'$.
Taking into account that $\mathrm{Tr}(PQ) = \delta_{P,Q} \cdot 2^n$, we see that the Pauli spectrum of the Choi state is the set:
\[
\set{ \abs{\mathrm{Tr}\at{Q U P U^\dagger}}/2^n : P,Q \in \set{I,X,Y,Z}^{\otimes n} },
\]
which is exactly the set of all entries of $U$ in the channel representation~\cite{Gosset2014} used to compute the maximum denominator exponent of $U$. Finally, maximum denominator exponent of the entry is proportional to minus the normalized p-adic valuation $v_2$. 

\subsection{General lower bounds for approximate unitary synthesis}
\label{app:general-approx-bounds}
The goal of this section is to generalize \propos{commuting-pauli-prob-on-t-states},
from the tensor power of $\ket{T}$ state to arbitrary 
resource state for which dyadic monotone is defined. 
This leads to a lower bound on single qubit unitary approximation which generalizes \lemm{cs-lower-bound}, \lemm{ccz-lower-bound}, \lemm{t-lower-bound}.
Finally we show a version of \theo{approximation-lower-bound} involving $T$ and $\sqrt{T}$ states.
\begin{lem} \label{lem:commuting-pauli-prob-bound}
Let $d \ge 1$ and $\ket{\Psi}$ be a state with entries in $\mathcal{R}_d = \z\of{\exp(i\pi/2^d),1/2}$ when written in computational basis.
Let $\set{P_1,\ldots,P_{m}}$ be independent commuting Pauli operators and
let the probability of joint measurement of $+1$ eignevalue of $\set{P_1,\ldots,P_{m}}$ on $\ket{\Psi}$ be
\begin{equation} \label{eq:Psi-expectations}
 p = \frac{1}{2^{m}}\sum_{P \in \ip{P_1,\ldots,P_{m} }} \bra{\Psi} P \ket{\Psi}.
\end{equation}
If the value of $p$ is non-zero, then  $\log_2 p \ge -2^{d-1}( m + \mu_2 \ket{\Psi})$.
\end{lem}
\begin{proof}
The proof consists of two steps. First we lower bound $p$ by $\sqrt{N_d\at{p}}$. 
Second, we observe that $N_d\at{p}$ is given by ratio $a/2^K$ for some odd integer $a$ and non-negative integer $K$~(\propos{int-norm}) and upper bound $K$ in terms of $\mu_2 \ket{\Psi}$.

Let us first recall the expression for $N_d$:
\[
    N_d\at{p} = \prod_{k=0}^{2^d-1} \sigma_{2k+1}\at{p}
\]
Using properties of map $\sigma_k$ from \propos{sigmas}, expression for $N_d$ can be rewritten as: 
\[
N_d\at{p} = 
\prod_{k=0}^{2^{d-1}-1} \sigma_{2k+1}\at{p}\sigma_{-(2k+1)}\at{p}=
\at{\prod_{k=0}^{2^{d-1}-1} \sigma_{2k+1}\at{p}}^2
\]
Above we used qualities $\sigma_{-1}\at{x} = x^\ast$, $\sigma_{kj}\at{x} = \sigma_k\at{\sigma_j\at{x}}$ with $j=-1$, and took into account that $p$ is a real number.
To lower bound $p$ in terms of $N_d\at{p}$ it remains to notice 
that $\sigma_k\at{p}$ is equal to: 
\[
\frac{1}{2^{m}}\sum_{P \in \ip{P'_1,\ldots,P'_{m} }} \bra{\Psi_k} P \ket{\Psi_k}.
\]
Where $P'_j$ are some Pauli matrices obtained from $P_j$ by element-wise application of $\sigma_k$ and $\ket{\Psi_k}$ is a state obtained from $\ket{\Psi}$ by element-wise application of $\sigma_k$. Because $\sigma_k\at{p}$ is probability of some measurement it must be less than 1. 
We see that $p \ge \sqrt{N_d\at{p}}$.

Recall that we have defined $v_2\at{p}$ as 
\[
v_2\at{N_d\at{p}}/2^d.
\]
Therefore, $K = - 2^d v_2\at{p}$. 
Now using inequality $v_2\at{a+b} \ge \min(v_2\at{a},v_2\at{b})$, we have:
\[
v_2\at{p} \ge - m + \min_{P \in \ip{P'_1,\ldots,P'_{m} } } v_2\at{\bra{\Psi_k} P \ket{\Psi_k}} \ge - m - \mu_2 \ket{\Psi}.
\]
We conclude that $K \le 2^d \mu_2 \at{\ket{\Psi}}$ and therefore $p \ge \sqrt{2}^{-2^d (m + \mu_2 \ket{\Psi})}$.
\end{proof}

Next we follow the proof technique of \lemm{cs-lower-bound} and use the proposition above and establish the following result: 
\begin{lem} \label{lem:mu2-approx-bound}
Let $d \ge 1$ and $\ket{\Psi}$ be a state with entries in $\mathcal{R}_d = \z\of{\exp(i\pi/2^d),1/2}$ when written in computational basis.
Suppose that qubit state $\ket{\psi}$ is approximated
to within trace distance $\varepsilon$ using stabilizer operations and has post-selection with input $\ket{\Psi}$.
Inequalities $\varepsilon < 1/8$ and $\varepsilon < \abs{\ip{\psi | 0}}^2 < 3 \varepsilon$ imply
$\mu_2\ket{\Psi} + \nu \ket{\Psi} \ge \frac{1}{2^{d-1}} \log_2\at{1/\varepsilon} - \frac{1}{2^{d-2}}$.
\end{lem}
\begin{proof}
First note that we can assume without loss of generality that $\nu \ket{\Psi}$ is equal to the number of of qubits on which $\ket{\Psi}$ is defined.
The rest of the proof is similar to \lemm{cs-lower-bound}.
\end{proof}

For example, above result implies a lower bound when approximating using $N_T$ copies of $\ketsm{T}$ state, and $N_{\sqrt{T}}$ copies of $\ketsm{\sqrt{T}}$ and $\ketsm{\sqrt{T}^3}$ states: 
\[
  N_{\sqrt{T}} + \frac{6}{7} N_T \ge \frac{1}{7} \log_2\atsm{1/\varepsilon} - \frac{1}{14}.
\]
Above inequality leads to the following generalization of 
\theo{approximation-lower-bound}.
\begin{thm} \label{thm:root-t-lower-bound}
Consider a protocol that uses $\mathcal{N}_{\sqrt{T}}(U,\varepsilon)$ copies of $\ketsm{\sqrt{T}}$ and $\ketsm{\sqrt{T}^3}$ states, $\mathcal{N}_{T}(U,\varepsilon)$ copies of $\ketsm{T}$ state and stabilizer operations to
approximate a one-qubit unitary $U$ to within precision $\varepsilon$ (measured by the diamond norm). 
For any positive $C > 1$ and $\varepsilon < 1/(2^8 C)$ there exists a unitary $U$ such that the following inequality must hold
\begin{align*}
\mathcal{N}_{{\sqrt{T}}}(U,\varepsilon) + \frac{6}{7} \mathcal{N}_{{T}}(U,\varepsilon)
& \ge \frac{1}{14} \log_2\at{1/\varepsilon} - \frac{1}{14} \log_2\at{C} - \frac{3}{14} 
\end{align*}
with probability at least $(C-1)/C$.
In particular, this is the case for all unitaries $U$ such that
$2\sqrt{C \varepsilon} \le \abs{\bra{0}U\ket{1}}^2 \le 6 \sqrt{C \varepsilon}$.
\end{thm}

\end{document}